\pdfoutput=1

\documentclass[final]{IEEEtran}

\usepackage[utf8]{inputenc} 
\usepackage[T1]{fontenc}    
\usepackage{hyperref}       

\usepackage{url}
\usepackage{caption}
\usepackage{bm,bbm}
\usepackage{amsthm,amssymb,amsmath,mathtools,stmaryrd}
\usepackage{microtype}
\usepackage{xcolor,graphicx,tabularx}
\usepackage{epstopdf}
\usepackage{comment}
\usepackage{cite}
\usepackage{tikz,pgf}
\usetikzlibrary{spy}
\usetikzlibrary{positioning}
 
\newcommand{\be}{\begin{eqnarray*}}
\newcommand{\ee}{\end{eqnarray*}}
\newcommand{\ben}{\begin{eqnarray}}
\newcommand{\een}{\end{eqnarray}}
\newcommand{\ba}{\left(\begin{array} }
\newcommand{\ea}{\end{array}\right)}
\newcommand{\ban}{\begin{array} }

\newcommand{\ean}{\end{array}}

\theoremstyle{plain}
\newtheorem{thm}{Theorem}
\newtheorem{lem}{Lemma}
\newtheorem{def1}{Definition}

\newcommand{\blds}{\boldsymbol}

\newcommand{\hi}{{_\textup{hi}}}
\newcommand{\lo}{{_\textup{lo}}}
\newcommand{\tD}{{\textup{D}}}

\newcommand{\xhi}{\mf{x}{\hi}}
\newcommand{\ylo}{\mf{y}{\lo}}

\newcommand{\mf}{\mathbf}

\usepackage{enumerate}
\usepackage{color}
\usepackage{enumerate}
\usepackage{color}
\usepackage{algorithm}
\usepackage[noend]{algpseudocode}
\usepackage[font=small,skip=0pt]{caption}
\let\hat\widehat
\let\tilde\widetilde

\newcommand{\bl}[1]{\textcolor[rgb]{0,0,0}{#1}}

\title{Super-resolution with Binary Priors: Theory and Algorithms}
\author{Pulak Sarangi, Ryoma Hattori, Takaki Komiyama and Piya Pal}
\begin{document}

\maketitle
\begin{abstract}
The problem of super-resolution is concerned with the reconstruction of temporally/spatially localized events (or spikes) from samples of their convolution with a low-pass filter. 
Distinct from prior works which exploit sparsity in appropriate domains in order to solve the resulting ill-posed problem, this paper explores the role of binary priors in super-resolution, where the spike (or source) amplitudes are assumed to be binary-valued. Our study is inspired by the problem of neural spike deconvolution, but also applies to other applications such as symbol detection in hybrid millimeter wave communication systems. 
This paper makes several theoretical and algorithmic contributions to enable binary super-resolution with very few measurements. 
Our results show that binary constraints offer much stronger identifiability guarantees than sparsity, allowing us to operate in “extreme compression" regimes, where the number of measurements can be significantly smaller than the sparsity level of the spikes. To ensure exact recovery in this "extreme compression" regime, it becomes necessary to design algorithms that exactly enforce binary constraints without relaxation. 
In order to overcome the ensuing computational challenges, we consider a first order auto-regressive filter (which appears in neural spike deconvolution), and exploit its special structure. 
This results in a novel formulation of the super-resolution binary spike recovery in terms of binary search in one dimension. 
We perform numerical experiments that validate our theory and also show the benefits of binary constraints in neural spike deconvolution from real calcium imaging datasets.
\end{abstract}
\begin{IEEEkeywords} 
Binary compressed sensing, super-resolution, spike deconvolution, sparsity, binary search, beta-expansions
\end{IEEEkeywords}
\vspace{-0.4cm}
\section{Introduction}
The problem of recovering localized events (spikes) from their convolution with a blurring kernel, arises in a wide range of scientific and engineering applications such as fluorescence microscopy \cite{small2014fluorophore}, neural spike deconvolution \cite{brette2012handbook,vogelstein2009spike,deneux2016accurate}, 
hybrid millimeter wave (mmWave) communication \cite{yang2015fifty}, to name a few. 
Consider $K$ temporal spikes, which can be represented as:
\[
   \smash{x{\hi}(t)=\sum_{k=1}^{K} c_k \delta(t-n_k T_{\hi})} \]
Here, the high-rate spikes are supported on a fine temporal grid with spacing $T_{\hi}$, $n_k\in \mathbb{Z}$ is an integer corresponding to the time index of the $k^{\text{th}}$ spike and $c_k$ denotes its amplitude. The convolution of spikes with a filter $h(t)$ is typically uniformly (down)sampled at a (low) rate $T_{\lo}=\tD T_{\hi}$ ($\tD> 1$), yielding measurements: 
\begin{equation}
     \smash{y[n]=x{\hi}(t)\star h(t)\vert_{t=nT{\lo}}
     =\sum_{k=1}^{K} c_k h(nT_{\lo}-n_kT_{\hi})} \label{eqn:m1}
\end{equation}
The goal of super-resolution is to recover the spike locations $n_k$ and amplitudes $c_k$, $k=1,2,\cdots,K$ from a limited number ($M$) of low-rate samples $\{y[n] \}_{n=0}^{M-1}$.
The problem is typically ill-posed due to systematic attenuation of high-frequency contents of the spikes by the low-pass filter $h(t)$. In order to make the problem well-posed, it becomes necessary to exploit priors such as sparsity \cite{donoho1992superresolution,candes2014towards,li2020super,batenkov2021super} and/or non-negativity \cite{schiebinger2017superresolution,bendory2017robust}. In recent times, there has been a substantial progress towards developing efficient algorithms for provably solving the super-resolution problem  \cite{li2020super,batenkov2021super,candes2014towards,liao2016music,qiao2019guaranteed,qiao2019non,qiao2020super,shahsavari2021fundamental,qiao2018modulus,chi2020harnessing,bhaskar2013atomic,schiebinger2017superresolution,bendory2017robust}.


In this paper, we investigate the problem of \emph{binary super-resolution}, where the amplitudes of the spikes are known apriori to be $c_k=A$, but their number ($K$) and locations ($n_k$) are unknown. 
Motivated by the problem of neural spike deconvolution in two-photon calcium imaging \cite{grewe2010high,brette2012handbook}, we will focus on a blurring kernel that can be represented as a stable first order auto-regressive (AR(1)) filter. 
Each neural spike results in a sharp rise in Ca$^{2+}$ concentration followed by a slow exponential decay (modeled as the impulse response of an AR(1) filter), which results in an overlap of the responses from nearby spiking events, leading to poor temporal resolution  \cite{pnevmatikakis2016simultaneous,brette2012handbook}. 

\vspace{-0.4cm}
\subsection{Related Works}
Early works on super-resolution date back to algebraic/subspace-based techniques such as Prony's method, MUSIC \cite{schmidt1986multiple,liao2016music}, ESPRIT \cite{roy1989esprit,li2020super} and matrix pencil \cite{hua1990matrix,batenkov2021super}. Following the seminal work in \cite{donoho1992superresolution}, substantial progress has been made in understanding the role of sparsity as a prior for super-resolution \cite{candes2014towards,bernstein2019deconvolution,koulouri2020adaptive}. In recent times, convex optimization-based techniques have been developed that employ Total Variational (TV) norm and atomic norm regularizers, in order to promote sparsity \cite{candes2014towards,bernstein2019deconvolution,koulouri2020adaptive,bhaskar2013atomic,chi2020harnessing} and/or non-negativity \cite{bendory2017robust,schiebinger2017superresolution,morgenshtern2016super}. These techniques primarily employ sampling in the Fourier/frequency domain by assuming the kernel $h(t)$ to be  (approximately) bandlimited. However, selecting the appropriate cut-off frequency is crucial for super-resolution and needs careful consideration  \cite{batenkov2019rethinking,bernstein2019deconvolution}.
Unlike subspace-based methods, theoretical guarantees for these convex algorithms rely on a minimum separation between the spikes, which is also shown to be necessary even in absence of noise \cite{da2018tight}. 
The finite rate of innovation (FRI) framework \cite{blu2008sparse,uriguen2013fri,onativia2013finite,tur2011innovation,rudresh2017finite} also considers the recovery of spikes from measurements acquired using 
an exponentially decaying kernel, which includes the AR(1) filter considered in this paper. 
In the absence of noise, FRI enables the exact recovery of $K$ spikes with {\em arbitrary amplitudes} from $M=\Omega(K)$\footnote{\bl{This notation essentially means that there exists a positive constant c such that $M\geq c K$.}} measurements, without any separation condition \cite{onativia2013finite}. It is to be noted that all of the above methods require $M>K$ measurements for resolving $K$ spikes. In contrast, we will show that it is possible to recover $K$ spikes from $M \ll K$  measurements by exploiting the {\em binary nature of the spiking signal}. The above algorithms are designed to handle {\em arbitrary real-valued amplitudes} and as such, they are oblivious to binary priors. Therefore, they cannot successfully recover spikes in the regime $M<K$, which is henceforth referred to as the  \emph{extreme compression regime}.

The problem of recovering binary signals from underdetermined linear measurements 
(with more unknowns than equations/measurements) has been recently studied under the parlance of Binary Compressed Sensing (BCS)\cite{stojnic2010recovery,keiper2017compressed,flinth2019recovery,fosson2019recovery,tian2009detection,sarangi2021no,sarangi2022bin,razavikia2019reconstruction}. 
In BCS, the undersampling operation employs random (and typically dense) sampling matrices, whereas we consider a deterministic and structured measurement matrix derived from a filter, followed by uniform downsampling. Moreover, existing theoretical guarantees for BCS crucially rely on sparsity assumptions that will be shown to be inadequate for our problem (discussed in Section II-C). Most importantly, in order to achieve computational tractability, BCS relaxes the binary constraints and solves continuous-valued optimization problems. Consequently, their theoretical guarantees do not apply in the extreme compression regime $M<K$.

As mentioned earlier, our study is motivated by the problem of neural spike deconvolution arising in calcium imaging \cite{friedrich2017fast,deneux2016accurate,grewe2010high,vogelstein2009spike,jewell2020fast,onativia2013finite,sarangi2020effect}. A majority of the existing spike deconvolution techniques\cite{friedrich2017fast,deneux2016accurate,jewell2020fast} infer the spiking activity at \emph{the same (low) rate at which the fluorescence signal is sampled}, and a single estimate such as spike counts or rates are obtained over a temporal bin equal to the resolution of the imaging rate. 
Although sequential Monte-Carlo based techniques have been proposed that generate spikes at a rate higher than the calcium frame rate\cite{vogelstein2009spike}, no theoretical guarantees are available that prove that these methods can indeed uniquely identify the high-rate spiking activity. 
Algorithms that rely on sparsity and non-negativity \cite{friedrich2017fast,jewell2020fast} alone are ineffective for inferring the neural spiking activity that occurs at a much higher rate than the calcium sampling rate. On the other hand, at the high-rate, the spiking activity is often assumed to be binary since the probability of two or more spikes occurring within two time instants on the fine temporal grid is negligible\cite{brette2012handbook,rupasinghe2020robust}. Therefore, we propose to exploit the inherent binary nature of the neural spikes and provide the first theoretical guarantees that it is indeed possible to resolve the high-rate binary neural spikes from calcium fluorescence signal acquired at a much lower rate.
\vspace{-0.2cm}
\subsection{Our Contributions}
We make both theoretical and algorithmic contributions to the problem of binary super-resolution \bl{in the setting when the spikes lie on a fine grid}. 
We theoretically establish that at very low sampling rates, sparsity and non-negativity are inadequate for the exact reconstruction of binary spikes (Lemma 2). However, by exploiting the binary nature of the spiking activity, much stronger identifiability results  can be obtained compared to classical sparsity-based results (Theorem 1). In the absence of noise, we show that it is possible to uniquely recover $K$ binary spikes from only $M=\Omega(1)$ low-rate measurements. The analysis also provides interesting insights into the interplay between binary priors and the ``infinite memory" of the AR(1) filter. 

Although it is possible to uniquely identify binary spikes in the extreme compression regime ($M\ll K$), the combinatorial nature of binary constraints introduce computational hurdles in exactly enforcing them. Our second contribution is to leverage the special structure of the AR(1) measurements to overcome this computational challenge in the extreme compression regime $M<K$ (Section III-A).
Our formulation reveals an interesting and novel connection between binary super-resolution, and finding the generalized radix representation of real numbers, known as $\beta$-expansion\cite{sidorov2003almost,glendinning2001unique,renyi1957representations} (Section III). In order to circumvent the problem of exhaustive search, we pre-construct and store (in memory) a binary tree that is completely determined by the model parameters (filter and undersampling factor). 
When the low-rate measurements are acquired, we can {\em efficiently perform a binary search} to traverse the tree and find the desired binary solution. This ability to trade-off memory for computational efficiency is made possible by the unique structure of the measurement model governed by the AR(1) filter. The algorithm guarantees exact super-resolution even when the measurements are corrupted by a small bounded (adversarial) noise, the strength of which depends on the AR filter parameter and the undersampling factor. When the measurements are corrupted by additive Gaussian noise, we characterize the probability of erroneous decoding (Theorem 3) in the extreme compression regime $M<K$ and indicate the trade-off among the filter parameter, SNR and the extent of compression. Finally, we also demonstrate how binary priors can improve the performance of a popularly used spike deconvolution algorithm (OASIS \cite{friedrich2017fast}) on real calcium imaging datasets.
\vspace{-0.2cm}
\section{Fundamental Sample Complexity of Binary Super-resolution}
Let $y{\hi}[n]$ be the output of a stable first-order Autoregressive AR(1) filter with parameter $\alpha$, $0< \alpha <1$, driven by an unknown binary-valued input signal $x{\hi}[n]\in \{0,A\}$, $A>0$:
\begin{align}
y{\hi}[n]=\alpha y{\hi}[n-1]+ x{\hi}[n]\label{eqn:AR11}
\end{align}
In this paper, we consider a super-resolution setting where we do not directly observe $y{\hi}[n]$, and instead acquire $M$ measurements $\{y{\lo}[n]\}_{n=0}^{M-1}$ at a lower-rate by uniformly subsampling $y{\hi}[n]$ by a factor of $\tD$:
\begin{align}
 y{\lo}[n]=y{\hi}[\tD n], \quad n=0,1,\cdots,M-1, \label{eqn:low_AR}
\end{align}
The signal $y{\lo}[n]$ corresponds to a filtered and downsampled version of the signal $x{\hi}[n]$ where the filter is an infinite impulse response (IIR) filter with a single pole at $\alpha$. Let $\mathbf{y}{\lo} \in \mathbb{R}^{M}$ be a vector obtained by stacking the low-rate measurements $\{y{\lo}[n]\}_{n=0}^{M-1}$:
\begin{align*}
    \mathbf{y}{\lo}=[y{\lo}[0],y{\lo}[1],\cdots,y{\lo}[M-1]]^\top
\end{align*}
Since $\eqref{eqn:AR11}$ represents a causal filtering operation, the low rate signal $\mf{y}{\lo}$ only depends on the present and past high-rate binary signal. Denote $L:=(M-1)\tD+1$. The $M$ low-rate measurements in $\mf{y}{\lo}$ are a function of $L$ samples of the high rate binary input signal $\{x{\hi}[n]\}_{n=0}^{L-1}$. These $L$ samples are given by the following vector $\xhi\in \{0,A\}^L$:
\begin{align*}
    \xhi:=[x{\hi}[0],x{\hi}[1],\cdots,x{\hi}[L-1]]^\top.
\end{align*}
Assuming the system to be initially at rest, i.e., $y{\hi}[n]=0, n< 0$, we can represent the $M$ samples from \eqref{eqn:low_AR} in a compact matrix-vector form as:
\begin{equation}
    \smash{\mathbf{y}{\lo}:=\mathbf{S}_{\tD}\mf{y}{\hi}=\mathbf{S}_{\tD}\mathbf{G}_{\alpha}\xhi} \label{eqn:AR_mat}
\end{equation}
where $\mathbf{G}_\alpha\in \mathbb{R}^{L\times L}$ is a Toeplitz matrix given by:
\begin{equation}
\mathbf{G}_\alpha=\begin{bmatrix}
1 & 0 &\cdots &0\\
\alpha & 1 &\cdots &0\\
\vdots & \vdots & \ddots & \vdots\\
\alpha ^{L-1} & \alpha ^{L-2} &\cdots &1
\end{bmatrix} \label{eqn:G_alpha}
\end{equation}
and $\mathbf{S}_{\tD}\in \mathbb{R}^{M\times L}$ is defined as:
\begin{equation*}
\smash{[\mathbf{S}_{\tD}]_{i,j}=\begin{cases}
1, \quad j=(i-1)\tD+1\\
0, \text{ else }
\end{cases}.}
\end{equation*}
The matrix $\mathbf{S}_{\tD}$ represents 
the $\tD-$fold downsampling operation. Our goal is to infer the unknown high-rate binary input signal $x{\hi}[n]$ from the low-rate measurements $y{\lo}[n]$. This is essentially a ``super-resolution" problem because the AR(1) filter first attenuates the high-frequency components of $x{\hi}[n]$, and the uniform downsampling operation systematically discards measurements. As a result, it may seem that the spiking activity $\{ x{\hi}[(n-1)\tD+k]\}_{k=1}^{\tD}$ occurring ``in-between" two low-rate measurements $y{\lo}[n-1]$ and $y{\lo}[n]$ is apparently lost. One can potentially interpolate arbitrarily, making the problem hopeless. In the next section, we will show that surprisingly, $\xhi$ still remains identifiable from $\mf{y}{\lo}$ in the absence of noise, due to the binary nature of $\xhi$ and ``infinite memory" of the AR(1) filter.
\vspace{-0.2cm}
\subsection{Identifiability Conditions for Binary super-resolution}
Consider the following partition of $\xhi$ into $M$ disjoint blocks, where the first block is a scalar and the remaining $M-1$ blocks are of length $\tD$,
$
\xhi=[x{\hi}^{(0)},\mf{x}{\hi}^{(1)\top},\dots,\mf{x}{\hi}^{(M-1)\top}]^{\top}
$.
Here, $x{\hi}^{(0)}=x{\hi}[0]$ and $\mathbf{x}{\hi}^{(n)}\in \{0,A\}^{\tD}$ is given by:
\begin{align}
 [\mathbf{x}{\hi}^{(n)}]_k=x{\hi}[(n-1)\tD+k],\quad1\leq n \leq M-1 \label{eqn:block_x}
\end{align} 
The sub-vectors $\mathbf{x}{\hi}^{(n)},\text{ and }\mathbf{x}{\hi}^{(n-1)}$ ($n\geq 1$) represent consecutive and disjoint blocks (of length $\tD$) of the high-rate binary spike signal. In order to study the identifiability of $\xhi$ from $\ylo$, we first introduce an alternative (but equivalent) representation for \eqref{eqn:AR_mat}, by constructing a sequence $c[n]$ as follows $c[0]=y{\lo}[0],$
\begin{align}
    c[n]=y{\lo}[n]-\alpha^\tD y{\lo}[n-1], \ 1\leq n \leq M-1 \label{eqn:c_comp}
\end{align}
Given the high rate AR(1) model defined in \eqref{eqn:AR11}, it is possible to recursively represent $y{\hi}[\tD n]$ in terms of $y{\hi}[\tD n-1]$, which in turn, can be represented in terms of $y{\hi}[\tD n-2]$, and so on. By this recursive relation, we can represent $y{\hi}[\tD n-1]$ in terms of $y{\hi}[\tD n-\tD]$ and $\{x{\hi}[\tD n-i]\}_{i=0}^{\tD-1}$ and re-write $y{\lo}[n]$ as
\vspace{-0.2cm}
\begin{align}
    y{\lo}[n]&=y{\hi}[\tD n]=\alpha y{\hi}[\tD n-1]+x{\hi}[\tD n] \notag\\
    &=\alpha^{\tD} y{\hi}[\tD n-\tD]+\alpha^{\tD-1} x{\hi}[\tD(n-1)+1]+\cdots \notag\\
    &\qquad \qquad +\alpha x{\hi}[\tD n-1]+x{\hi}[\tD n], \notag\\
    y{\lo}[n]&-\alpha^{\tD} y{\lo}[n-1]=\alpha^{\tD-1} x{\hi}[\tD(n-1)+1]+\cdots \notag\\
    &\qquad \qquad +\alpha x{\hi}[\tD n-1]+x{\hi}[\tD n] \label{eqn:per_m}
\end{align}
The last equality holds due to the fact that $y{\lo}[n-1]=y{\hi}[\tD n-\tD]$. Combining \eqref{eqn:c_comp} and \eqref{eqn:per_m}, the sequence $c[n]$ can be re-written as $c[0]=y{\lo}[0]=x{\hi}^{(0)}$, and for $1\leq n \leq M-1$
\begin{align}
    c[n]=\sum_{i=1}^{\tD} \alpha^{\tD-i}x{\hi}[(n-1)\tD+i]=\mathbf{h}_{\alpha}^T\mathbf{x}{\hi}^{(n)} \label{eqn:scal1}
\end{align}
where $\mathbf{h}_{\alpha}=[\alpha^{\tD-1}, \alpha^{\tD-2},\dots,\alpha,1]^T \in \mathbb{R}^{\tD}$. This implies that $c[n]$ depends only on the block $\mathbf{x}{\hi}^{(n)}$. Denote  $\mathbf{c} := [c[0],c[1],\dots,c[M-1]]^{\top}\in \mathbb{R}^{M}$. For any $\tD$, \eqref{eqn:scal1} can be compactly represented as:
\vspace{-0.2cm}
\begin{align}
    \mathbf{c}=\mathbf{H}_{\tD}(\alpha)\xhi \label{eqn:binary_eq}
\end{align}
where $\mathbf{H}_{\tD}(\alpha)\in \mathbb{R}^{M\times L}$ is given by:
\begin{align*}
\mathbf{H}_{\tD}(\alpha)=\begin{bmatrix}
1 & \mf{0}^{\top} & \mf{0}^{\top} & \cdots& \mf{0}^{\top}\\
0 & \mathbf{h}_{\alpha}^{\top}& \mf{0}^{\top}& \cdots & \mf{0}^{\top}\\
0 &  \mf{0}^{\top}& \mathbf{h}_{\alpha}^{\top}&\cdots & \mf{0}^{\top}\\
\vdots & \vdots & \vdots & \ddots & \vdots\\
0 &\mf{0}^{\top} & \mf{0}^{\top}& \cdots &\mathbf{h}_{\alpha}^{\top}
\end{bmatrix}
\end{align*}
The following Lemma establishes the equivalence between \eqref{eqn:AR_mat} and \eqref{eqn:binary_eq}.
\begin{lem}\label{lem:equival}
Given $\mf{y}_{\textup{lo}}$, construct $\mathbf{c}$ following \eqref{eqn:c_comp}. Then, there is a unique binary $\mathbf{x}_{\textup{hi}}\in \{0,A\}^{L}$ satisfying \eqref{eqn:AR_mat} if and only if $\mathbf{x}_{\textup{hi}}$ is a unique binary vector satisfying \eqref{eqn:binary_eq}.
\end{lem}
\begin{proof}
First suppose that there is a unique binary $\xhi\in \{0,A\}^{L}$ satisfying \eqref{eqn:AR_mat} but \eqref{eqn:binary_eq} has a non-unique binary solution, i.e., there exists $\xhi^{\prime}\in \{0,A\}^{L}$, $\xhi^{\prime}\neq \xhi$, such that 
\ben
\smash{\mf{c}=\mf{H}_{\tD}(\alpha)\xhi=\mf{H}_{\tD}(\alpha)\xhi^{\prime}} \label{eqn:TSol}
\een
Define $\mf{y}{\hi}^{\prime}:=\mf{G}_{\alpha}\xhi^{\prime}$ whose entries are given by:
\begin{align}
    y{\hi}^{\prime}[n]=\sum_{k=0}^{n}\alpha^{n-k}x{\hi}^{\prime}[k], \quad 0\leq n \leq L-1 \label{eqn:y_prim}
\end{align}
Notice that \eqref{eqn:c_comp} can be re-written as 
\vspace{-0.2cm}
\begin{align*}
&y{\lo}[0]=c[0]=x{\hi}[0],y{\lo}[1]=c[1]+\alpha^{\tD}y{\lo}[0]=c[1]+\alpha^{\tD}c[0]\\
&y{\lo}[2]=c[2]+\alpha^{\tD}y{\lo}[1]=c[2]+\alpha^{\tD}c[1]+\alpha^{2\tD}c[0] \\
&\smash{\vdots}
\end{align*}
Following this recursive relation, and using \eqref{eqn:scal1} and \eqref{eqn:TSol}, 
we can further re-write $y{\lo}[n]$ as:
\begin{align}
    y{\lo}[n]&=\sum_{i=0}^{n} \alpha^{(n-i)\tD} c[i]=\alpha^{n\tD}x^{\prime}{\hi}^{(0)}+\sum_{i=1}^{n}\alpha^{(n-i)\tD} \mf{h}_{\alpha}^{\top} \xhi^{\prime(i)}\notag\\
    &=\alpha^{n\tD}x^{\prime}{\hi}^{(0)}+\sum_{i=1}^{n}\sum_{j=1}^{\tD} \alpha^{n\tD-(i-1)\tD-j} x^{\prime}{\hi}[(i-1)\tD+j]\notag\\
    &\overset{(a)}{=}
    \sum_{k=0}^{n\tD}\alpha^{n\tD-k}x^{\prime}{\hi}[k]\overset{(b)}{=}y^{\prime}{\hi}[n\tD]\label{eqn:scal2mat}
\end{align}
The equality $(a)$ follows by a re-indexing of the summation into a single sum, and $(b)$ follows from \eqref{eqn:y_prim}. By arranging \eqref{eqn:scal2mat} in a matrix form we obtain the following relation:
\be
\mf{y}{\lo}=\mf{S}_{\tD}\mf{G}_{\alpha}\xhi^{\prime}
\ee
However from \eqref{eqn:AR_mat}, we have $\mf{y}{\lo}=\mf{S}_{\tD}\mf{G}_{\alpha}\xhi$. This contradicts the supposition that \eqref{eqn:AR_mat} has a unique binary solution.

Next, suppose that \eqref{eqn:binary_eq} has a unique binary solution but the binary solution to \eqref{eqn:AR_mat} is non-unique, i.e., there exists $\xhi^{\prime} \in \{0,A\}^{L}$,  $\xhi^{\prime}\neq \xhi$ such that 
\be
\ylo=\mf{S}_{\tD}\mf{G}_{\alpha}\xhi^{\prime}=\mf{S}_{\tD}\mf{G}_{\alpha}\xhi
\ee
By following \eqref{eqn:c_comp} and \eqref{eqn:binary_eq}, we also have $\mf{c}=\mathbf{H}_{\tD}(\alpha)\xhi^{\prime}=\mathbf{H}_{\tD}(\alpha)\xhi$ which contradicts the assumption that solution of \eqref{eqn:binary_eq} is unique. 
\end{proof}
Lemma \ref{lem:equival} assures that a binary $\xhi$ is uniquely identifiable from measurements $\mf{y}{\lo}$ if and only if there is a unique binary solution $\xhi\in \{0,A\}^L$ to \eqref{eqn:binary_eq}. From \eqref{eqn:scal1}, it can be seen that $c[n]$ and $c[n-1]$ have contributions from only disjoint blocks of high rate spikes $\xhi^{(n)},\text{ and }\xhi^{(n-1)}$.
Hence effectively, we only have a \emph{single scalar measurement} $c[n]$ to decode an entire block $\mathbf{x}{\hi}^{(n)}$ of length $\tD$, regardless of how sparse it is. The task of decoding $\xhi^{(n)}$ from a single measurement seems like a hopelessly ``ill-posed" problem, caused by the uniform downsampling operation. But this is precisely where the binary nature of $\xhi$ can be used as a powerful prior to make the problem well-posed. Theorem \ref{thm:ident} specifies conditions under which it is possible to do so.
\begin{thm}(Identifiability)\label{thm:ident}
For any $\alpha \in (0,1)$, with the possible exception of $\alpha$ belonging to \bl{a set} of Lebesgue measure zero, 
there is a unique $\mathbf{x}_{\textup{hi}}\in \{0,A\}^{L}$ that satisfies (\ref{eqn:binary_eq}) for every $\tD\geq 1$.
\end{thm}
\begin{proof}
In Appendix A.
\end{proof}

Using Lemma \ref{lem:equival} and Theorem \ref{thm:ident}, we can conclude that $\xhi$ is uniquely identifiable from $\ylo$ for almost all $\alpha \in (0,1)$. \bl{It can be verified that for $\alpha=1$ the mapping is non-injective.}
Theorem \ref{thm:ident} establishes that it is fundamentally possible to decode each block $\xhi^{(n)}$ of length D, from effectively a single measurement $c[n]$. Since $\xhi^{(n)}$ can take $2^{\tD}$ possible values, in principle, one can always perform an exhaustive search over these $2^\tD$ possible binary sequences and by Theorem \ref{thm:ident}, only one of them will satisfy $c[n]=\mathbf{h}_{\alpha}^{\top}\xhi^{(n)}$. Since exhaustive search is computationally prohibitive, this leads to the natural question regarding alternative solutions. In Section III, we will develop an alternative algorithm that leverages the trade-off between memory and computation to achieve a significantly lower run-time decoding complexity.

\vspace{-0.3cm}
\subsection{Comparison with Finite Rate of Innovation Approach}
In a related line of work \cite{blu2008sparse,uriguen2013fri,onativia2013finite,rudresh2017finite}, the FRI framework has been developed to reconstruct spikes from the measurement model considered here. However, in the general FRI framework, there is no assumption on the amplitude of the spikes, and there are a total of $2\tD$ real valued unknowns corresponding to the locations and amplitudes of $\tD$ spikes. In \cite{onativia2013finite}, it was shown that by leveraging the property of exponentially reproducing kernels, it is possible to recover arbitrary amplitudes and spike locations using Prony-type algorithms, provided at least $2\tD+1 (>\tD)$ low-rate measurements are available. However, since we exploit the binary nature of spiking activity, we can operate at a much smaller sample complexity than FRI. In fact, Theorem $1$ shows that when we exploit the fact that the spikes occur on a high-resolution grid with binary amplitudes, $M=\Omega(1)$ measurements suffice to identify $\tD$ spikes regardless of how large $\tD$ is. A direct application of the FRI approach cannot succeed in this regime, since the number of spikes is larger than the number of measurements. That being said, with enough measurements, FRI techniques are powerful, and they can also identify off-grid spikes. In future, it would be interesting to combine the two approaches by incorporating binary priors to FRI based techniques and remove the grid assumptions. 
\vspace{-0.2cm}
\subsection{Curse of Uniform Downsampling: Inadequacy of sparsity and non-negativity}\label{subsec:spar}
By virtue of being a binary signal, $\xhi$ is naturally sparse and non-negative. Therefore, one may ask if sparsity and/or non-negativity are sufficient to uniquely identify $\xhi$ from $\mf{c}$, without the need for imposing any binary constraints. In particular, we would like to understand if the solution to the following problem that seeks the sparsest non-negative vector in $\mathbb{R}^{L}$ satisfying \eqref{eqn:binary_eq} indeed coincides with the true $\xhi\in \{0,A\}^{L}$
\begin{align*}
    \min_{\mathbf{x}\in \mathbb{R}^L}\quad \Vert \mathbf{x} \Vert_0 \quad\text{subject to } \mathbf{c}=\mathbf{H}_{\tD}(\alpha) \mathbf{x}, \quad \mathbf{x} \geq \mathbf{0} \tag{P0}\label{eqn:P0}
\end{align*}
\begin{lem}\label{lem:sp_l0}
For every $\xhi\in \{0,A\}^{L}$ (except $\xhi=A \mf{e}_1$), $\text{ and }\mathbf{c}\in \mathbb{R}^{M}$ satisfying \eqref{eqn:binary_eq}, the following are true
\begin{enumerate}[(i)]
    \item  There exists a solution \bl{$\mathbf{x}^{\star} \neq \mf{x}_{\textup{hi}}$} to \eqref{eqn:P0} satisfying 
\begin{align}
    \Vert \mathbf{x}^{\star} \Vert_0 \leq \Vert\mf{x}_{\textup{hi}} \Vert_0    \label{eqn:sparse}
\end{align}
\item The inequality in \eqref{eqn:sparse} is strict as long as there exists an integer $n_0 \geq 1$ such that the block $\mf{x}_{\textup{hi}}^{(n_0)}$ of $\mf{x}_{\textup{hi}}$ (defined in \eqref{eqn:block_x}) satisfies $
     \Vert \mf{x}_{\textup{hi}}^{(n_0)} \Vert_0 \geq 2$.
\end{enumerate}
\end{lem}
\begin{proof}
The proof is in Appendix B.
\end{proof}

Lemma \ref{lem:sp_l0} shows there exist other non-binary solution(s) to \eqref{eqn:binary_eq} (different from $\xhi$) that have the same or smaller sparsity as the binary signal $\xhi \in \{0,A\}^{L}$. Furthermore, there exist problem instances where the sparsest solution to \eqref{eqn:P0} is strictly sparser than $\xhi$. Hence, sparsity and/or non-negativity are inadequate to identify the ground truth $\xhi$ uniquely.

\bl{\textbf{Implicit Bias of Relaxation:}
The optimization problem \eqref{eqn:P0} is non-convex and the binary constraints are not enforced. In binary compressed sensing \cite{stojnic2010recovery,keiper2017compressed}, it is common to relax the binary constraints using box-constraint and $l_0$ norm is relaxed to $l_1$ norm in the following manner:
\begin{align*}
    \min_{\mathbf{x}\in \mathbb{R}^L}\ \Vert \mathbf{x} \Vert_1 \quad\text{subject to } \mathbf{c}=\mathbf{H}_{\tD}(\alpha) \mathbf{x}, \ \mf{0}\leq \mathbf{x} \leq A\mathbf{1} \tag{P1-B}\label{eqn:P1_box}
\end{align*}
In the following Lemma, we show that there is an implicit bias introduced to the solution of \eqref{eqn:P1_box}. 
\begin{lem}\label{lem:sp_l1_box}
For every $\xhi\in \{0,A\}^{L}$, $\text{and }\mathbf{c}\in \mathbb{R}^{M}$ satisfying \eqref{eqn:binary_eq}. There exists a solution $\mathbf{x}^{\star}$ to \eqref{eqn:P1_box} satisfying 
\begin{equation}
    \smash{\Vert \mathbf{x}^{\star} \Vert_1 \leq \Vert\mf{x}_{\textup{hi}} \Vert_1.}    \label{eqn:relax_l1}
\end{equation}
Moreover, for all $n\geq 1$, the blocks $\mf{x}^{(n)\star}\! \in\! \mathbb{R}^{\tD}$ of $\mathbf{x}^{\star}$ satisfy:
\begin{align}
    \text{supp}(\mf{x}^{(n)\star})=\{\tD,\tD-1,\cdots,\tD-j_n\}, \text{ if } c[n]\neq 0 \label{eqn:supp_l1_box}
\end{align}
for some $0\leq j_n \leq \tD-1$ and $\mathbf{x}^{(n)\star}=\mf{0}$ if $c[n]=0$, irrespective of the support of $\xhi$.
\end{lem}
\begin{proof}
The proof is in Appendix B.
\end{proof}}
\bl{Lemma \ref{lem:sp_l1_box} shows that even in the noiseless setting, introducing the box-constraint as a means of relaxing the binary constraint introduces a bias {in} the support of the recovered spikes. The optimal solution always results in spikes with support clustered towards the end of {each} block of length $\tD$, irrespective of the ground truth spiking pattern $\xhi$ that generated the measurements. This bias is a consequence of the {nature} of relaxation, as well as the specific structure of the measurement matrix $\mathbf{H}_{\tD}(\alpha)$ arising in the problem. 
}
\vspace{-0.3cm}
\subsection{Role of Memory in Super-resolution: IIR vs. FIR filters}\label{ssec:FIR}
The ability to identify the high-rate binary signal  $\xhi \in \{0,A\}^{L}$ from $\tD-$fold undersampled measurements $\mathbf{y}{\lo}$ (for arbitrarily large $\tD$) in the absence of noise, is in parts also due to the ``infinite memory" or infinite impulse response of the AR(1) filter. Indeed, for an Finite Impulse Response (FIR) filter, there is a limit to downsampling without losing identifiability. This was recently studied in our earlier work \cite{sarangi2021no} where we showed that the undersampling limit is determined by the length of the FIR filter. To see this, consider the convolution of a binary valued signal $\xhi$ with a FIR filter $\mathbf{u}=[u[0],u[1],\cdots,u[r-1]]^T\in \mathbb{R}^r$ of length $r$:
$
    z_f[n]=\sum_{i=0}^{r-1} u[r-1-i]x{\hi}[n+i].
$
These samples are represented in the vector form as $\mf{z}_f:=\mf{u}\star \xhi \in \mathbb{R}^{L}$ (by suitable zero padding). Suppose, as before, we only observe a ${\tD}-$fold downsampling of the output $z_\tD[n]=z_f[\tD n]$. Two consecutive samples $z_\tD[p],z_\tD[p+1]$ of the low-rate observation are given by:
\begin{align*}
    &\smash{z_\tD[p]=\sum_{i=0}^{r-1} u[r-1-i]x{\hi}[\tD p+i]}, \\
    &z_\tD[p+1]=\sum_{i=0}^{r-1} u[r-1-i]x{\hi}[\tD(p+1)+i]
\end{align*}
If ${\tD}>r$, notice that none of the measurements is a function of the samples $x{\hi}[\tD p+r],x{\hi}[\tD p+r+1],\cdots,x{\hi}[\tD(p+1)-1]$. Hence, it is possible to assign them arbitrary binary values and yet be consistent with the low-rate measurements $z_{\tD}[n]$. This makes it impossible to exactly recover $\xhi$ (even if it is known to be binary valued) if the decimation is larger than the filter length ($\tD>r$). The following lemma summarizes this result.
\begin{lem}\label{lem:fir}
For every FIR filter $\mathbf{u} \in \mathbb{R}^{r}$, if the undersampling factor exceeds the filter length, i.e. ${\textup{D}}>r$, there exist $\mathbf{x}_0,\mf{x}_1 \in \{0,A\}^L$, $\mf{x}_0\neq \mf{x}_1$ such that $\mathbf{S}_{\textup{D}}(\mathbf{u}\star\mathbf{x}_0)=\mathbf{S}_{{\textup{D}}}(\mathbf{u}\star\mathbf{x}_1)$.
\end{lem}
This shows that the identifiability result presented in Theorem 1 is not merely a consequence of binary priors but the infinite memory of the autoregressive process is also critical in allowing arbitrary undersampling $\tD>1$ in absence of noise. For such IIR filters, the memory of all past (binary) spiking activity is encoded (with suitable weighting) into every measurement captured after the spike, which would not be the case for a finite impulse response filter.



\vspace{-0.2cm}
\section{Efficient Binary Super-Resolution Using Binary Search with Structured Measurements}
By Theorem \ref{thm:ident}, we already know that it is possible to uniquely identify $\xhi$ from $\mf{c}$ (or equivalently, each block $\mf{x}{\hi}^{(n)}$ from a single measurement $c[n]$) by exhaustive search. We now demonstrate how this exhaustive search can be avoided by formulating the decoding problem in terms of ``binary search" over an appropriate set, and thereby attaining computational efficiency.  
We begin by introducing some notations and definitions. Given a non-negative integer $k, 0 \leq k \leq 2^{\tD}-1$, let $(b_1(k), b_2(k), \cdots, b_\tD(k))$ be the unique $\tD$-bit binary representation of $k$:
$
    k=\sum_{d=1}^{\tD} 2^{\tD-d} b_d(k), \quad b_d(k)\in \{0,1\} \ \forall \ 1\leq d \leq \tD.
$
Here $b_1(k)$ is the most significant bit and $b_{\tD}(k)$ is the least significant bit. Using this notation, we define the following set: \begin{align}
    \mathcal{S}_{\text{all}}:=\{ \mathbf{v}_0,\mathbf{v}_1,\mathbf{v}_2,\cdots,\mathbf{v}_{2^\tD-1}\}, \label{eqn:S_all}
\end{align}
where each $\mathbf{v}_k\in \{0,A\}^{\tD}$ is a binary vector given by
\begin{align}
    [\mf{v}_k]_{d}= Ab_{d}(k). \quad 1\leq d \leq \tD
\end{align}
In other words, the binary vector $\frac{1}{A}\mf{v}_k$ is the $\tD$-bit binary representation of its index $k$. Using this convention, $\mathbf{v}_0=\mathbf{0}$ (i.e., a binary sequence of all $0'$s) and $\mathbf{v}_{2^{\tD}-1}=A\mathbf{1}$ (i.e., a binary sequence of all $A'$s). Recall the partition of $\xhi$ defined in \eqref{eqn:block_x}, where each block $\xhi^{(n)}$ ($n\geq 1$) is a binary vector of length $\tD$ and $x{\hi}^{(0)}\in \{0,A\}$ is a scalar. It is easy to see that \eqref{eqn:S_all} comprises of all possible values that each block $\xhi^{(n)}$ can assume. 
According to \eqref{eqn:scal1} each scalar measurement $c[n]$ can be written as:
$
    c[0]=x^{(0)}, \quad c[n]=\mathbf{h}_{\alpha}{^\top} \xhi^{(n)},  \ 1\leq n \leq M-1. 
$
For every $\alpha$, we define the following set:
\begin{align}
    \Theta_{\alpha}:=\{\theta_0,\theta_1,\cdots,\theta_{2^{\tD}-1}\}, \text{ where }\theta_k:=\mf{h}_{\alpha}^{\top}\mf{v}_k \label{eqn:Talpha}
\end{align}
Observe that every measurement $c[n]=\sum_{i=1}^{\tD} \alpha^{\tD-i}x{\hi}[(n-1)\tD+i]$ takes values from this set $\Theta_{\alpha}$, depending on the value taken by the underlying block of spiking pattern from $ \mathcal{S}_{\text{all}}$. Our goal is to recover the spikes $\{x{\hi}[(n-1)\tD+i]\}_{i=1}^{\tD}$ from $c[n]$. 

In the following, we show that this problem is equivalent to finding the representation of a real number over an arbitrary radix, which is known as ``\emph{$\beta$}-expansion" \cite{renyi1957representations}. Given a real (potentially non-integer) number $\beta>1$, the representation of another real number $p \geq 0$ of the form:
\begin{align}
    p=\sum_{n=1}^{\infty} a_n \beta^{-n}, \text{ where } 0\leq a_n < \lfloor \beta \rfloor \label{eqn:beta_exp}
\end{align}
is referred to as a $\beta$-expansion of $p$. The coefficients $0\leq a_n < \lfloor \beta \rfloor$ are integers. This is a generalization of the representation of numbers beyond integer-radix to a system where the radix can be chosen as an arbitrary real number. This notion of representation over arbitrary radix was first introduced by Renyi in \cite{renyi1957representations}, and since then has been extensively studied \cite{glendinning2001unique,frougny1992finite,sidorov2003almost}. There is a direct connection between $\beta$-expansion and the binary super-resolution problem considered here. In the problem at hand, any element $\theta_k\in \Theta_{\alpha}$ can be written as:
\begin{align*}
   \smash{ \theta_k=\mathbf{h}_{\alpha}^{\top}\mf{v}_k=\sum_{i=1}^{\tD}\alpha^{\tD-i} [\mathbf{v}_k]_i}
\end{align*}
When $1/2< \alpha <1$, by letting $\beta=1/\alpha$, we see that the coefficients in \eqref{eqn:beta_exp} must satisfy $0 \leq a_n <\lfloor 1/\alpha \rfloor <2$, i.e., they are restricted to be binary valued $a_n\in \{0,1\}$. \emph{Therefore, decoding the spikes $\mf{v}_k$ from the observation $\theta_k$ is equivalent to finding a $\tD-$bit representation for the number $\theta_k/A$ over the non-integer radix $\beta=1/\alpha$.} Questions regarding the existence of $\beta$-expansion, and finding the coefficients of a finite $\beta-$expansion (whenever it exists) has been an active topic of research \cite{glendinning2001unique,frougny1992finite,sidorov2003almost,komornik2011expansions}. When $\beta \geq 2$ (equivalently, $0<\alpha \leq 1/2$), it is possible to find the coefficients using a greedy algorithm which proceeds in a fashion similar to finding the $\tD$-bit binary representation of an integer \cite{komornik2011expansions,sidorov2003almost}. However, the regime $\beta \in (1,2)$ (equivalently $1/2<\alpha <1$), is significantly more complicated and is of continued research interest \cite{glendinning2001unique,frougny1992finite,sidorov2003almost}. 
{To the best of our knowledge, there are no known computationally efficient ways to find the finite $\beta$-expansion when $1/2 < \alpha <1$ (if it exists) [N. Sidorov, personal communication, May 24, 2022]. In practice, we encounter filter values $\alpha \ (=1/\beta)$ that are much closer to $1$, and hence, we need an alternative approach to find this finite $\beta$-radix representation for $\theta_k$. In the next section, we show that by performing a suitable preprocessing, finite $\beta$-radix representation can be formulated as a binary search problem which is guaranteed to succeed for all values of $\beta$ that permit unique finite $\beta-$expansions.}
\vspace{-0.5cm}
\subsection{Formulation as a Binary Search Problem}
Before describing the algorithm, we first introduce the notion of a \emph{collision-free} set.
\begin{def1}[Collision Free set]
Given an undersampling factor $\textup{D}$, define a class of ``collision free" AR(1) filters as:
\begin{align*}
    \mathcal{G}_{\tD}=\{ \alpha \in (0,1) \text{ s.t. } \mathbf{h}_\alpha^{\top} \mathbf{v}_i \neq \mathbf{h}_{\alpha}^{\top} \mathbf{v}_j \ \forall\  i \neq j, \mathbf{v}_i, \mathbf{v}_j \in \mathcal{S}_{\textup{all}} \}
\end{align*}
\end{def1}
The set $\mathcal{G}_{\tD}$ denotes permissible values of the AR(1) filter parameter $\alpha$ such that each of the $2^\tD$ binary sequences in $\mathcal{S}_{\text{all}}$ maps to a unique element in the set $\Theta_{\alpha}$. In other words, every $\theta_k \in \Theta_{\alpha}$ has a unique $\tD-$bit expansion for all $\alpha \in \mathcal{G}_{\tD}$. This naturally raises the question ``How large is the set $\mathcal{G}_{\tD}$?". Theorem 1 already provided the answer to this question, where the identifiability result implies that for every $\tD$, almost all $\alpha \in (0,1)$ belong to this set $\mathcal{G}_{\tD}$ (with the possible exception of a measure zero set). Hence, Theorem \ref{thm:ident} ensures that there are infinite choices for collision-free filter parameters. 
\begin{lem} \label{lem:bij}
For every $\alpha\in \mathcal{G}_{\tD}$, the mapping $\blds{\Phi}_{\alpha}(.):\mathcal{S}_{\textup{all}} \rightarrow \Theta_{\alpha}$, $\blds{\Phi}_{\alpha}(\mathbf{v})=\mathbf{h}_{\alpha}^{\top} \mathbf{v}$ forms a bijection between $\mathcal{S}_{\textup{all}}$ and $\Theta_{\alpha}$.
\end{lem}
\begin{proof} 
Since $\alpha\in \mathcal{G}_{\tD}$, from the definition of the set $\mathcal{G}_{\tD}$, it is clear that for any $\mathbf{v}_i,\mathbf{v}_j \in \mathcal{S}_{\text{all}}$, $\mathbf{v}_i\neq \mathbf{v}_j$ we have $\mathbf{h}_{\alpha}{^\top}\mathbf{v}_i\neq\mathbf{h}_{\alpha}{^\top} \mathbf{v}_{j}$. Therefore, the mapping is injective. Furthermore, from \eqref{eqn:Talpha} we also have $\vert\Theta_{\alpha}\vert \leq \vert \mathcal{S}_{\text{all}}\vert=2^{\tD}$. Since
$\blds{\Phi}_{\alpha}(\cdot)$ is injective, we must also have $\vert\Theta_{\alpha}\vert=2^{\tD}$ and hence the mapping $\blds{\Phi}_{\alpha}(.)$ forms a bijection between $\mathcal{S}_{\text{all}}$ and $\Theta_{\alpha}$.
\end{proof}
When $\alpha \in \mathcal{G}_{\tD}$, Lemma \ref{lem:bij} states that the finite beta expansion for every $\theta_k \in \Theta_{\alpha}$ is unique. Lemma \ref{lem:bij} provides a way to avoid exhaustive search over $\mathcal{S}_{\text{all}}$, and yet identify $\xhi^{(n)}$ from $c[n]$ in a computationally efficient way. From Lemma \ref{lem:bij}, we know that each of the $2^{\tD}$ spiking patterns in $\mathcal{S}_{\text{all}}$ maps to a unique element in $\Theta_{\alpha}$, and each element in $\Theta_{\alpha}$ has a corresponding spiking pattern. Hence instead of searching $\mathcal{S}_{\text{all}}$, we can equivalently search the set $\Theta_{\alpha}$ in order to determine the unknown spiking pattern. \bl{Since $\Theta_{\alpha}$ permits ``ordering", searching $\Theta_{\alpha}$ has a distinct computational advantage over searching $\mathcal{S}_{\text{all}}$.} This ordering enables us to employ binary search over (an ordered) $\Theta_{\alpha}$ and find the desired element in a computationally efficient manner. 
To do this, we first sort the set $\Theta_{\alpha}$ (in ascending order) and {\em arrange the corresponding elements of $\mathcal{S}_{\text{all}}$ in the same order}. Given $\Theta_{\alpha}$ as an input, the function $\text{SORT}(\cdot)$ returns a sorted list $\Theta^{\text{sort}}_{\alpha}$, and an index set $\mathcal{I}=\{i_0,i_1,\cdots,i_{2^{\tD}-1}\}$ containing the indices of the sorted elements in the list $\Theta_{\alpha}$. 
\vspace{-0.1cm}
\begin{equation*}
    \smash{\Theta^{\text{sort}}_{\alpha}, \mathcal{I}\leftarrow\text{SORT}(\Theta_{\alpha})}
\end{equation*}
Let us denote the elements of the sorted lists as $\Theta^{\text{sort}}_{\alpha}=\{\tilde{\theta}_0,\cdots,\tilde{\theta}_{2^{\tD}-1}\}$, and  $\mathcal{S}^{\text{sort}}_{\text{all}}=\{\mathbf{\tilde{v}}_0,\cdots,\mathbf{\tilde{v}}_{2^{\tD}-1}\}$ where:
\begin{equation*}
\smash{\tilde{\theta}_0<\tilde{\theta}_1<\dots <\tilde{\theta}_{2^{\tD}-1}} \quad \text{ and } 
\smash{\tilde{\theta}_j=\theta_{i_j},\quad \mathbf{\tilde{v}}_j=\mathbf{v}_{i_j} \quad \forall j}.
\end{equation*}
It is important to note that this sorting step does not depend on the measurements $\mf{c}$, and can therefore be part of a pre-processing pipeline that can be performed offline. However, it does require memory to store the sorted lists. 
 \begin{algorithm}[h]
\caption{Noiseless Spike Recovery}
\begin{algorithmic}[1]
\item \textbf{Input:} Measurement $c[n]$, Sorted list $\Theta^{\text{sort}}_{\alpha}$ and the corresponding (ordered) spike patterns $\mathcal{S}^{\text{sort}}_{\text{all}}$
\item \textbf{Output:} Decoded spike block  $\hat{\mathbf{x}}{\hi}^{(n)}$\\
$i^{\star} \leftarrow \text{BINSEARCH}(\Theta^{\text{sort}}_{\alpha},c[n])$\\
Return $\hat{\mathbf{x}}{\hi}^{(n)} \leftarrow \mathbf{\tilde{v}}_{i^{\star}}$
\end{algorithmic}
 \label{alg:LLS}
\end{algorithm}
In the noiseless setting, we know that every scalar measurement $c[n]=\mf{h}_{\alpha}^{\top}\xhi^{(n)}$ belongs to the set $\Theta^{\text{sort}}_{\alpha}$. Therefore, if we identify its index, say $i^{\star}$, then we can successfully recover $\mf{x}{\hi}^{(n)}$ by returning the corresponding binary vector  $\tilde{\mf{v}}_{i^{\star}}$ from $\mathcal{S}^{\text{sort}}_{\text{all}}$. Therefore, we can formulate the decoding problem as searching for the input $c[n]$ in the sorted list $\Theta^{\text{sort}}_{\alpha}$. This can be efficiently done by using ``Binary Search". The noiseless spike decoding procedure is summarized as Algorithm $1$. 
Since the complexity of performing a binary search over an ordered list of $N$ elements is $O(\log N)$, the complexity of Algorithm $1$ is logarithmic in the cardinality of $\Theta^{\text{sort}}_{\alpha}$, which results in a complexity of $O(\log(2^\tD))=O(\tD)$. We summarize this result in the following Lemma. 
\begin{lem}\label{lem:algo_g}
Assume $\alpha \in \mathcal{G}_{\tD}$. Given the ordered set $\Theta^{\text{sort}}_{\alpha}$ , and an input $c[n]=\mathbf{h}_{\alpha}^{\top}\xhi^{(n)}$, Algorithm $1$ terminates in $O(\tD)$ steps and its output $\mathbf{\hat{x}}{\hi}^{(n)}$ satisfies $\mathbf{\hat{x}}{\hi}^{(n)}=\xhi^{(n)}$.
\end{lem}
\vspace{-0.4cm}
\subsection{Noisy Measurements and $1$ \textup{D} Nearest Neighbor Search}
We demonstrate how binary search can still be useful in presence of noise by formulating noisy spike detection as a {one dimensional} \emph{nearest neighbor search} problem. Suppose $\{z{\lo}[n]\}_{n=0}^{M-1}$ denote noisy $\tD$-fold decimated filter output
\begin{equation}
    \smash{z{\lo}[n]=y{\lo}[n]+w[n],} \quad 0\leq n \leq M-1 \label{eqn:noisy_measurement}
\end{equation}
Here $w[n]$ represents the additive noise term that corrupts the (noiseless) low-rate measurements $y{\lo}[n]$. Similar to \eqref{eqn:c_comp}, we compute $c_{e}[n]$ from $z{\lo}[n]$  
as follows:
\begin{align}
    c_e[n]&=z{\lo}[n]-\alpha^{\tD} z{\lo}[n-1] \label{eqn:c_op_noise}\\
    &=\sum_{i=1}^{\tD}\alpha^{\tD-i}x{\hi}[(n-1)\tD+i]+e[n]\smash{= c[n] + e[n]}\label{eqn:noisy_AR}
\end{align}
where $c[n]=\mathbf{h}_{\alpha}^{\top}\mathbf{x}{\hi}^{(n)}\in \Theta^{\text{sort}}_{\alpha}$, and $e[n]=w[n]-\alpha^{\tD} w[n-1]$. We can interpret $c_e[n]$ as a noisy/perturbed version of an element $c[n] \in \Theta^{\text{sort}}_{\alpha}$, with $e[n]$ representing the noise. This perturbed signal may no longer belong to $\Theta^{\text{sort}}_{\alpha}$ (i.e. $c_e[n]\not\in \Theta^{\text{sort}}_{\alpha}$) and hence, we cannot find an exact match in the set $\Theta^{\text{sort}}_{\alpha}$. Instead, we aim to find the closest element in $\Theta^{\text{sort}}_{\alpha}$ (the nearest neighbor of $c_e[n]$) by solving the following problem:
\begin{equation}
     \smash{\hat{\mf{x}}{\hi}^{(n)}=\arg \min_{\mf{v} \in \mathcal{S}^{\text{sort}}_{\text{all}}} \vert c_e[n]- \mf{h}_{\alpha}^{\top}\mf{v} \vert} \label{eqn:g_nn}
\end{equation}
Solving \eqref{eqn:g_nn} is equivalent to finding the spike sequence $\mf{\tilde{v}}\in \mathcal{S}^{\text{sort}}_{\text{all}}$ that maps to the nearest neighbor of $c_e[n]$ in the set $\Theta^{\text{sort}}_{\alpha}$.
By leveraging the sorted list $\Theta^{\text{sort}}_{\alpha}$, it is no longer necessary to parse the list sequentially (which would incur $O(2^{\tD})$ complexity), instead we can perform a modified binary search as summarized in Algorithm $2$,  
that keeps track of additional indices compared to the vanilla binary search. Finally, we return the unique spiking pattern from $\mathcal{S}^{\text{sort}}_{\alpha}$ that gets mapped to the nearest neighbor of the noisy measurement $c_e[n]$. 
It is well-known that the nearest neighbor for any query could be found in $O(\log (2^{\tD}))=O(\tD)$ steps, instead of the linear complexity of $O(2^{\tD})$. This guarantees a computationally efficient decoding of spikes by solving \eqref{eqn:g_nn}. 

Next, we characterize the error events that lead to erroneous detection of a block of spikes. Recall that the set $\Theta^{\text{sort}}_{\alpha}$ is sorted, and its elements satisfy the ordering:
\begin{equation*}
\smash{0=\tilde{\theta}_0<\tilde{\theta}_1<\dots <\tilde{\theta}_{l_\tD}=1+\alpha+\cdots+\alpha^{\tD-1}}
\end{equation*}
where $l_\tD:=2^{\tD}-1$. We also have  $\tilde{\theta}_k=\mf{h}_{\alpha}^{\top}\mf{\tilde{v}}_k$, where $\mf{\tilde{v}}_k \in \mathcal{S}_{\text{all}}^{\text{sort}}$ is a binary spiking sequence of length $\tD$.

For each $\mathbf{\tilde{v}}_k$ and each $n$, we will determine the error event $\hat{\mathbf{x}}{\hi}^{(n)} \neq \xhi ^{(n)}$, when $\xhi ^{(n)} = \mathbf{\tilde{v}}_k$. First, consider the scenario when $\xhi^{(n)}=\mathbf{\tilde{v}}_k$ for some $0<k<l_\tD$ (excluding $\mathbf{\tilde{v}}_0, \mathbf{\tilde{v}}_{l_D}$). 
The corresponding noiseless measurement is $c[n]=\tilde{\theta}_k=\mathbf{h}_{\alpha}^{\top}\mathbf{\tilde{v}}_k$ which satisfies
$
    \smash{\tilde{\theta}_{k-1}< c[n]=\tilde{\theta}_k <\tilde{\theta}_{k+1}}
$.
Since $\Theta^{\text{sort}}_{\alpha}$ is sorted, it can be easily verified that the nearest neighbor of $c_e[n]$ will be $\tilde{\theta}_k$, if and only if $c_e[n]$ satisfies the following condition:
\begin{equation}
    \smash{(\tilde{\theta}_{k-1}+\tilde{\theta}_{k})/2\leq  c_e[n]\leq (\tilde{\theta}_{k+1}+\tilde{\theta}_{k})/2} \label{eqn:succ_event}
\end{equation}
Since $\tilde{\theta}_k=\mf{h}_{\alpha}^{\top}\tilde{\mf{v}}_k$, the solution to \eqref{eqn:g_nn} is attained at $\tilde{\mf{v}}_k \in \mathcal{S}^{\text{sort}}_{\text{all}}$, and the decoding is successful. Therefore Algorithm $2$ produces an erroneous estimate of $\mathbf{\tilde{v}}_k$ if and only if $c_e[n]$ violates \eqref{eqn:succ_event}. The event $c_e[n] \not\in [\frac{\tilde{\theta}_{k-1}+\tilde{\theta}_{k}}{2}, \frac{\tilde{\theta}_{k+1}+\tilde{\theta}_{k}}{2}]$ is equivalent to $e[n]\in \mathcal{E}_k$ ($e[n]$ is defined earlier in \eqref{eqn:noisy_AR}), where
\begin{align}
    \smash{\mathcal{E}_k=\{e[n] < -\frac{\tilde{\theta}_{k}-\tilde{\theta}_{k-1}}{2},\text{ or } e[n] > \frac{\tilde{\theta}_{k+1}-\tilde{\theta}_{k}}{2}\}} \label{eqn:err_in}
\end{align}
Finally, we characterize the error events for $k=0,l_{\tD}$. The error events for $c[n]=\theta_0=0$ or $c[n]=\theta_{l_{\tD}}$ are given by:
\begin{align}
    \smash{\hspace{-0.1cm}\mathcal{E}_0=\{e[n] \geq \tilde{\theta}_{1}/2\},\ \mathcal{E}_{l
    _{\tD}}=\{e[n]\leq -(\tilde{\theta}_{l_\tD}-\tilde{\theta}_{l_\tD-1})/2\}}\label{eqn:err_bou}
\end{align}
Define the ``minimum distance" between points in $\Theta^{\text{sort}}_{\alpha}$:
\be
\smash{\Delta \theta_{\min}(\alpha,\tD)= \min_{1\leq k \leq l_{\tD}} \vert \tilde{\theta}_{k}-\tilde{\theta}_{k-1}\vert}.
\ee 
This minimum distance depends on $A, \alpha$ and $\tD$. From \eqref{eqn:err_in}, \eqref{eqn:err_bou} it can be verified that if $2\vert w[n]\vert<\Delta \theta_{\min}(\alpha,\tD)/2$ (which would imply $\vert e[n]\vert<\Delta \theta_{\min}(\alpha,\tD)/2$) for all $n$, then $\mathbf{\hat{x}}{\hi}^{(n)}=\xhi^{(n)}$. As summarized in Theorem \ref{thm:rob}, Algorithm $2$ can exactly recover the ground truth spikes from measurements corrupted by bounded adversarial noise, the extent of the robustness is determined by the parameters $A,\alpha, \tD$. 
\begin{algorithm}[h]
\caption{Noisy Spike Recovery}
\begin{algorithmic}[1]
\item \textbf{Input:} Measurement $c_e [n]$, Sorted list $\Theta^{\text{sort}}_{\alpha}$ and the corresponding (ordered) spike patterns $\mathcal{S}^{\text{sort}}_{\text{all}}$
\item \textbf{Output:} Decoded spike block $\hat{\mathbf{x}}{\hi}^{(n)}$\\
Set $l\leftarrow 0, u\leftarrow 2^{\tD}-1$\\
\textbf{ while } $u-l > 1$\\
\quad \quad Set $m \leftarrow l + \lfloor( u - l ) / 2\rfloor$\\
\quad \quad \textbf{ if } $\tilde{\theta}_m >c_e[n]$ \textbf{ then }\\
\quad \quad \quad $u\leftarrow m$\\
\quad \quad \textbf{ else }\\
\quad \quad \quad $l\leftarrow m$\\
\quad \quad \textbf{ end if }\\
\textbf{ end while }\\
Find the nearest neighbor 
$
    i^{\star}=\arg\min_{i\in \{l,u\}} (c_e[n]-\tilde{\theta}_i)^2
$\\
Return $\hat{\mathbf{x}}{\hi}^{(n)} \leftarrow \mathbf{\tilde{v}}_{i^{\star}}$
\end{algorithmic}
 \label{alg:LLS}
\end{algorithm}
\begin{thm}\label{thm:rob}
Assume $\alpha \in \mathcal{G}_{\tD}$. Given the ordered set $\Theta^{\text{sort}}_{\alpha}$, the output of Algorithm $2$ with input $c_e[n]$ exactly coincides with the solution of the optimization problem \eqref{eqn:g_nn} in at most $O(\tD)$ steps. Furthermore, if for all $n$, $\vert w[n]\vert < \Delta \theta_{\min}(\alpha,\tD)/4$,  then the output of Algorithm 2 satisfies $\mathbf{\hat{x}}{\hi}^{(n)}=\xhi^{(n)}$.
\end{thm}
\bl{
From Theorem \ref{thm:rob}, it is evident that  $\Delta \theta_{\min}(\alpha,\tD)$ plays an important role in characterizing the upper bound on noise. We attempt to gain insight into how $\Delta \theta_{\min}(\alpha,\tD)$ varies as a function of $\alpha$ when $\tD$ is held fixed.
\begin{lem}\label{lem:d_min_decay}
Given $\textup{D}$, $\Delta \theta_{\min}(\alpha,\tD)=\alpha^{\tD-1}$ for $\alpha\in (0,0.5]$.
\end{lem}
\begin{proof}
The proof is in Appendix C.
\end{proof}}
\bl{When $\alpha\in (0,0.5]$, $\Delta \theta_{\min}(\alpha,\tD)$ is monotonically increasing with $\alpha$. However, for $\alpha>0.5$ the trend {fluctuates with $\alpha$ differently for different $\tD$}, and becomes quite challenging to predict. This is also confirmed by the empirical plot in Fig. \ref{fig:my_label_d}. A refined analysis of $\Delta \theta_{\min}(\alpha,\tD)$ to gain insight into desirable filter parameters $\alpha$ is an interesting direction for future work.}
\vspace{-0.3cm}
\subsection{Trade-off between memory and computational complexity}
A crucial aspect of Algorithms 1 and 2 is that they achieve efficient run-time complexity by leveraging the off-line construction of the sorted list $\Theta^{\text{sort}}_{\alpha}$ and $\mathcal{S}^{\text{sort}}_{\text{all}}$. These lists, each with $2^{\tD}$ elements, need to be stored in memory and made available during run-time. 
Since there is no free lunch, the resulting computational efficiency of $O(\tD)$ at run-time is attained at the expense of the additional memory that is required to store the sorted lists $\Theta^{\text{sort}}_{\alpha},  \mathcal{S}^{\text{sort}}_{\text{all}}$.
\vspace{-0.38cm}
\subsection{Parallelizable Implementation} 

Algorithm $2 \ (\text{also Algo. }1)$ only takes $c_e[n] (c[n])$ as input and returns $\hat{\mf{x}}{\hi}^{(n)}$, and is completely de-coupled from any other $\hat{\mf{x}}{\hi}^{(n^{\prime})}$, $n^{\prime}\neq n$. 
Recall that in reality, we are provided with measurements $z{\lo}[n] (y{\lo}[n])$, and $c_e[n] (\text{respectively }c[n])$ needs to be computed. Due to this de-coupling, we can compute $c_e[n]'s$ in parallel using two consecutive low-rate samples $z{\lo}[n],z{\lo}[n-1]$ and perform a nearest neighbor search without waiting for any previously decoded spikes. Therefore, the total decoding complexity can be further improved depending on the available parallel computing resources. 

\vspace{-0.2cm}
\section{Error Analysis for Gaussian Noise}

Algorithm $2$ solves \eqref{eqn:g_nn} without requiring any knowledge of the noise statistics. However, in order to analyze its performance, we will make the following (standard) assumptions on the statistics of the high-rate spiking signal $\xhi$ and the measurement noise $w[n]$ as follows:
\begin{itemize}
    \item (\textbf{A1}) The entries of the binary vector $\xhi \in \{0,A\}^{L}$ are i.i.d random variables distributed as \bl{$x{\hi}[n]\sim A\text{Bern}(p)$}.
    \item (\textbf{A2}) The additive noise $w[n], 0\leq n \leq M-1$ is independent of $x{\hi}[n]$, and distributed as $w[n]\sim \mathcal{N}(0,\sigma^2)$
\end{itemize}
\vspace{-0.2cm}
\subsection{Probability of Erroneous Decoding}\label{sec:prob}
Under assumption (\textbf{A2}), the ML estimate of $\xhi$ is given by the solution to the following problem:  
\begin{equation*}
    \smash{\hat{\mathbf{x}}_{\text{ML}}=\arg\min_{\mf{v} \in \{0,A\}^{L}} \Vert \mathbf{z}{\lo}-\mathbf{S}_{\tD}\mathbf{G}_{\alpha}\mathbf{v}\Vert_2} \quad (P_{\text{NN}})
\end{equation*}

The proposed Algorithm $2$ does not attempt to solve ($P_{\text{NN}}$), which is computationally intractable. Instead, it solves {a set of $M-1$ one dimensional nearest neighbor search problems, by finding the nearest neighbor of $c_e[n]$ for each $n=1,2,\cdots,M-1$.} This scalar nearest neighbor search is implemented in a computationally efficient manner by using parallel binary search on a pre-sorted list. Notice that by the operation \eqref{eqn:c_op_noise}, the variance of the equivalent noise term $e[n]$ gets amplified by a factor of at most $(1+\alpha^{2\tD})<2$. This can be thought of as a price paid to achieve computational efficiency and parallelizability. The following theorem characterizes the dependence of certain key quantities of interest, such as the signal-to-noise ratio (SNR), undersampling factor $\tD$, and filter's frequency response (controlled by $\alpha$) on the performance of Algorithm $2$. 
\begin{thm}\label{thm:full_prob}
Suppose $\alpha\in \mathcal{G}_{\tD}$ and assumptions \textup{\textbf{(A1-A2)}} hold. Given $\delta>0$, if the following condition is satisfied:
\begin{equation}
    \smash{\Delta \theta^2_{\min}(\alpha,\textup{D})/\sigma^2 \geq 4\ln\left(2M/\delta\right)} \label{eqn:snr_cond}
\end{equation}
then Algorithm $2$ can exactly recover the binary signal $\mf{x}_{\textup{hi}}$ with probability at least $1-\delta$.
\end{thm}
\begin{proof} 
The proof follows standard arguments for computing the probability of error for symbol detection in Gaussian noise, followed by certain simplifications and is included in Appendix $D$ for completeness.
\end{proof}

\begin{figure}
\setlength\belowcaptionskip{-1.4\baselineskip}
    \centering
    \includegraphics[width=0.76\linewidth]{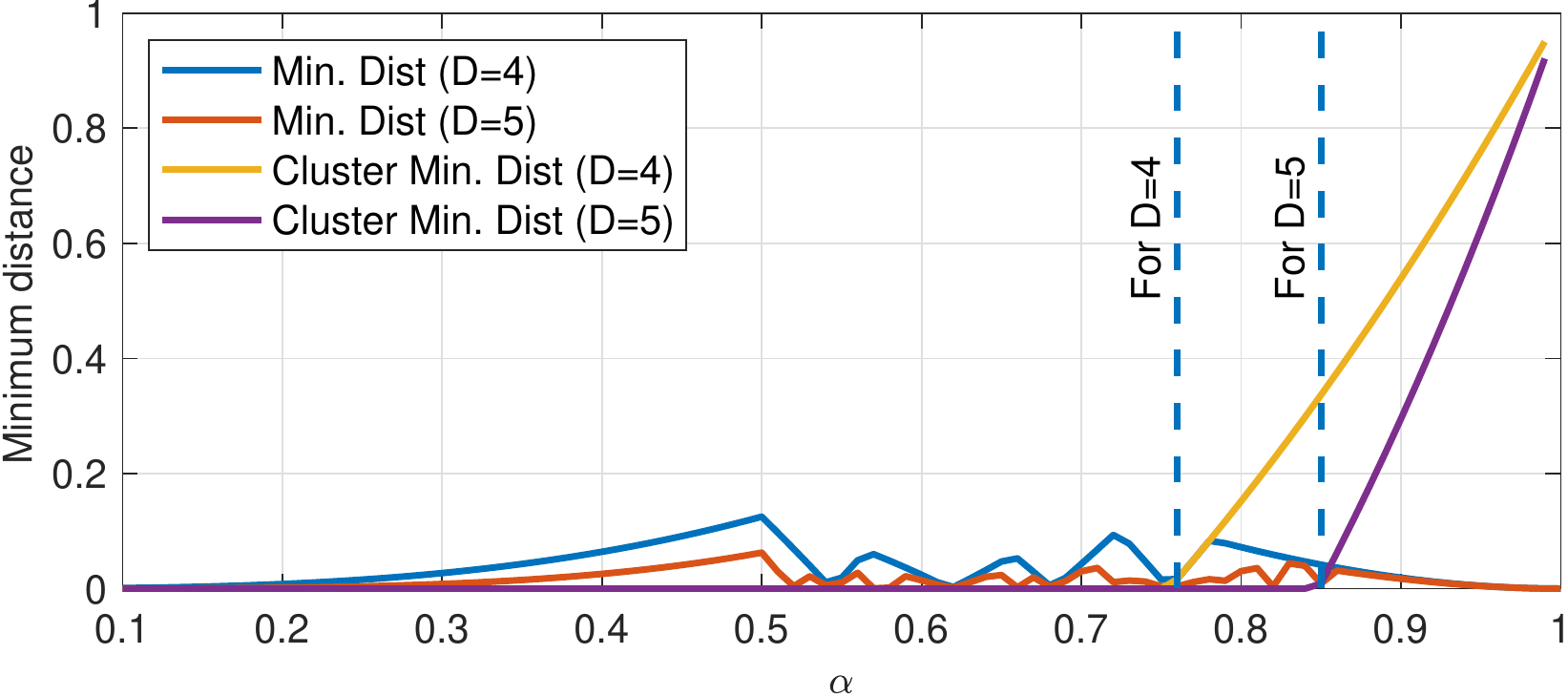}\\
    \includegraphics[width=0.76\linewidth]{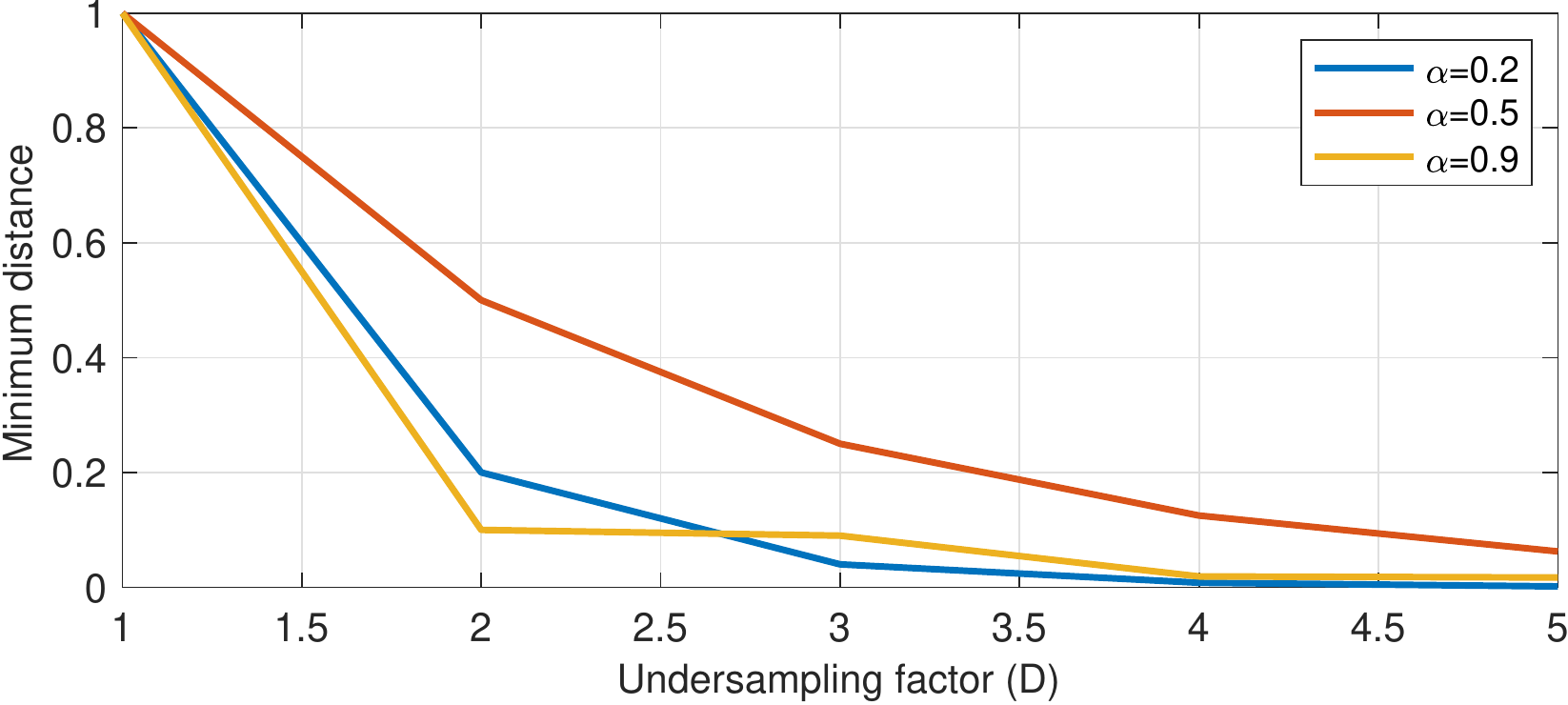}
    \caption{Variation of $\Delta \theta_{\min}(\alpha,\tD)$ as a function of undersampling factor $\tD$ and $\alpha$. The cluster-distance $\Delta^{c}_{\min}(\alpha,\tD)$ vs. $\alpha$ is also overlaid. Each dotted line denotes the start of the interval $\mathcal{F}_{\tD}$.}
    \label{fig:my_label_d}
\end{figure}

In Fig. \ref{fig:my_label_d}, we plot $\Delta \theta_{\min}(\alpha,\tD)$ as a function of $\tD$ for different values of $\alpha$. As expected, $\Delta \theta_{\min}(\alpha,\tD)$ decays as the $\tD$ increases. Understandably, for a fixed $\alpha$, as $\tD$ increases, it becomes harder to recover the spikes exactly, and higher SNR is needed to compensate for the lower sampling rate. This can be interpreted as the price paid for super-resolution in presence of noise. This phenomenon is also reminiscent of the noise amplification effect in super-resolution, where the ability to super-resolve point sources {becomes more severely hindered by noise as the target resolution grid becomes finer\cite{donoho1992superresolution}}. \bl{In Fig. \ref{fig:my_label_d}, we plot $\Delta \theta_{\min}(\alpha,\tD)$ as a function of $\alpha$ and as predicted by Lemma \ref{lem:d_min_decay}, it monotonically increases upto $0.5$, { but for $\alpha > 0.5$}, the behavior becomes much more erratic and a precise characterization becomes challenging. It is to be noted that in Theorem \ref{thm:full_prob}, we aim to {\em exactly recover} $\xhi$. The SNR requirement can be relaxed if our goal is to recover only spike counts instead of the true spikes as discussed in the next subsection. One can define {other notions of approximate recovery, the analysis of which} will be a topic of future research.}
\vspace{-0.6cm}
\bl{\subsection{Relaxed Spike reconstruction: Count Estimation}}
\label{sec:count_rob}
\bl{As shown in Theorem \ref{thm:rob}, exact recovery of spikes is possible under somewhat restrictive condition on the noise in terms of $\Delta \theta_{\min}(\alpha,\tD)$, which becomes quite small as $\tD$ increases. This naturally calls for other relaxed notions of recovery which {can handle larger noise levels}. In neuroscience, it is believed that information is encoded as either the spike timing (temporal code) or the firing rates (rate
coding) of individual neurons in the brain. Therefore, the spike counts over an interval can be informative to understand neural functions, {even when it is impossible to temporally localize the neural spikes}.  For example, neurons in the visual cortex encode stimulus orientations as their firing rates \cite{hubel1959receptive}. We will therefore focus on spike count as an approximate recovery metric, which concerns estimating the number of spikes occurring between two consecutive low-rate measurements instead of resolving the individual spiking activity at a higher resolution.} 

\bl{Let $\gamma[n]$ denote the total number of spikes occurring between two consecutive low-rate samples $z\lo[n]$ and $z\lo[n-1]$. Since $\xhi$ and its estimate $\mf{\widehat{x}}_{\text{hi}}$ are both binary valued (amplitude $A$), the true spike count ($\gamma[n]$) and estimated count ($\hat{\gamma}[n]$) are given by: 
 $
     \gamma[n]=\Vert\xhi^{(n)}\Vert_0, \quad \hat{\gamma}[n]=\Vert\hat{\mathbf{x}}_{\text{hi}}^{(n)}\Vert_0,\ n=1,\cdots,M-1,
 $
$\gamma[0]=x_{\text{hi}}[0]/A$  and $\hat{\gamma}[0]=\hat{x}_{\text{hi}}[0]/A$ since the first block is of size $1$ as described in \eqref{eqn:block_x}.
Define a set $\mathcal{C}_{k}^{\tD}$ as:
\begin{align*}
\mathcal{C}_k^{\tD}:=\{\mathbf{v}\in \{0,A\}^{\tD},  \Vert\mathbf{v}\Vert_0=k\}, \quad 0\leq k \leq \tD
\end{align*} 
It is a collection of all binary vectors (of length $\tD$) with spike count $k$. The ground truth spike block belongs to $\mathcal{C}_{\gamma[n]}^{\tD}$. Any element from $\mathcal{C}_{\gamma[n]}^{\tD}$ will give the true spike count. {Hence, exact recovery of count can be possible even when spikes cannot be recovered. }}

\bl{For a fixed $\tD$, we define a set of $\alpha$ denoted by $\mathcal{F}_{\tD}$: 
\begin{align}
\mathcal{F}_{\tD}:=\{\alpha\in (0,1)\vert\alpha^{\tD}-\alpha^{\tD-k_0-1}-\alpha^{k_0}+1<0\}  \label{eqn:count_int}
\end{align}
where {$k_0=\lfloor\tD/2\rfloor$}. We will obtain a sufficient condition for robust spike count estimation {when $\alpha \in \mathcal{F}_{\tD}$.} It can be shown that for any $\tD$, $\mathcal{F}_{\tD}$ will always be non-empty. 
Define
\begin{align}
    \theta^{k}_{\min}:=\min_{\mf{u}\in \mathcal{C}_{k}^{\tD}} \mf{h}_{\alpha}^{\top}\mf{u} \quad \theta^{k}_{\max}:=\max_{\mf{u}\in \mathcal{C}_{k}^{\tD}} \mf{h}_{\alpha}^{\top}\mf{u} \label{eqn:theta_max_min}
\end{align}
Observe that if 
\begin{equation}
    {\theta^{k+1}_{\min} > \theta^{k}_{\max}, k=0,1,\cdots,\tD-1} \label{eqn:cluster}
\end{equation}
then all spike patterns $\mathbf{u}_i\in \mathcal{C}_{k}^{\tD}$ (with the same spike count $k$)  are clustered together when mapped on to the real line by the transformation $\mf{h}_{\alpha}^{\top}\mathbf{u}$ as shown in Figure \ref{fig:my_label}.  
{ When \eqref{eqn:cluster} holds}, we can define a ``cluster-restricted minimum distance" as:
\begin{equation}
    \smash{\Delta_{\min}^{c}(\alpha,\tD):= \min_{0\leq k \leq \tD-1} \theta^{k+1}_{\min}-\theta^{k}_{\max}} \label{eqn:cluster_d_min}
\end{equation}
Given a noisy observation $c_e [n]=\mf{h}_{\alpha}^{\top}\mf{x}\hi^{(n)}+e[n]$,  the solution to the nearest neighbor problem \eqref{eqn:g_nn} may return an incorrect neighbor $\theta_{j}\neq \mf{h}_{\alpha}^{\top}\mf{x}\hi^{(n)}$. However, { when \eqref{eqn:cluster} holds and if the noisy observation satisfies the following conditions}:
\begin{align}
 {(\theta^{\gamma[n]}_{\min}+\theta^{\gamma[n]-1}_{\max})/2<  c_e[n] < (\theta^{\gamma[n]+1}_{\min}+\theta^{\gamma[n]}_{\max})/2} \label{eqn:band_count}
\end{align}
{then the nearest-neighbor decision rule in Algorithm $2$ will still ensure that $\theta_j \in \mathcal{C}_{\gamma[n]}^{\tD}$.} This has also been visualized in Fig. \ref{fig:my_label} where each colored band represents the ``safe-zone" for each count and the black dotted-line denotes the boundary. {This will result in correct identification of the spike count but will incur error in terms of spiking pattern.} We formally summarize this in the following Theorem that provides robustness guarantee for exact count recovery from measurements corrupted by adversarial noise (similar to Theorem $2$ for spike recovery). 
\begin{thm}\label{thm:spk_count}
Assume $\alpha \in \mathcal{F}_{\tD}$. Given the ordered set $\Theta^{\text{sort}}_{\alpha}$, let $\hat{\gamma}[n]$ be the estimated spike count obtained from Algorithm $2$ with input $c_e[n]$. If for all $n$, $\vert w[n]\vert < \Delta_{\min}^{c}(\alpha,\tD)/4$,  then the count can be exactly recovered, i.e., $\hat{\gamma}[n]=\gamma[n]$.
\end{thm}
\begin{proof}
Proof is in Appendix E.
\end{proof}
{It is clear that when \eqref{eqn:cluster} holds, $\Delta_{\min}^{\text{c}}(\alpha,\tD)$ is no smaller than $\Delta\theta_{\min}(\alpha,\tD)$, since the former is computed over neighboring elements of the cluster whereas $\Delta\theta_{\min}(\tD,\alpha)$ computes the minimum distance over all consecutive elements (both inter-cluster as well as intra-cluster) in $\Theta_{\alpha}^{\text{sort}}$.} This essentially suggests that estimation of counts (for this range of $\alpha$ and $\tD$) can be more robust compared to inferring the individual spiking patterns. {We also illustrate this numerically in Figure \ref{fig:my_label_d} (top), where we plot both $\Delta^{c}_{\min}$ and $\Delta \theta_{\min}$ as a function of $\alpha$ and the start of the interval $\mathcal{F}_{\tD}$ (computed numerically) is denoted using dotted lines. For both values of $\tD$, we can see that $\Delta^{c}_{\min}>\Delta \theta_{\min}$ and the gap grows as $\alpha$ increases.}}
\begin{figure}
\setlength\belowcaptionskip{-1.3\baselineskip}
    \centering
    \includegraphics[width=0.8\linewidth]{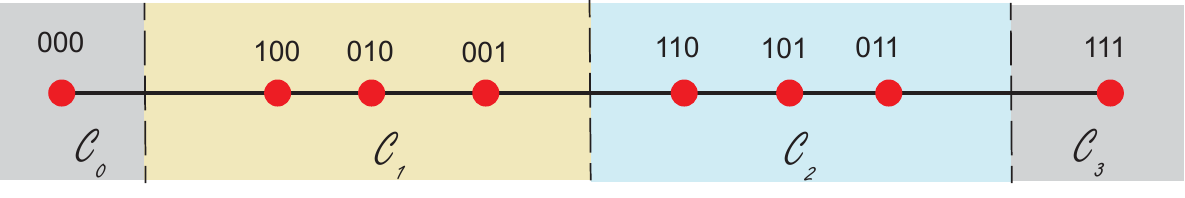}
    \caption{Visualization of the sets $\mathcal{C}_k^{\tD}$ for $\tD=3$. In this scenario, the spiking patterns corresponding to the same count are clustered together and hence, are favorable for robust count estimation.}
    \label{fig:my_label}
\end{figure}
\vspace{-0.3cm}
\section{Numerical Experiments}
We conduct numerical experiments to evaluate the performance of the proposed super-resolution spike decoding algorithm on both synthetic and real calcium imaging datasets. 
\vspace{-0.5cm}
\begin{figure}[h]
\setlength\belowcaptionskip{-1.4\baselineskip}
    \centering
     \includegraphics[width=0.76\linewidth]{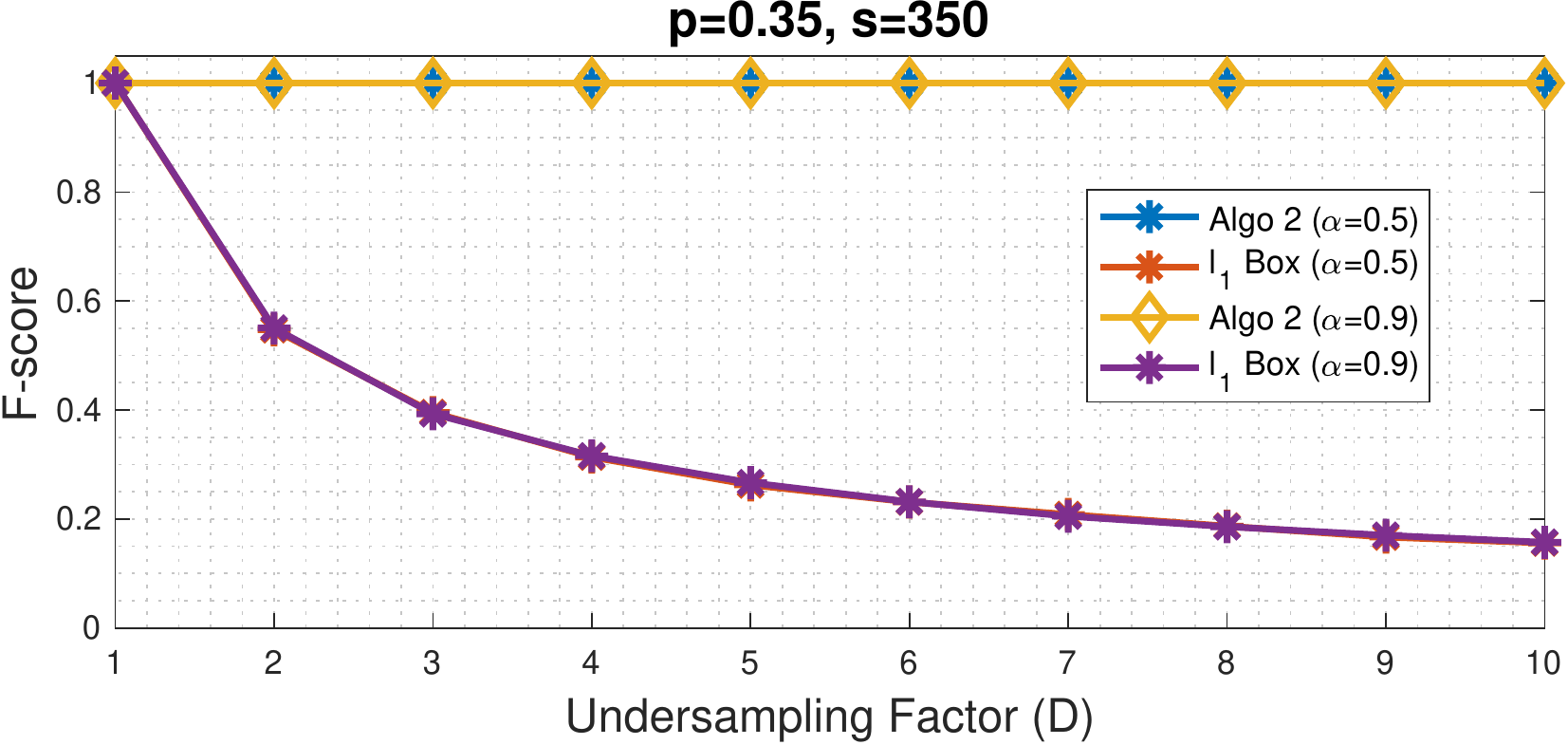}\\
     \vspace{0.1cm}
     \includegraphics[width=0.76\linewidth]{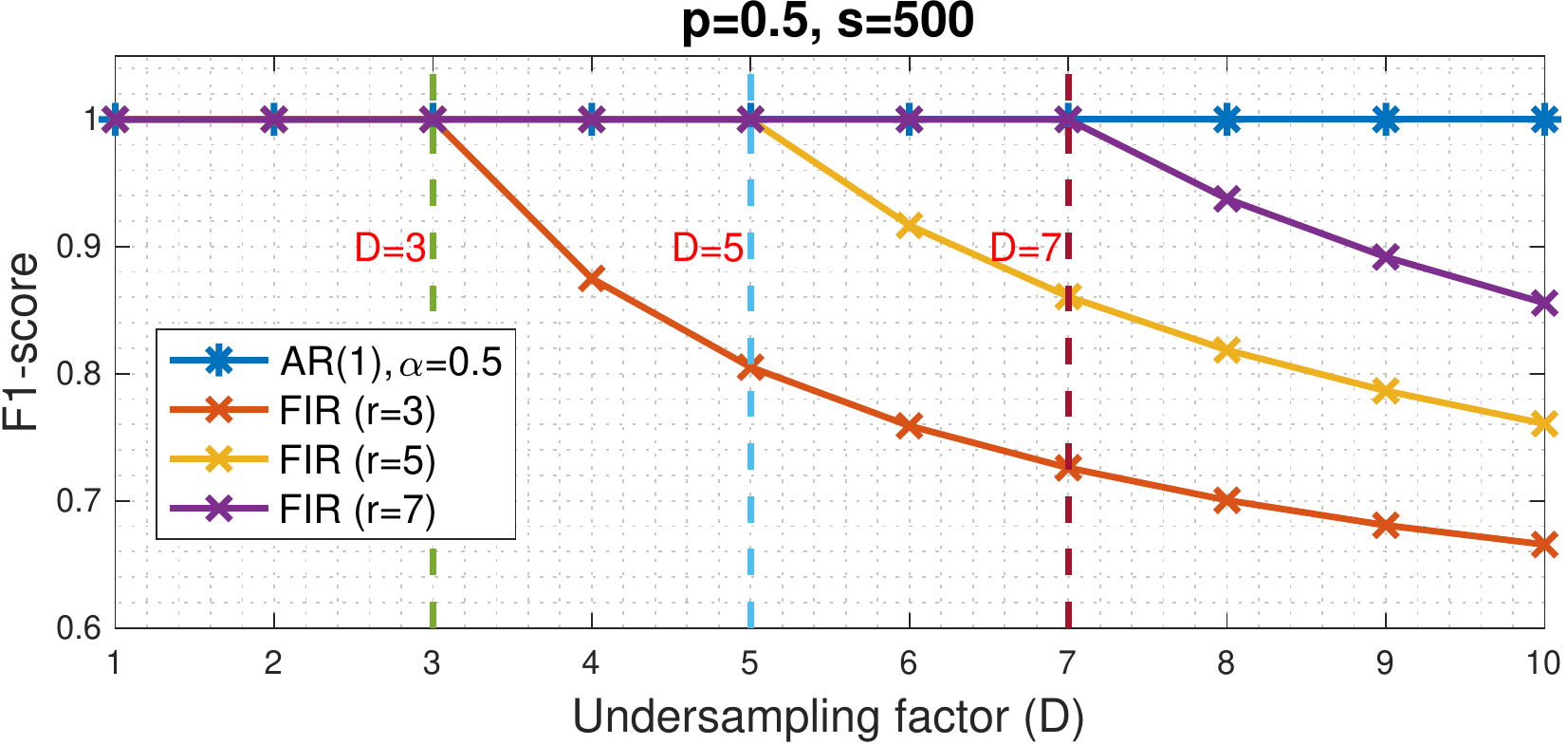}
     \vspace{0.2cm}
    \caption{(Top) Quantitative comparison of Algorithm $2$ against box-constrained $l_1$ minimization method with noiseless measurements (with tolerance $t_0=0$). (Bottom) (Role of Filter Memory): Average F-score vs. $\tD$ for FIR and IIR (AR(1)) filters. Each dotted line indicates the corresponding theoretical transition point ($\tD=r$).}
    \label{fig:noiseless_m}
\end{figure}
\begin{figure*}[h]
\setlength\belowcaptionskip{-1.3\baselineskip}
\centering
\resizebox{!}{9.0cm}{
\begin{tikzpicture}
\begin{scope}[spy using outlines={rectangle,magnification=2, width=5cm, height=9cm,,connect spies}]
    \node (image1) at (0,0){\includegraphics[width=0.5\linewidth]{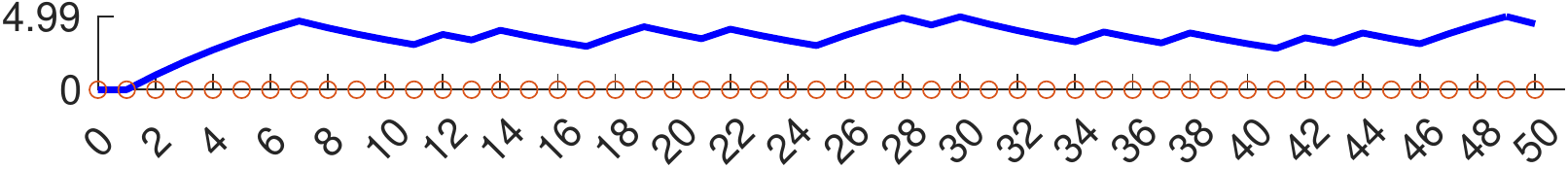}};
    \coordinate (a) at (-4.4,0.2);
    \node[label={left:{\scriptsize$y{\hi}[n]$}}] at (a)  {};
    \coordinate (aa) at (0.5,1);
    \node[label={left:{$\tD=5$}}] at (aa)  {};
    \coordinate (aab) at (5,1.7);
    \node[label={south:{\large (Top)}}] at (aab)  {};
    \coordinate (aab2) at (5,-4.7);
    \node[label={south:{\large (Bottom)}}] at (aab2)  {};
    \node (image2) at (0,-1.1){\includegraphics[width=0.5\linewidth]{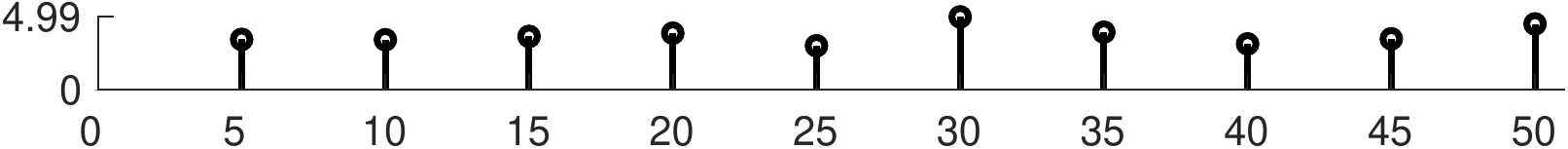}};
    \coordinate (b) at (-4.4,-1.0);
    \node[label={left:{\scriptsize$y{\lo}[n]$}}] at (b)  {};
    \node (image3) at (0,-2.2){\includegraphics[width=0.5\linewidth]{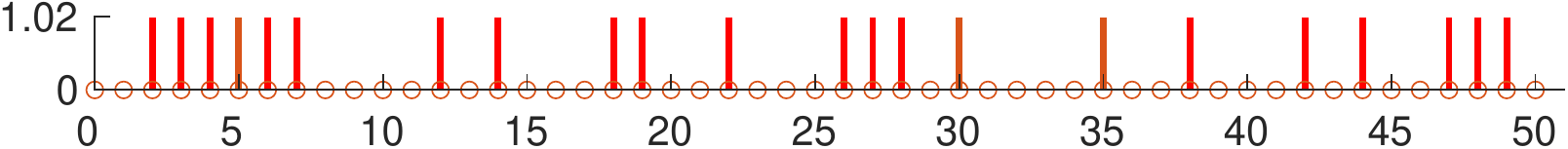}};
     \coordinate (c) at (-4.4,-2.1);
    \node[label={left:{\scriptsize$x{\hi}[n]$}}] at (c)  {};
    \node (image3) at (0,-3.3){\includegraphics[width=0.5\linewidth]{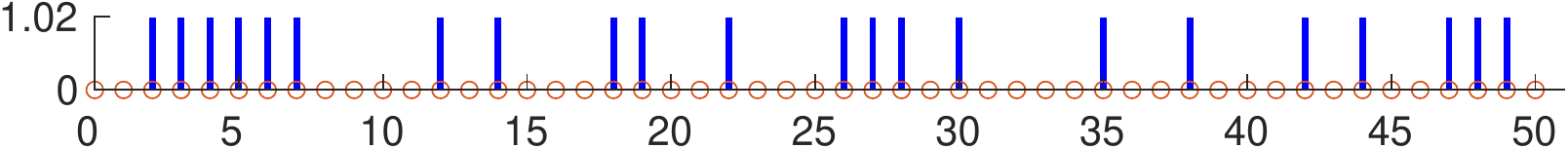}};
     \coordinate (d) at (-4.4,-3.1);
    \node[label={left:{\scriptsize$\hat{x}{\hi}[n]$}}] at (d)  {};
    \node (image) at (0,-4.4){\includegraphics[width=0.5\linewidth]{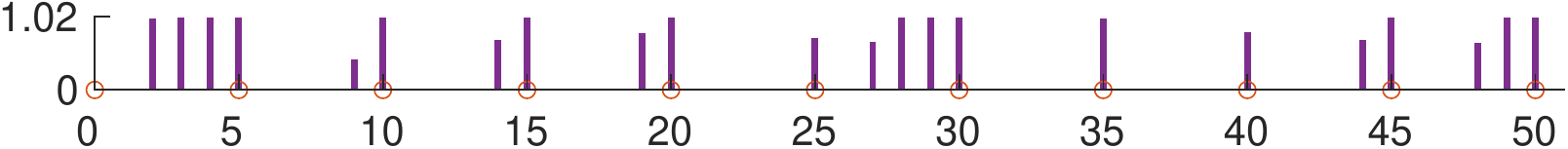}};
    \coordinate (e) at (-4.4,-4.3);
    \node[label={left:{\scriptsize$\hat{x}_{\text{l}_1}[n]$}}] at (e)  {};
    \coordinate (D1) at (-2.3,0.8);
    \coordinate (D2) at (-2.3,-5);
    \draw [red,dashed] (D1) -- (D2);
    
    \coordinate (D1h) at (-3.14,0.8);
    \coordinate (D2h) at (-3.14,-5);
    \draw [red,dashed] (D1h) -- (D2h);
    
    \coordinate (D3p) at (-1.475,0.8);
    \coordinate (D4p) at (-1.475,-5);
    \draw [red,dashed] (D3p) -- (D4p);
    
    \coordinate (D3) at (-0.63,0.8);
    \coordinate (D4) at (-0.63,-5);
    \draw [red,dashed] (D3) -- (D4);

    \coordinate (D3h) at (  0.19,0.8);
    \coordinate (D4h) at (  0.19,-5);
    
    \draw [red,dashed] (D3h) -- (D4h);
    
    \coordinate (D5) at (1.02,0.8);
    \coordinate (D6) at (1.02,-5);
    \draw [red,dashed] (D5) -- (D6);

    \coordinate (D5p) at (1.85,0.8);
    \coordinate (D6p) at (1.85,-5);
    \draw [red,dashed] (D5p) -- (D6p);
    
    \coordinate (D7) at (2.68,0.8);
    \coordinate (D8) at (2.68,-5);
    \draw [red,dashed] (D7) -- (D8);

    \coordinate (D7p) at (3.52,0.8);
    \coordinate (D8p) at (3.52,-5);
    \draw [red,dashed] (D7p) -- (D8p);
    \node (image11) at (10.5,0){\includegraphics[width=0.5\linewidth]{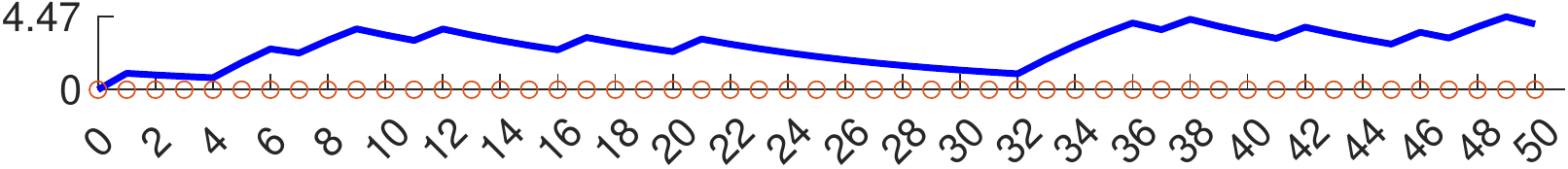}};
    \coordinate (ba) at (6.2,0.2);
    \node[label={left:{\scriptsize$y{\hi}[n]$}}] at (ba)  {};
    \coordinate (baa) at (9.9,1);
    \node[label={left:{$\tD=10$}}] at (baa)  {};
    \node (image21) at (10.5,-1.1){\includegraphics[width=0.5\linewidth]{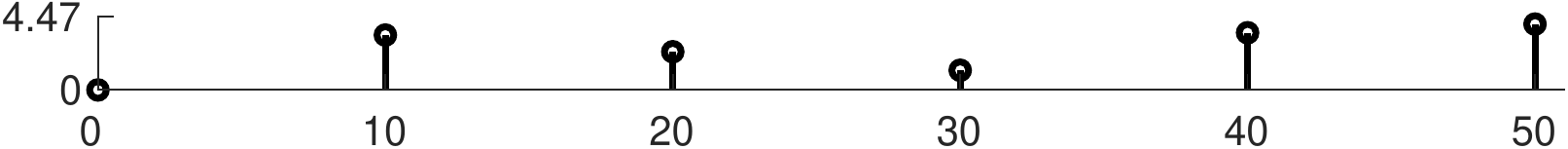}};
    \coordinate (bb) at (6.2,-1.0);
    \node[label={left:{\scriptsize$y{\lo}[n]$}}] at (bb)  {};
    \node (image31) at (10.5,-2.2){\includegraphics[width=0.5\linewidth]{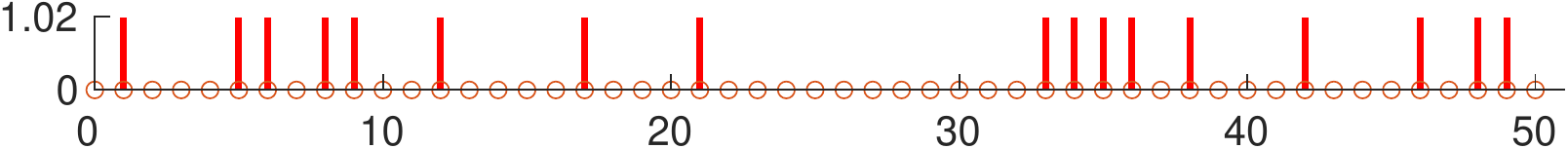}};
     \coordinate (bc) at (6.2,-2.1);
    \node[label={left:{\scriptsize$x{\hi}[n]$}}] at (bc)  {};
    \node (image31) at (10.5,-3.3){\includegraphics[width=0.5\linewidth]{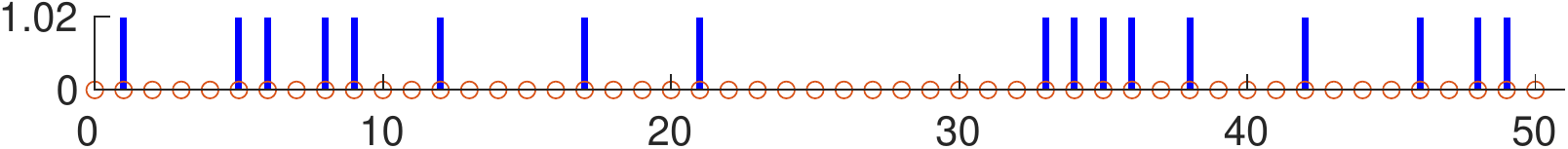}};
     \coordinate (bd) at (6.2,-3.1);
    \node[label={left:{\scriptsize$\hat{x}{\hi}[n]$}}] at (bd)  {};
    \node (imageb) at (10.5,-4.4){\includegraphics[width=0.5\linewidth]{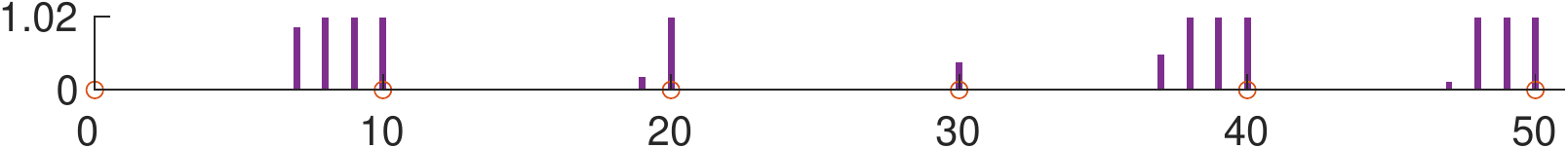}};
    \coordinate (be) at (6.2,-4.3);
    \node[label={left:{\scriptsize$\hat{x}_{\text{l}_1}[n]$}}] at (be)  {};
    \coordinate (D11) at (8.21,0.8);
    \coordinate (D21) at (8.21,-5);
    \draw [red,dashed] (D11) -- (D21);
    \coordinate (D31) at (9.85,0.8);
    \coordinate (D41) at (9.85,-5);
    
    \draw [red,dashed] (D31) -- (D41);
    
    \coordinate (D51) at (11.54,0.8);
    \coordinate (D61) at (11.54,-5);
    
    \draw [red,dashed] (D51) -- (D61);
    
    \coordinate (D71) at (13.18,0.8);
    \coordinate (D81) at (13.18,-5);
    \draw [red,dashed] (D71) -- (D81);
\end{scope}
\begin{scope}[spy using outlines={rectangle,magnification=2, width=5cm, height=9cm,,connect spies}]
     \node (image2) at (0,-6.0){\includegraphics[width=0.5\linewidth]{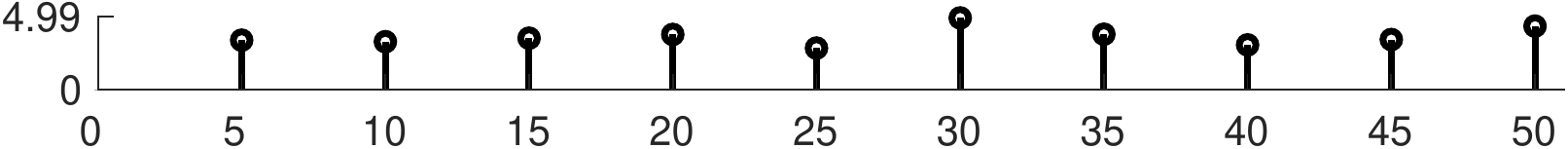}};
    \coordinate (b) at (-4.4,-5.8);
    \node[label={left:{\scriptsize$y{\lo}[n]$}}] at (b)  {};
     \node (image3) at (0,-6.9){\includegraphics[width=0.5\linewidth]{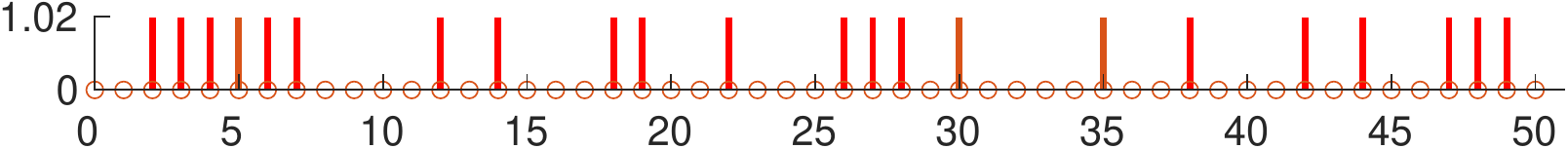}};
     \coordinate (c) at (-4.4,-6.8);
    \node[label={left:{\scriptsize$x{\hi}[n]$}}] at (c)  {};
    \node (image3) at (0,-8.0){\includegraphics[width=0.5\linewidth]{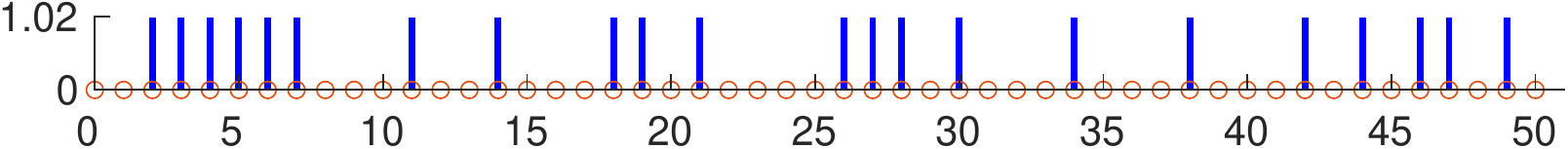}};
     \coordinate (d) at (-4.4,-7.9);
    \node[label={left:{\scriptsize$\hat{x}{\hi}[n]$}}] at (d)  {};
    \node (image) at (0,-9.2){\includegraphics[width=0.5\linewidth]{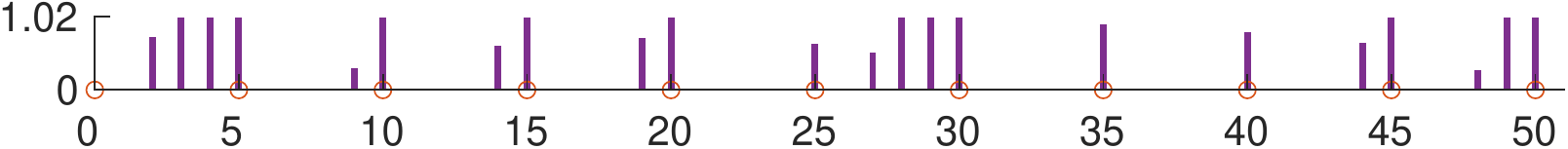}};
    \coordinate (e) at (-4.4,-9.1);
    \node[label={left:{\scriptsize$\hat{x}_{\text{l}_1}[n]$}}] at (e)  {};
  \coordinate (D1) at (-2.31,-5.4);
    \coordinate (D2) at (-2.31,-9.8);
    \draw [red,dashed] (D1) -- (D2);
    
    \coordinate (D1h) at (-3.15,-5.4);
    \coordinate (D2h) at (-3.15,-9.8);
    \draw [red,dashed] (D1h) -- (D2h);
    
    \coordinate (D3p) at (-1.475,-5.4);
    \coordinate (D4p) at (-1.475,-9.8);
    \draw [red,dashed] (D3p) -- (D4p);
    
    \coordinate (D3) at (-0.64,-5.4);
    \coordinate (D4) at (-0.64,-9.8);
    \draw [red,dashed] (D3) -- (D4);

    \coordinate (D3h) at (  0.16,-5.4);
    \coordinate (D4h) at (  0.16,-9.8);

    \draw [red,dashed] (D3h) -- (D4h);
    
    \coordinate (D5) at (0.99,-5.4);
    \coordinate (D6) at (0.99,-9.8);
     
    \draw [red,dashed] (D5) -- (D6);

    \coordinate (D5p) at (1.85,-5.4);
    \coordinate (D6p) at (1.85,-9.8);
    \draw [red,dashed] (D5p) -- (D6p);
    
    \coordinate (D7) at (2.65,-5.4);
    \coordinate (D8) at (2.65,-9.8);
    \draw [red,dashed] (D7) -- (D8);

    \coordinate (D7p) at (3.52,-5.4);
    \coordinate (D8p) at (3.52,-9.8);
    \draw [red,dashed] (D7p) -- (D8p);
    \node (image21) at (10.5,-6.0){\includegraphics[width=0.5\linewidth]{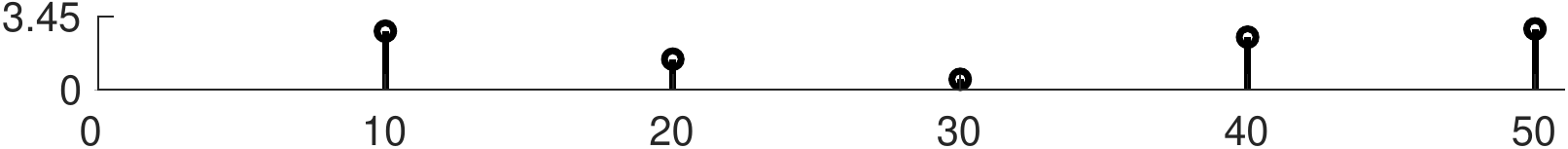}};
    \coordinate (bb) at (6.2,-5.8);
    \node[label={left:{\scriptsize$y{\lo}[n]$}}] at (bb)  {};
    \node (image31) at (10.5,-6.9){\includegraphics[width=0.5\linewidth]{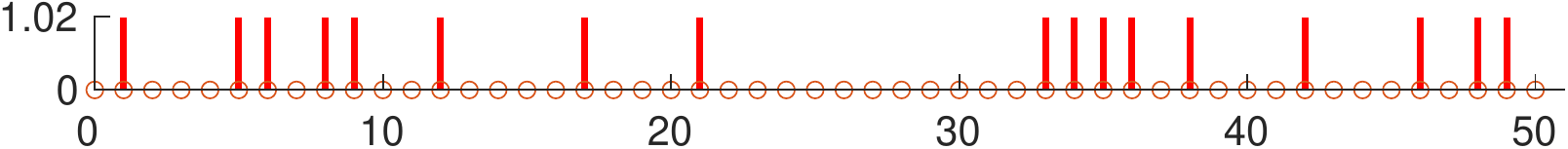}};
     \coordinate (bc) at (6.2,-6.9);
    \node[label={left:{\scriptsize$x{\hi}[n]$}}] at (bc)  {};
    \node (image31) at (10.5,-8.0){\includegraphics[width=0.5\linewidth]{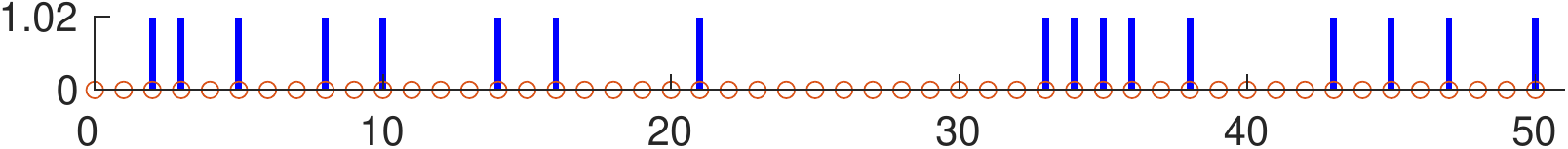}};
     \coordinate (bd) at (6.2,-7.9);
    \node[label={left:{\scriptsize$\hat{x}{\hi}[n]$}}] at (bd)  {};
    \node (imageb) at (10.5,-9.1){\includegraphics[width=0.5\linewidth]{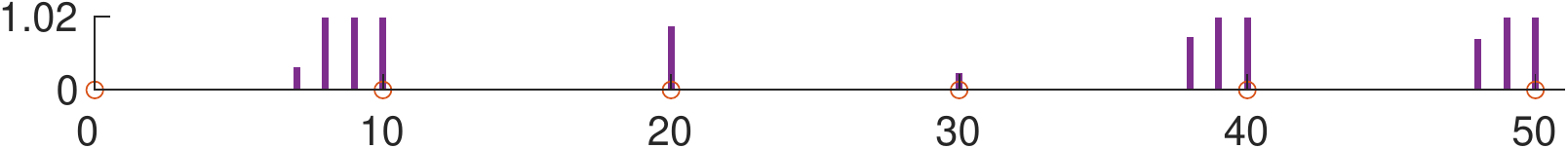}};
    \coordinate (be) at (6.2,-8.9);
    \node[label={left:{\scriptsize$\hat{x}_{\text{l}_1}[n]$}}] at (be)  {};
        \coordinate (D11) at (8.17,-5.4);
    \coordinate (D21) at (8.17,-9.8);
    \draw [red,dashed] (D11) -- (D21);
    \coordinate (D31) at (9.85,-5.4);
    \coordinate (D41) at (9.85,-9.8);
    
    \draw [red,dashed] (D31) -- (D41);
    
    \coordinate (D51) at (11.52,-5.4);
    \coordinate (D61) at (11.52,-9.8);
    
    \draw [red,dashed] (D51) -- (D61);
    
    \coordinate (D71) at (13.15,-5.4);
    \coordinate (D81) at (13.15,-9.8);
    \draw [red,dashed] (D71) -- (D81);
    \coordinate (aab3) at (-3.5,-9.7);
    \node[label={south:{$x{\hi}[n]$: Ground Truth Spikes,}}] at (aab3)  {};
    \coordinate (aab33) at (3.5,-9.7);
    \node[label={south:{$\hat{x}{\hi}[n]$: Output of Algorithm $2$, {$\hat{x}_{\text{l}_1}[n]$}: Output of $l_1$ minimization,}}] at (aab33)  {};
    \coordinate (aab4) at (12.4,-9.7);
    \node[label={south:{$y{\hi}[n]$: High rate waveform,\hspace{0.1cm} $y{\lo}[n]$: Low rate samples }}] at (aab4)  {};
\end{scope}
\end{tikzpicture}
}
\vspace{0.05cm}
\vspace{-0.05cm}
\caption{Qualitative comparison of Algorithm $2$ and box-constrained $l_1$ minimization on simulated data. For each simulation noisy measurements are generated with $\alpha=0.9$ such that the noise realization (Top) obeys the bound $\vert w[n]\vert\leq \Delta\theta_{\min}$ (from Theorem \ref{thm:rob}) and (Bottom) violates the bound. For larger noise (Bottom), the spike recovery is imperfect but the spike count can still be exactly recovered using Algorithm $2$.}
\label{fig:my_simulated}
\end{figure*}
\subsection{Synthetic Data Generation and Evaluation Metrics}
We create a synthetic dataset by generating high-rate binary spike sequence $\xhi \in \{0,1\}^{L}$ ($A=1$ and $L=1000$) that satisfies assumption (\textbf{A1}). The spiking probability $p$ controls the average sparsity level given by $s:=\mathbb{E}[\Vert \xhi \Vert_0]=Lp$. We aim to reconstruct $\xhi$ from $M\approx L/\tD$ low-rate measurements $z{\lo}[n]$ defined in \eqref{eqn:noisy_measurement}. Notice that we operate in a regime where 
the expected sparsity is greater than the total number of low-rate measurements, i.e., $s>M$. We employ the widely-used \emph{F-score} metric to evaluate the accuracy of spike detection \cite{schiebinger2017superresolution,deneux2016accurate}. The F-score is computed by first matching the estimated and ground truth spikes. An estimated spike is considered a ``match" to a ground truth spike if it is within a distance of $t_0$ of the ground truth (many-to-one matching is not allowed) \cite{schiebinger2017superresolution,deneux2016accurate}. Let $K$ and $K'$ be the total number of ground truth and estimated spikes, respectively. The number of spikes declared as true positives is denoted by $T_p$. After the matching procedure, we compute the recall $(R=\frac{T_p}{K})$ which is defined as the ratio of true positives ($T_p$) and the total number of ground truth spikes ($K$). Precision ($P=\frac{T_p}{K^{\prime}}$) measures the fraction of the total detected spikes which were correct. Finally, the F-score is given by the harmonic mean of recall and precision $\text{F-score} = 2PR/(P+R)$.

\vspace{-0.4cm}
\subsection{Noiseless Recovery: Role of Binary priors and memory}
We first consider the noiseless setting ($w[n]=0$ in \eqref{eqn:noisy_measurement}). We compare the performance of Algorithm $2$ against \bl{box-constrained $l_1$ minimization method \cite{stojnic2010recovery,keiper2017compressed}}, where we solve:
\begin{equation*}
    \smash{\min_{\mathbf{x}\in \mathbb{R}^L}\ \Vert \mathbf{x} \Vert_1 \text{ s.t. } \Vert\mathbf{y}{\lo}-\mathbf{S}_{\tD}\mf{G}_{\alpha} \mathbf{x}\Vert_2\leq \epsilon, \mf{0}\leq \mathbf{x} \leq A\mathbf{1}} \tag{P1}\label{eqn:P1_noisy}
\end{equation*}
\bl{For synthetic data, $\epsilon$ is chosen using the norm of the noise term $\Vert \mf{w}\Vert_2$. This \emph{oracle} choice ensures most favorable parameter tuning for the \eqref{eqn:P1_noisy}, although a more realistic choice would be to set $\epsilon=\sqrt{M}\sigma$ according to the noise power ($\sigma$).} In the noiseless setting, we choose $\epsilon=0$. \bl{The problem \eqref{eqn:P1_noisy} is a standard convex relaxation of \eqref{eqn:P0} which promotes sparsity as well as tries to impose the binary constraint via the box-relaxation (introduced in Section II-C).} In Fig. \ref{fig:noiseless_m} (Top), we plot the F-score ($t_0=0$) as a function of $\tD$. As can be observed, Algorithm $2$ consistently achieves an F-score of $1$, whereas the F-score of $l_1$ minimization shows a decay as $\tD$ increases.  \bl{This confirms Lemma \ref{lem:sp_l1_box} that for $\tD>1$, using box-constraints with $l_1$ norm minimization  is not enough to enable exact recovery from low rate measurements.} \bl{In absence of noise, the performance of Algorithm $2$ is not affected by the filter parameter $\alpha$ as shown in Fig. \ref{fig:noiseless_m} (Top).} 

Next, we compare the reconstruction from the decimated output of (i) an AR(1) filter and (ii) an FIR filter of length $r$ driven by the same input $\xhi\in \{0,1\}^{1000}$. We choose the FIR filter $\mathbf{h}=[1,\alpha,\cdots,\alpha^{r-1}]^{\top}$ (truncation of the IIR filter) with $\alpha=0.5$.  Algorithm $2$ is applied to the low-rate AR(1) measurements, whereas the algorithm proposed in \cite{sarangi2021no} is used for the FIR case. \bl{The algorithm applied for the FIR case can provably operate with the optimal number of measurements when $\alpha=0.5$ and hence, we chose this specific value for the filter parameter.}
In Figure \ref{fig:noiseless_m} (Bottom), \bl{we again compare the average F-score as a function of $\tD$, averaged over $10000$ Monte Carlo runs, for $p=0.5$}. As predicted by Lemma \ref{lem:fir}, despite utilizing binary priors, the error for the FIR filter shows a phase transition when $\tD>r$. This demonstrates the critical role played by the infinite memory of the AR(1) filter in achieving exact recovery with arbitrary $\tD$.  
\vspace{-0.3cm}
\subsection{Performance of noisy spike decoding}
We generate noisy measurements of the form \eqref{eqn:noisy_measurement}, where $w[n]$ and $x{\hi}[n]$ satisfy assumptions \textbf{(A1-A2)}. 
We illustrate some representative examples of recovered spikes on synthetic data. In Fig. \eqref{fig:my_simulated}, we display the recovered super-resolution estimates on synthetically generated measurements for two undersampling factors $\tD=5 \ (\text{left}),10\ (\text{right})$. For each $\tD$, the top plots show the spikes recovered using Algorithm $2$ and \bl{$l_1$ minimization with box-constraint where the noise realization obeys the bound in Theorem \ref{thm:rob}, while the bottom plots show the same for noise realization violating the bound. The output of $l_1$ minimization with box-constraint is inaccurate, and the spikes are clustered towards the end of each block of length $\tD$. This bias is consistent with the prediction made by our theoretical results in Lemma \ref{lem:sp_l1_box}. When the noise is small enough (top), Algorithm $2$ exactly decodes the spikes, including the ones occurring between two consecutive low-rate samples as predicted by Theorem \ref{thm:rob}. In presence of larger noise (violating the bound), the spikes estimated using $l_1$ minimization continue to be biased to be clustered towards the end of the block. Although the spikes recovered using Algorithm $2$ are not exact, {most of the detected spikes are within a tolerance window of ground truth spikes}. In fact, the spike count estimation is perfect as predicted by Theorem \ref{thm:spk_count}.}
\begin{figure}[h]
    \centering
    \includegraphics[width=0.75\linewidth]{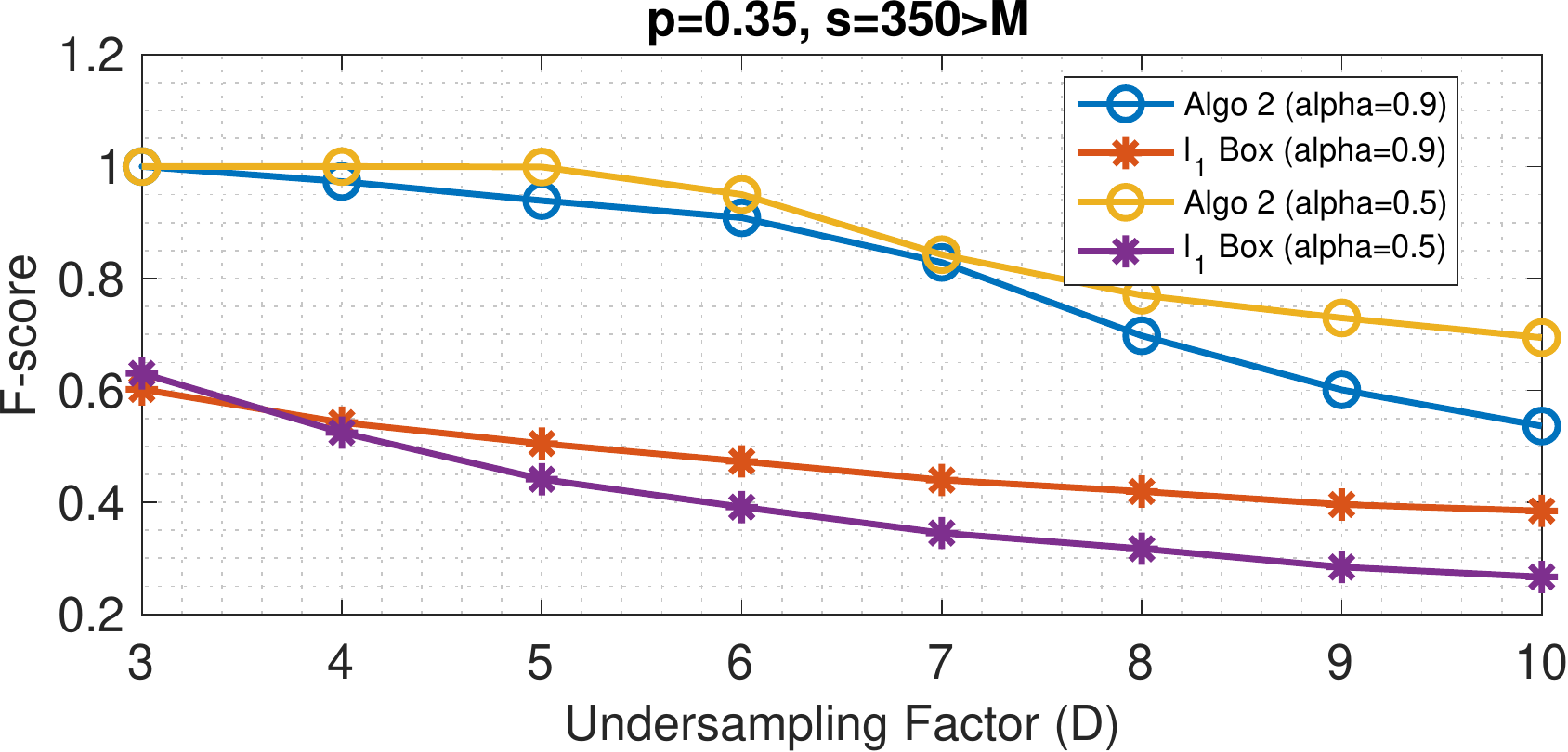}\\ 
    \caption{Spike detection performance with noisy measurements. (Top) F-score vs. $\tD$ for different  filter parameters $\alpha$ ($\sigma=0.01$). Here, $L=1000$ and expected sparsity $s=350$ where we operate in the regime $s>M$. The F-score is computed with a tolerance of $t_0=2$.}
    \label{fig:my_label_m_vary}
\end{figure}
We next quantitatively evaluate the performance in presence of noise, where the metrics are computed with $t_0=2$. \bl{In Fig. \ref{fig:my_label_m_vary} (Top), we plot the F-score as a function of $\tD$ for different values of $\alpha$.} For a fixed $\alpha$, the F-score of both methods decays with increasing $\tD$, but Algorithm $2$ consistently attains a higher F-score compared to $l_1$ minimization. \bl{We observe that $\alpha=0.5$ leads to a higher F-score potentially due to having a larger $\Delta \theta_{\min}(\alpha,\tD)$ compared to $\alpha=0.9$.} 
Next, in Fig. \ref{fig:my_label_prob}, we study the behavior of spike detection as a function of the spiking probability $p$, while keeping $\tD$ fixed at $\tD=5$. When $\sigma$ is fixed, the performance trend is not significantly affected by the spiking probability. At first, this may seem surprising as the expected sparsity is growing while the number of measurements is unchanged. However, since our algorithm exploits the binary nature of the spikes (and not just sparsity), it can handle larger sparsity levels. The spikes reconstructed using $l_1$ minimization achieve a much lower F-score than Algorithm $2$ since the former fails to succeed when the sparsity is large. As expected, smaller $\sigma$ leads to higher F-scores.

In Fig. \ref{fig:my_label_bound}, we study the probability of erroneous spike detection as a function of $\tD$ and validate the upper bound derived in Theorem \ref{thm:full_prob}. Recall that the decoding is considered successful if ``every" spike is detected correctly. Therefore, it becomes more challenging to ``exactly super-resolve" all the spikes in presence of noise as the desired resolution becomes finer. We calculate the empirical probability of error and overlay the corresponding theoretical bound. As shown in Fig. \ref{fig:my_label_bound}, the empirical probability of error is indeed upper bounded by the bound computed by our analysis. The empirical probability of error increases as a function of undersampling factor $\tD$. 

\begin{figure}[h]
\setlength\belowcaptionskip{-1.0\baselineskip}
    \centering
    \includegraphics[width=0.75\linewidth]{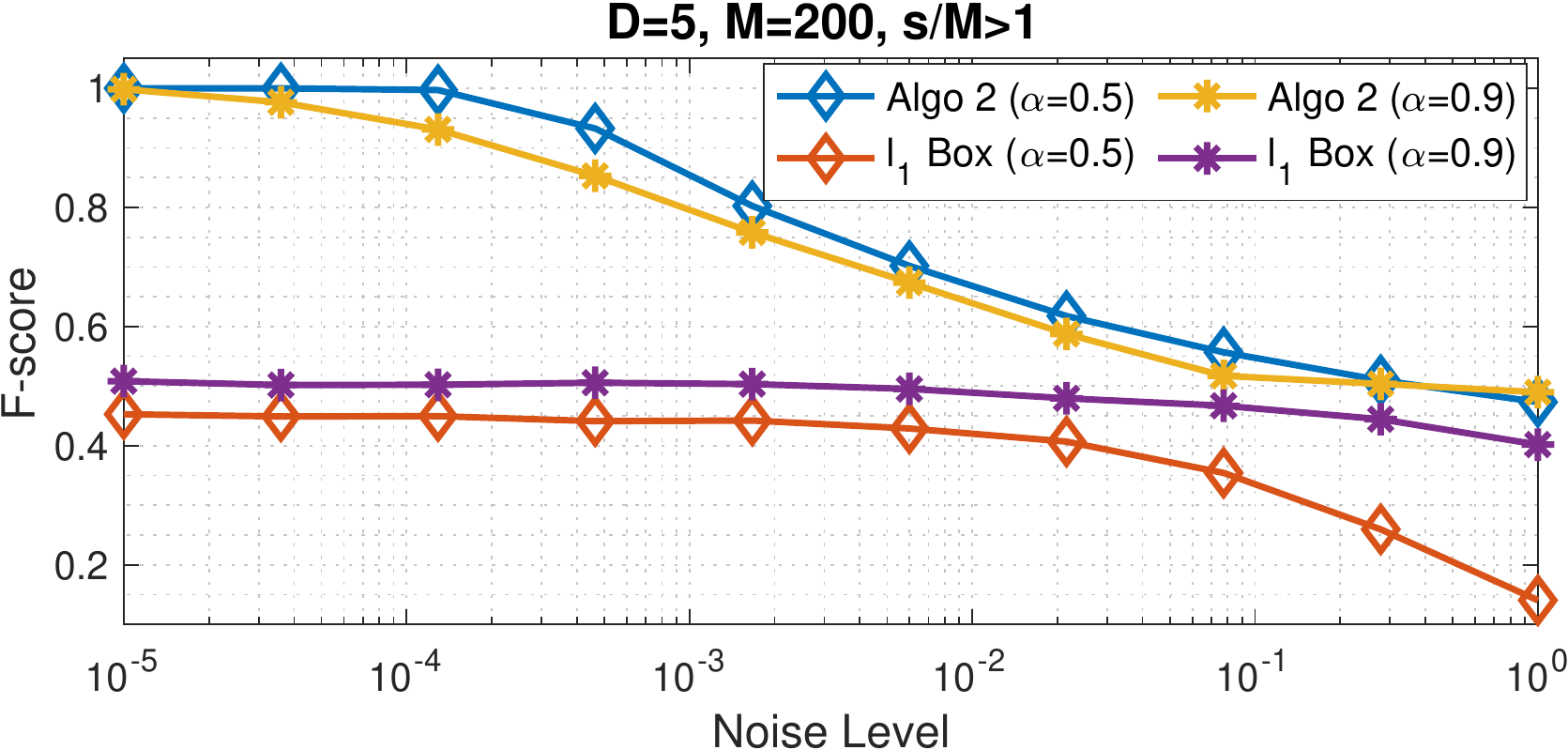}\\ 
    \vspace{0.1cm}
    \includegraphics[width=0.75\linewidth]{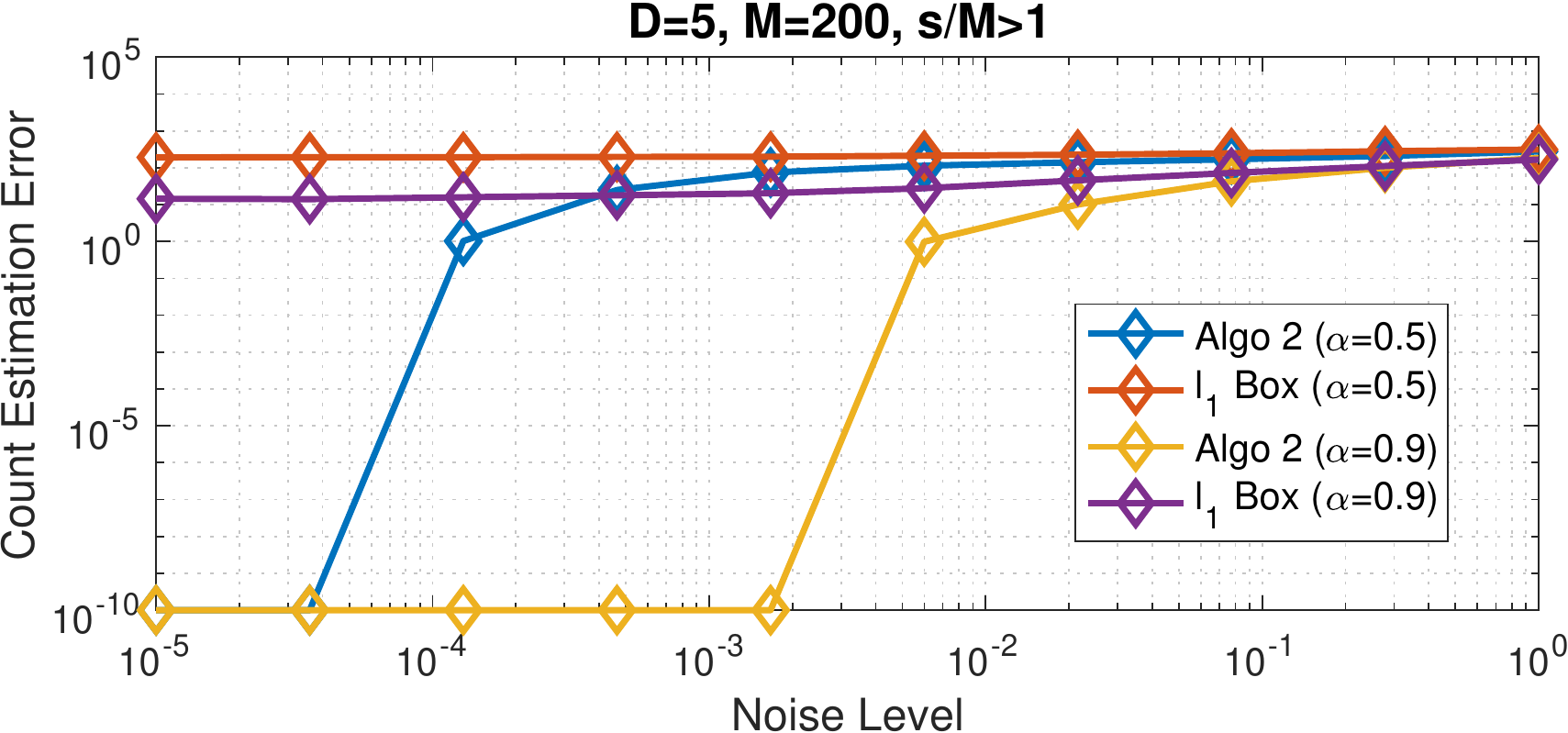}\\
    \caption{Spike detection performance with noisy measurements for different filter parameters $\alpha$. (Top) F-score vs. noise level ($\sigma$)  (Bottom) Count estimation error vs. noise level. Here, $L=1000$ and expected sparsity is fixed at $s=350$ where we operate in the regime $s>M$. The F-score is computed with a tolerance of $t_0=2$.}
    \label{fig:my_label_vary_noise_new}
\end{figure}

\bl{Finally, we evaluate the noise tolerance of the proposed methodology by comparing the average F-score as a function of the noise level $\sigma$, {while} keeping the spiking rate and undersampling factor fixed at $p=0.35$ and $\tD=5$, respectively. As seen in Fig. \ref{fig:my_label_vary_noise_new} (Top), the performance of both algorithms degrades with increasing noise level and this is also consistent with the intuition that it becomes harder to super-resolve spikes with more noise. However, for both filter parameters considered in this experiment Algorithm $2$ has a higher F-score compared to box-constrained $l_1$ minimization. {For large noise levels} (comparable to spike amplitude $A=1$), the performance gap decreases for $\alpha=0.9$ but Algorithm $2$ achieves a much higher F-score for $\alpha=0.5$ at all noise levels.} 

\bl{As discussed {in Section IV-B, we next study a relaxed notion of spike recovery which focuses on the spike counts occurring between two consecutive low-rate samples}. Let $\blds{\Gamma}=[\gamma[0],\cdots,\gamma[M-1]]^{\top}$ be the vector of counts and $\blds{\widehat{\Gamma}}$ be its estimate. In Fig. \ref{fig:my_label_vary_noise_new} (Bottom) we plot the average $l_1$ distance $\Vert\blds{\Gamma}-\blds{\widehat{\Gamma}}\Vert_1$ as a function of the noise level. {We observe that for $\alpha=0.9$ (it can be verified from Fig. \ref{fig:my_label_d} (Top) that $0.9\in \mathcal{F}_5$)}, it is possible to exactly recover the spike counts at higher noise even though the F-score (for timing recovery) has dropped below $1$. However, this is not the case for $\alpha=0.5$, since $0.5\not\in\mathcal{F}_{5}$. This is consistent with the conclusion of Theorem $4$ which states that {when $\alpha\in \mathcal{F}_{\tD}$, the noise tolerance for exact count recovery can be much larger than exact spike recovery since $\Delta^{c}_{\min}(\alpha,\tD)>\Delta \theta_{\min}(\alpha,\tD)$.}}
\vspace{-0.3cm}
\begin{figure}[h]
    \centering
    \includegraphics[width=0.76\linewidth]{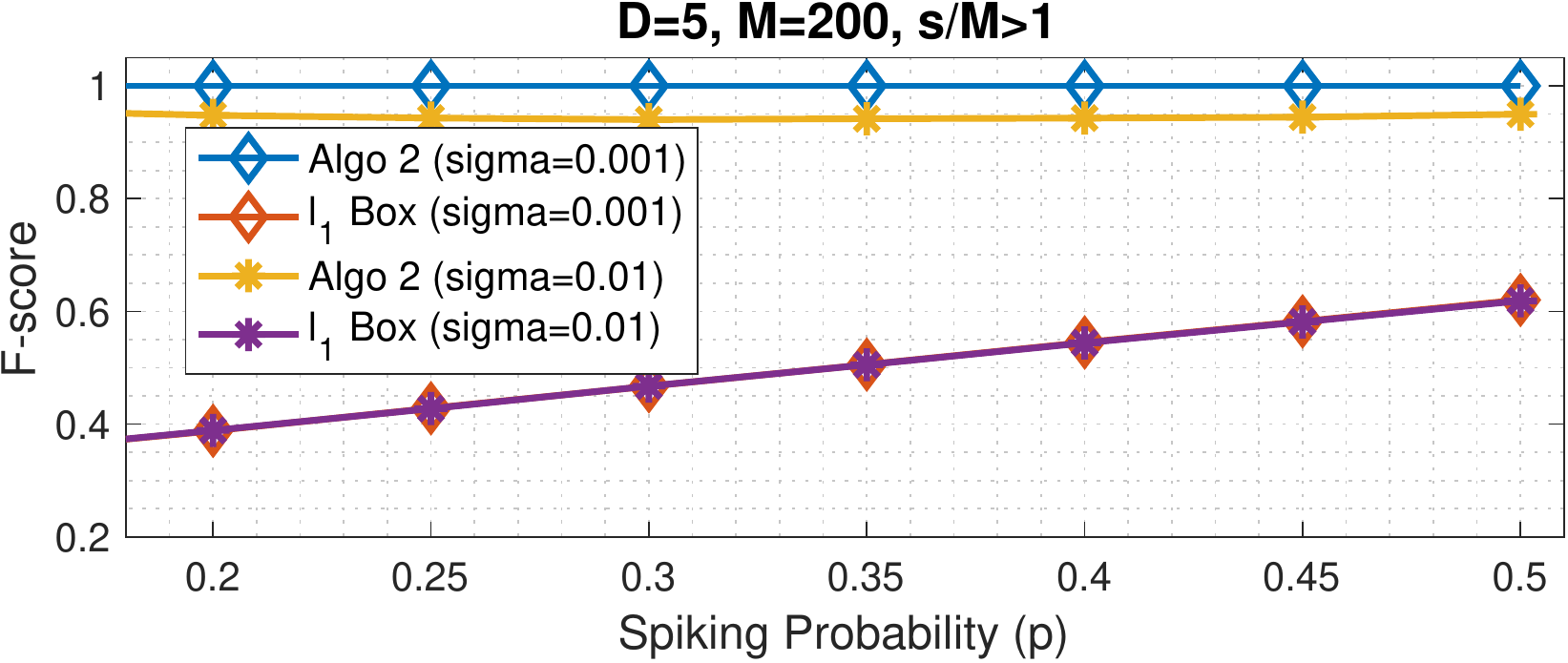}\\
    \caption{Spike detection performance with noisy measurements. F-score vs. spiking probability ($p$) for different noise levels $\sigma$ (fix $\alpha=0.9$, $\tD=5,L=1000$) in the extreme compression regime $s>M$.}
    \label{fig:my_label_prob}
\end{figure}
\vspace{-0.6cm}
\begin{figure}[h]
\setlength\belowcaptionskip{-1.4\baselineskip}
    \centering
    \begin{tabular}{c}
         \includegraphics[width=0.75\linewidth]{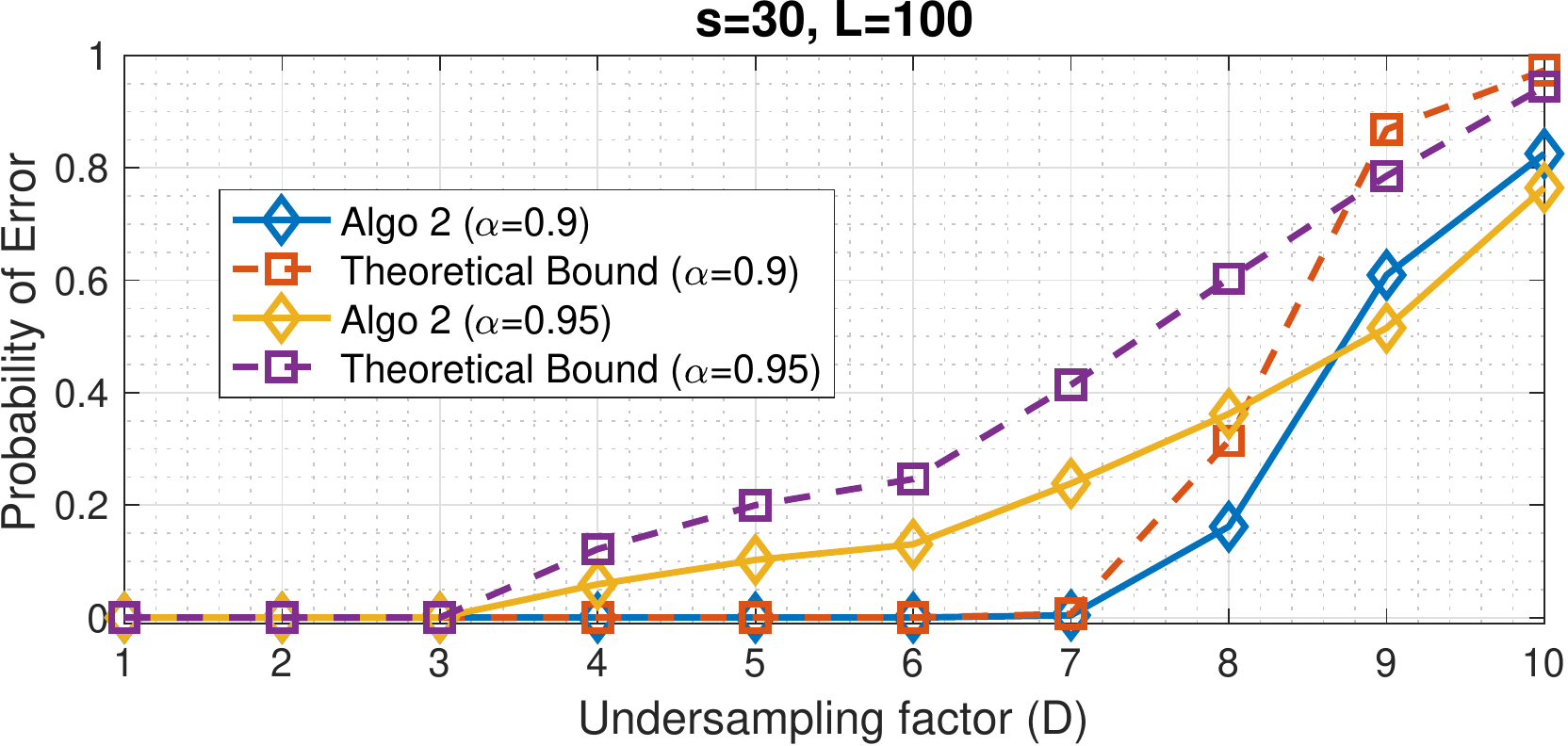}
    \end{tabular}
    \caption{Probability of erroneous detection of high-rate spikes $\xhi \in \{0,1\}^{L}$ as a function of the  undersampling factor $\tD$. Theoretical upper bounds are overlaid using dotted lines. Here, $L=100$.}
    \label{fig:my_label_bound}
\end{figure}
\vspace{-0.15cm}
\subsection{Spike Deconvolution from Real Calcium Imaging Datasets}
\bl{{We now discuss how the mathematical framework developed in this paper can be used for super-resolution spike deconvolution in calcium imaging. Two-photon calcium imaging is a widely used imaging technique for large scale recording of neural activity with high spatial but poor temporal resolution.} In calcium imaging, the signal $\xhi$ corresponds to the underlying neural spikes which is modeled to be binary valued on a finer temporal scale \cite{brette2012handbook,rupasinghe2020robust}. Each neural spike results in a sharp rise in Ca$^{2+}$ concentration followed by a slow exponential decay, leading to superposition of the responses from nearby spiking events \cite{brette2012handbook,vogelstein2009spike,deneux2016accurate}. This calcium transient can be modeled by the first order autoregressive model introduced in Section II. The decay time constant depends on the calcium indicator and essentially determines the filter parameter $\alpha$. The signal $y{\hi}[n]$ is an unobserved signal corresponding to sampling the calcium fluorescence at a high sampling rate (at the same rate as the underlying spikes). The {observed} calcium signal $y{\lo}[n]$ corresponds to downsampling $y\hi[n]$ at an interval determined by the frame rate of the microscope. The frame rate of a typical scanning microscopy system (that captures the changes in the calcium fluorescence) is determined by the amount of time required to spatially scan the desired field of view, which makes it significantly slower compared to the temporal scale of the neural spiking activity. We model this discrepancy by the downsampling operation (by a factor $\tD$). Therefore, the mathematical framework developed in this paper {can be
directly applied} to reconstruct the underlying spiking activity at a temporal scale finer than the sampling rate of the calcium signal.} {Using real calcium imaging data, we demonstrate a way to fuse our algorithm with a popular spike deconvolution algorithm called OASIS \cite{friedrich2017fast}.} \bl{OASIS solves an $l_1$ minimization problem similar to \eqref{eqn:P1_noisy} with only the non-negativity constraint, in order to exploit the sparse nature of the spiking activity. Unlike our approach where we wish to obtain spikes representation on a finer temporal scale, OASIS returns the spike estimates on the low-resolution grid. This is typically {used to infer the} spiking rate over a temporal bin equal to the sampling interval. We demonstrate that our proposed framework {can be
integrated with OASIS and improve its performance}. As we saw in the synthetic experiments, the noise level is an important consideration. By augmenting Algorithm $2$ with OASIS, referred as ``B-OASIS", the denoising power of $l_1$ minimization can be leveraged.}\bl{Let $\mf{\hat{x}}_{\text{l1}}\in \mathbb{R}^{M}$ be the estimate obtained on a low-resolution grid by solving the $l_1$ minimization problem such as the one implemented in OASIS. We can obtain an estimate of the denoised calcium signal as $\hat{y}{\lo}[n]=\alpha^{\tD}\hat{y}{\lo}[n]+\hat{x}_{\text{l1}}[n], n\geq 1$ and $\hat{y}{\lo}[0]=\hat{x}_{\text{l1}}[0]$. We can now utilize the denoised calcium signal $\hat{y}{\lo}[n]$ generated by OASIS to obtain the estimate $c_e[n]$ indirectly.} Due to the non-linear processing done by OASIS, it is difficult to obtain the resulting noise statistics. An important advantage of Algorithm $2$ is that it does not rely on the knowledge of the noise statistics. Hence, we can directly apply Algorithm $2$ on $\hat{c}_e[n]=\hat{y}{\lo}[n]-\alpha^{\tD} \hat{y}{\lo}[n-1]$ (instead of $c_e[n]$) to obtain a binary ``fused super-resolution spike estimate". 
\vspace{-0.3cm}
\begin{figure}[h]
\setlength\belowcaptionskip{-1.4\baselineskip}
    \centering
    \begin{tabular}{cc}
      \includegraphics[width=0.3\linewidth]{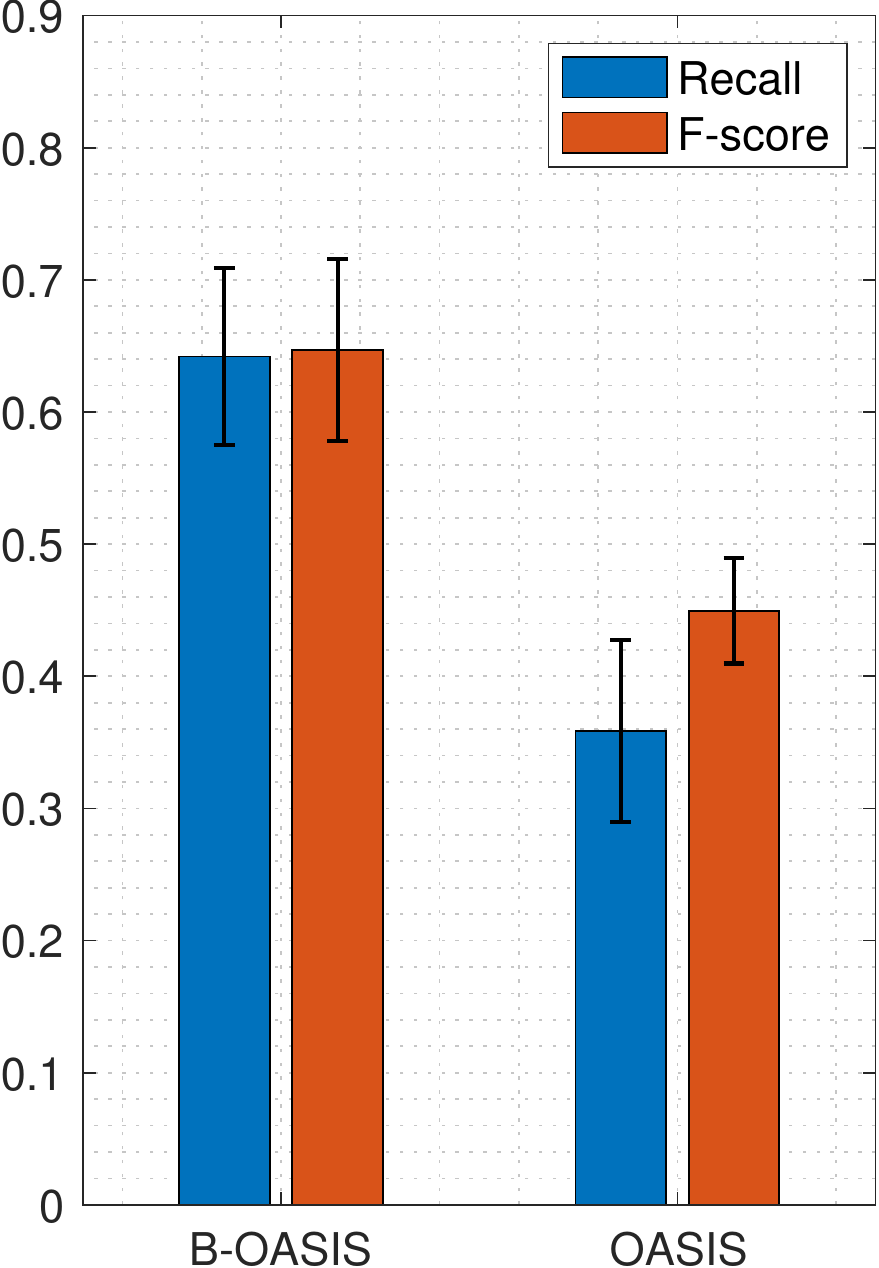}&\includegraphics[width=0.3\linewidth]{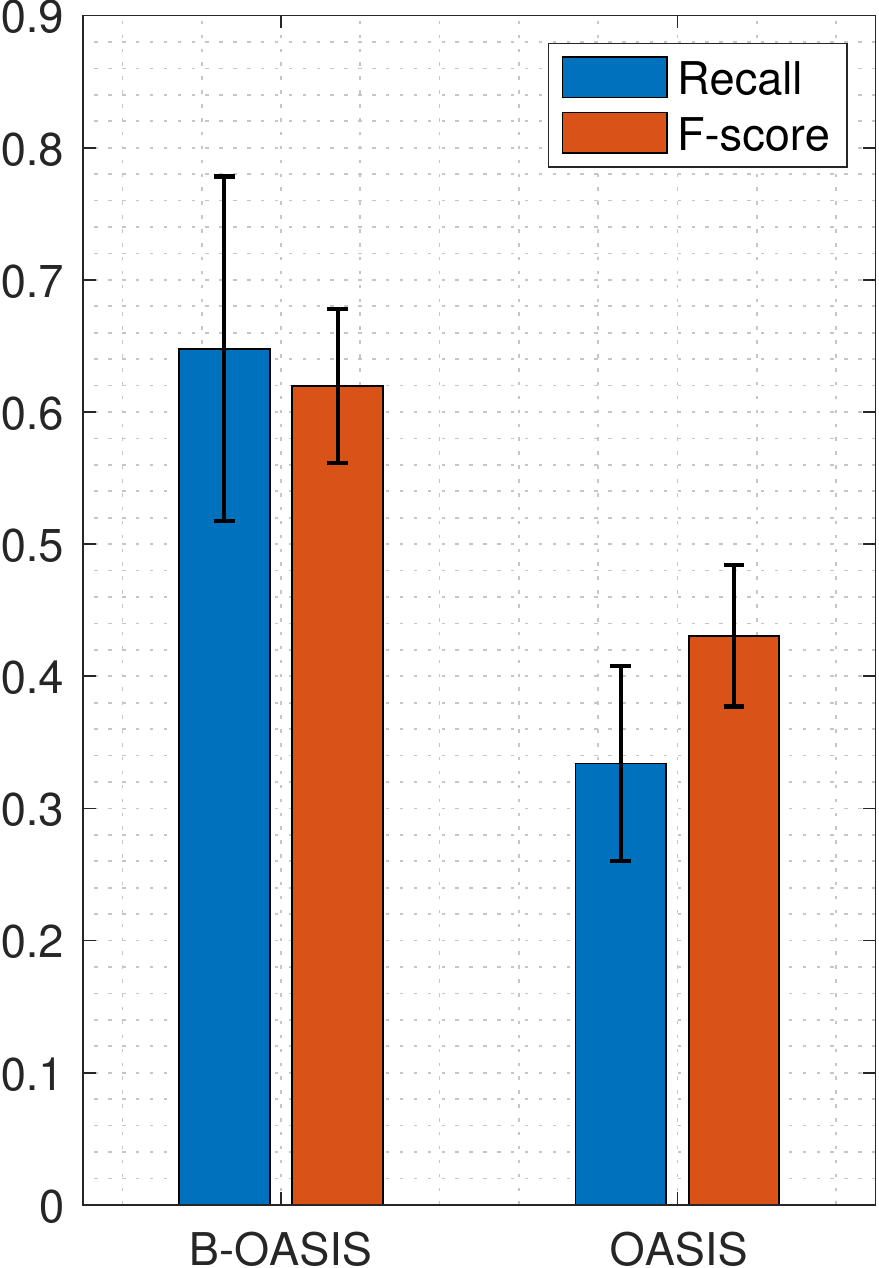}
    \end{tabular}
    \caption{Spike detection performance of OASIS and B-OASIS on GCaMP6f dataset sampled at (Left) $60$ Hz and (Right) $30$ Hz. We compare the average F-score of data points where the F-score of OASIS is $<0.5$. Standard deviation is depicted using the error bars.}
    \label{fig:quant}
\end{figure}



 
\vspace{-0.2cm}
\subsection{Results}
\vspace{-0.1cm}
We evaluate the algorithms on the publicly available GENIE dataset\cite{chen2013ultrasensitive,genie2015simultaneous} which consists of simultaneous calcium imaging and \textit{in vivo} cell-attached recording from the mouse visual cortex using genetically encoded GCaMP6f calcium indicator GCaMP6f\cite{chen2013ultrasensitive,genie2015simultaneous}. The calcium images were acquired at a frame rate of $~60$ Hz and the ground truth electrophysiology signal was digitized at $10$ KHz and synchronized with the calcium frames. In addition to using the original data, we also synthetically downsample it to emulate the effect of a lower frame rate of $30$ Hz, and evaluate how the performance changes by this downsampling operation. 
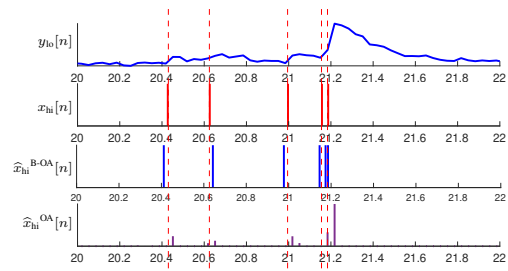
\begin{figure}[t!]
\setlength\belowcaptionskip{-1.5\baselineskip}
\begin{center}
\resizebox{!}{3.5cm}{
\begin{tikzpicture}
    \node (image1) at (0,0){\includegraphics[width=0.9\linewidth]{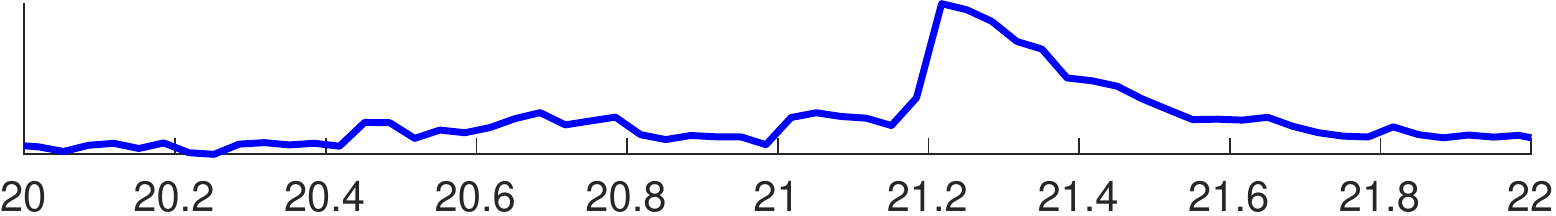}};
    \coordinate (a) at (-3.7,0.2);
    \node[label={left:{\scriptsize$y{\lo}[n]$}}] at (a)  {};
    \node (image2) at (0,-1.1){\includegraphics[width=0.9\linewidth]{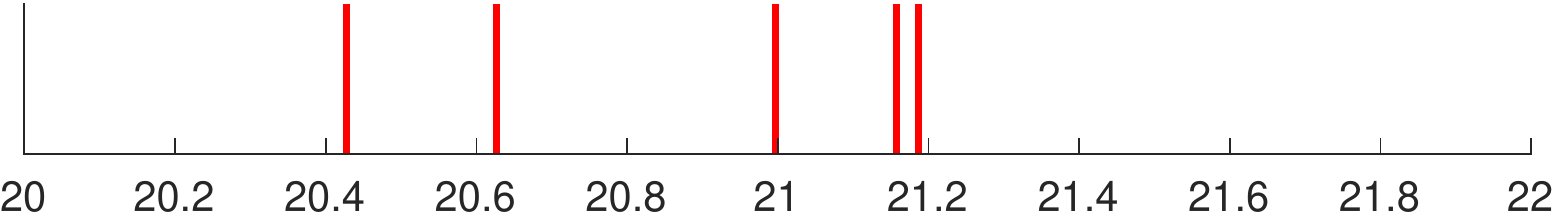}};
    \coordinate (b) at (-3.7,-1.0);
    \node[label={left:{\scriptsize$x{\hi}[n]$}}] at (b)  {};
    \node (image3) at (0,-2.2){\includegraphics[width=0.9\linewidth]{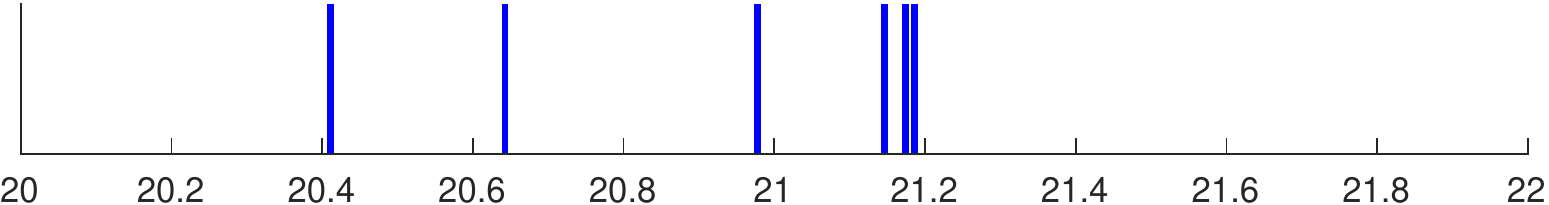}};
     \coordinate (c) at (-3.7,-2.1);
    \node[label={left:{\scriptsize$\hat{x}{\hi}^{\tiny\text{B-OA}}[n]$}}] at (c)  {};
    \node (image3) at (0,-3.3){\includegraphics[width=0.9\linewidth]{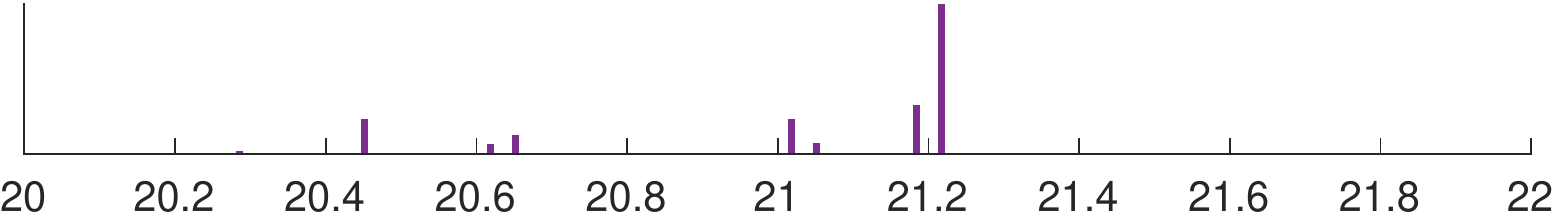}};
     \coordinate (d) at (-3.7,-3.1);
    \node[label={left:{\scriptsize$\hat{x}{\hi}^{\tiny\text{OA}}[n]$}}] at (d)  {};
    \coordinate (D1) at (-2.2,0.8);
    \coordinate (D2) at (-2.2,-4);
    \draw [red,dashed] (D1) -- (D2);
    
    
    \coordinate (D3p) at (-1.45,0.8);
    \coordinate (D4p) at (-1.45,-4);
    \draw [red,dashed] (D3p) -- (D4p);
    
    
    \coordinate (D3h) at (  -0.02,0.8);
    \coordinate (D4h) at (  -0.02,-4);
    \draw [red,dashed] (D3h) -- (D4h);
    
    \coordinate (D5) at (0.6,0.8);
    \coordinate (D6) at (0.6,-4);
    \draw [red,dashed] (D5) -- (D6);
    
    \coordinate (D5p) at (0.71,0.8);
    \coordinate (D6p) at (0.71,-4);
    \draw [red,dashed] (D5p) -- (D6p);

\end{tikzpicture}
}
\end{center}
\caption{Example of spike reconstruction on GENIE dataset (GCaMP6f indicator) using OASIS and B-OASIS (binary augmented) with calcium signal sampled at $30$Hz.}
\label{fig:real_raw_plot}
\end{figure}
\bl{In Fig. \ref{fig:real_raw_plot}, we extract an interval of $\sim 2$ sec ({from the neuron $1$ of the GCaMP6f indicator dataset}) and qualitatively compare the detected spikes with the ground truth. 
We downsample the data by a factor of $2$ to emulate frame rate of $30$ Hz, the low-rate grid becomes coarser. As a result of which, we observe an offset between ground truth spikes and estimate produced by OASIS. However, with the help of binary priors (B-OASIS), we can output spikes that are not restricted to be on the coarser scale, and this mitigates the offset observed in the raw estimates obtained by OASIS.}

We quantify the improvement in the performance by comparing the F-scores of OASIS and B-OASIS at both sampling rates ($60$ and $30$ Hz). Since the output of OASIS is non-binary, the estimated spikes are binarized by thresholding. To ensure a fair comparison, we select the threshold by a $80-20$ cross-validation scheme that maximizes the average F-score on a held-out validation set (averaged over $3$-random selections of the validation set). The tolerance for the F-score was set at $100$ ms. The dataset consisted of $34$ traces of length $\sim 234$ s. \bl{The OASIS algorithm has an automated routine to estimate the parameter $\alpha$, {which we utilize for our experiments}. {The amplitude $A$} is estimated using the procedure described in Appendix F. We use $\tD=12$ to obtain the spike representation for B-OASIS.} In order to quantify the performance boost achieved by augmentation, we isolate the traces where the $F-$score of OASIS drops below $0.5$ and compare the average F-score and recall for these data points. As shown in Fig. \ref{fig:quant}, at both sampling rates, we see a significant improvement in the average F-score of B-OASIS over OASIS, attributed to an increase in recall while keeping the precision unchanged. Additionally, despite downsampling, the spike detection performance is not significantly degraded with binary priors, although the detection criteria were unchanged. 
\vspace{-0.5cm}
\section{Conclusion}
\vspace{-0.12cm}
We theoretically established the benefits of binary priors in super-resolution, and showed that it is possible to achieve significant reduction in sample complexity over sparsity-based techniques. Using an AR(1) model, we developed and analyzed an efficient algorithm that can operate in the extreme compression regime ( $M\ll K$) by exploiting the special structure of measurements and trading memory for computational efficiency at run-time. We also demonstrated that binary priors can be used to boost the performance of existing neural spike deconvolution algorithms. In the future, we will develop algorithmic frameworks for incorporating binary priors into different neural spike deconvolution pipelines and evaluate the performance gain on diverse datasets. \bl{ The extension of this binary framework for higher-order AR filters is another exciting future direction.}
\vspace{-0.4cm}
\begin{appendix}
\section*{\small{Appendix A: Proof of Theorem $1$}}
\vspace{-0.2cm}
\begin{proof}
We show that for any $\alpha$ in $0<\alpha <1$, except possibly for a set consisting of only a finite number of points, \eqref{eqn:binary_eq} always has a unique binary solution. Consider all possible $\tD-$dimensional ternary vectors with their entries chosen from $\{-1,0,1\}$, and denote them as  $ \mathbf{v}^{(i)}=[v^{(i)}_1, v^{(i)}_2, \cdots, v^{(i)}_{\tD}]^T\in \{-1,0,1\}^{\tD}, 0\leq i \leq 3^{\tD}-1. $ We use the convention that $\mathbf{v}^{(0)}=0$. For every $i>0$, we define a set $\mathcal{Z}_{\mf{v}^{(i)}}$ determined by $\mf{v}^{(i)}$ as 
$
    {\mathcal{Z}_{\mf{v}^{(i)}}:={\big\{}x\in (0,1){\big\vert} \sum_{k=1}^{\tD} v^{(i)}_k x^{\tD-k}=0{\big\}}}.
$
Notice that $p_i(x):= \sum_{k=1}^{\tD} v^{(i)}_k x^{\tD-k}$ denotes a polynomial (in $x$) of degree at most $\tD\!-\!1$, whose coefficients are given by the ternary vector $\mf{v}^{(i)}$. The set $\mathcal{Z}_{\mf{v}^{(i)}}$ denotes the set of zeros of $p_i(x)$ that are contained in $(0,1)$. Since the degree of $p_i(x)$ is at most $\tD\!-\!1$,  $\mathcal{Z}_{\mf{v}^{(i)}}$ is a finite set with cardinality at most $\tD\!-\!1$.

Now suppose that the binary solution of \eqref{eqn:binary_eq} is non-unique, i.e., there exist $\mathbf{u},\mathbf{w}\in \{0,A \}^{L}$, $\mf{u}\neq \mf{w}$, such that \ben \smash{\mathbf{H}_{\tD}(\alpha) \mathbf{u}=\mathbf{H}_{\tD}(\alpha) \mathbf{w}\Rightarrow \mathbf{H}_{\tD}(\alpha) \mathbf{u}-\mathbf{H}_{\tD}(\alpha)\mathbf{w}=\mathbf{0}} \label{eqn: Huv}\een By partitioning $\mf{u,w}$ into blocks $\mf{u}^{(n)},\mf{w}^{(n)}$ in the same way as in \eqref{eqn:block_x}, we can re-write \eqref{eqn: Huv} as $u^{(0)} = w^{(0)}$ and 
\begin{align}
\smash{\sum_{i=1}^{\tD}\frac{1}{A}([\mf{u}^{(j)}]_i-[\mf{w}^{(j)}]_i) \alpha^{\tD-i} =0, \quad 1\leq j \leq M-1} \label{eqn:poly}
\end{align}
Since $\mf{u}\neq \mf{w}$, they differ at least at one block, i.e., there exists some $j_0, 1\leq j_0\leq M-1$ such that $\mf{u}^{(j_0)}\neq \mf{w}^{(j_0)}$. Define $\mf{b}:=\frac{1}{A}(\mf{u}^{(j_0)} -  \mf{w}^{(j_0)})$. Then, $\mf{b}$ is a non-zero ternary vector, i.e., $\mf{b} \in \{-1,0,1\}^{\tD}$. Now from \eqref{eqn:poly}, we have \ben \smash{\sum_{i=1}^{\tD} [\mf{b}]_i\alpha^{\tD-i} =0,} \label{eqn:polyb} \een
which implies that $\alpha \in \mathcal{Z}_{\mf{b}}$. Since $\mf{b}$ can be any one of the $3^{\tD}-1$ ternary vectors $\{\mf{v}^{(i)} \}_{i=1}^{3^{\tD}-1}$, \eqref{eqn:polyb} holds if and only if $\alpha \in \mathbb{S}:=\bigcup_{i=1}^{3^{\tD}-1} \mathcal{Z}_{\mf{v}^{(i)}}$, i.e., $\alpha$ is a root of at least one of the polynomials {$p_i(x)$ defined by the vectors $\mf{v}^{(i)}$ as their coefficients}. For each $\mf{v}^{(i)}$, since the cardinality of $\mathcal{Z}_{\mf{v}^{(i)}}$ is at most $\tD-1$, $\mathbb{S}$ is a finite set (of cardinality at most $(\tD-1)(3^{\tD}-1)$), and therefore its Lebesgue measure is $0$. This implies that $\eqref{eqn:binary_eq}$ has a non-unique binary solution only if $\alpha$ belongs to the measure zero set $\mathbb{S}$, thereby proving the theorem. 
\end{proof}
\vspace{-0.5cm}
\section*{Appendix B: Proof of Lemma $2$ and Lemma $3$}
\vspace{-0.1cm}
\begin{proof}
    (i) Let $s_n$ denote the sparsity (number of non-zero elements) of the $n^{\text{th}}$ block $\xhi^{(n)}$ of $\xhi$. 
    Then, the total sparsity is $\Vert \xhi \Vert_0=\sum_{n=0}^{M-1}s_n$. We will construct a vector 
    $\mathbf{v} \in \mathbb{R}^L$, $\mf{v}\neq \xhi$ that satisfies $\mathbf{c}=\mathbf{H}_{\tD}(\alpha)\mathbf{v}$ and $\Vert \xhi \Vert_0 \geq \Vert \mathbf{v}\Vert_0$. Following \eqref{eqn:block_x}, consider the partition of  $\mathbf{v}$ 
    $  
    \mf{v}=[v^{(0)},\mf{v}^{(1)\top},\cdots,\mf{v}^{(M-1)\top}]^{\top}
    $.
Firstly, we assign $v^{(0)}=c[0]=x{\hi}^{(0)}$. We construct $\mathbf{v}^{(n)}$ as follows. For each $n\geq 1$, there are three cases: 

{\bf Case I:} $s_n=0$. In this case, $\xhi^{(n)} =\mathbf{0}$ and hence $c[n]=0$. Therefore, we assign $\mathbf{v}^{(n)} = \xhi^{(n)} =\mathbf{0}$.

{\bf Case II:} $s_n =1$. First suppose that $[\xhi^{(n)}]_{\tD}=0$. We construct $\mathbf{v}^{(n)}$ as follows:
\begin{equation}
    \smash{[\mathbf{v}^{(n)}]_k =\begin{cases}
    c[n], &\text{ if }  k=\tD\\
    0, &\text{ else }
    \end{cases}} \label{eqn:construct1}.
\end{equation}
Next suppose that $[\xhi^{(n)}]_{\tD}\neq 0$. Since $s_n=1$, this implies that $[\xhi^{(n)}]_k= 0, k=1,\cdots, \tD-1$. In this case, we construct $\mathbf{v}^{(n)}$ as follows:
\begin{equation}
    \smash{[\mathbf{v}^{(n)}]_k =\begin{cases}
    c[n]/\alpha, &\text{ if } k=\tD-1\\
    0, &\text{ else }
    \end{cases}} \label{eqn:construct2}.
\end{equation} 
Notice that both \eqref{eqn:construct1} and \eqref{eqn:construct2} ensure that $\mathbf{v}^{(n)}\neq \xhi^{(n)}$ and $c[n] = \mathbf{h}^{T}_{\alpha}\mathbf{v}^{(n)}$. Moreover, $\Vert \mathbf{v}^{(n)}\Vert_0 = s_n.$ 

{\bf Case III:} $s_n\geq 2$. In this case, we follow the same construction as $\eqref{eqn:construct1}$. As before $\mathbf{v}^{(n)}$ satisfies $c[n]=\mathbf{h}^{\top}_{\alpha}\mathbf{v}^{(n)}$. Since $\Vert \xhi^{(n)}\Vert_0 \geq 2$ and $\Vert\mathbf{v}^{(n)}\Vert_0 =1$, we automatically have $\mathbf{v}^{(n)}\neq \xhi^{(n)}$, and  $\Vert \mathbf{v}^{(n)}\Vert_0 < s_n$.
Therefore, combining the three cases, we can construct the desired vector $\mathbf{v}$ that satisfies   $\mathbf{v}\neq \xhi$, $\mathbf{c}=\mathbf{H}_{\tD}(\alpha)\mathbf{v}$, and $\Vert \mathbf{v}\Vert_0 \leq \sum_{n=0}^{M-1} s_n = \Vert \xhi^{(n)}\Vert_0$. Therefore, the solution $\mathbf{x}^{\star}$ to \eqref{eqn:P0} satisfies $\Vert \mathbf{x}^{\star} \Vert_0 \leq \Vert \mathbf{v}\Vert_0 \leq  \Vert \xhi^{(n)}\Vert_0$.

(ii) In this case, we construct $\mathbf{v}^{(n_0)}$ according to Case III. Since $\Vert\mathbf{v}^{(n_0)} \Vert_0 <s_{n_0}$, and $\Vert\mathbf{v}^{(n)} \Vert_0 \leq s_{n}, n\neq n_0$, we have $\Vert \mathbf{v} \Vert_0 < \Vert \xhi \Vert_0$, implying $\Vert \mathbf{x}^{\star} \Vert_0 \leq \Vert \mathbf{v} \Vert_0 < \Vert \xhi \Vert_0$.
\end{proof}
\vspace{-0.4cm}
\subsection{Proof of Lemma $3$}
\begin{proof}
\bl{We will construct a vector $\mathbf{v}\in \mathbb{R}^{L}$ whose support is of the form \eqref{eqn:supp_l1_box}, that is feasible for \eqref{eqn:P1_box}, and
we will prove that it has the smallest $l_1$ norm. {Using the block structure given by \eqref{eqn:block_x}, we choose $\mf{v}^{(0)}=c[0]$. For each $n\geq 1$, we construct $\mf{v}^{(n)}$ based on the following two cases:} \\
\textbf{Case I:} $c[n] \geq A$. 
Let $k_n$ be the largest integer such that the following holds:
$
\mu[n]:=A(1+\alpha+\cdots+\alpha^{k_n-1})\leq c[n],
$
where $1\leq k_n\leq \tD$. Note that $k_n=1$ {always produces a valid lower bound. However, we are interested in the largest lower bound on $c[n]$ of the above form.} We choose 
\begin{equation*}
   { [\mf{v}^{(n)}]_k=\begin{cases}
    A, \quad \text{ if } \tD-k_n+1\leq k \leq \tD\\
    (c[n]-\mu[n])/\alpha^{k_n},\text{ if } k=\tD-k_n \\
    0, \text{ else }
    \end{cases}}
\end{equation*}
{ It is easy to verify that  
$\mf{h}_{\alpha}^{\top}\mf{v}^{(n)}=c[n]$.} From the definition of $k_n$, it follows that {$\mu[n]\leq c[n]<\mu[n]+A\alpha^{k_n}$ and hence, $0\leq (c[n]-\mu[n])/\alpha^{k_n}<A$}, which ensures that $\mf{v}$ obeys the box-constraints { in \eqref{eqn:P1_box}}. Now, let $\mf{v}_f\in \mathbb{R}^{L}$ be any feasible point { of \eqref{eqn:P1_box}}  which {must be of the form } $\mf{v}_f^{(0)}=c[0], \mf{v}_f^{(n)}=\mf{v}^{(n)}+\mf{r}^{(n)}$, where $\mf{r}^{(n)}\in \mathcal{N}(\mf{h}_{\alpha}^{\top})$ is a vector in the null-space of $\mf{h}_{\alpha}^{\top}$. It can be verified that the {following vectors $\{ \mf{w}_t\}_{t=1}^{D-1}$} form a basis for  $\mathcal{N}(\mf{h}_{\alpha}^{\top})$:
\begin{equation*}
    [\mf{w}_t]_k=\begin{cases}
    1, \quad &k=t\\
    -\alpha, \quad &k=t+1\\
    0, \quad &\text{ else }
    \end{cases}, 
\end{equation*}
Therefore, $\exists$ $\{\beta_t^{(n)}\}_{t=1}^{\tD-1}$ {such that $\mf{r}^{(n)}=\sum_{t=1}^{\tD-1}\beta_t^{(n)}\mf{w}_t$.} {We further consider two scenarios: (i) $1\leq k_n \leq \tD-2.$ In this case $[\mathbf{v}^{(n)}]_1=0$, and for $k=1,2,\cdots D$, $[\mathbf{v}_f^{(n)}]_k$ satisfies \footnote{In the definition of $\mathbf{v}_f^{(n)}$, an assignment will be ignored if the specified interval for $k$ is empty.} }
\begin{align*}
[\mf{v}_f^{(n)}]_k=\begin{cases}
 \beta_k^{(n)},\text{ if } k=1\\
 \beta_k^{(n)}-\alpha \beta_{k-1}^{(n)}, \text{ if } 2\leq k \leq \tD-k_n-1\\
 [\mf{v}^{(n)}]_k+ \beta_k^{(n)}-\alpha \beta_{k-1}^{(n)}, \text{ if }k=\tD-k_n\\
 A+ \beta_k^{(n)}-\alpha \beta_{k-1}^{(n)}, \text{ if }\tD-k_n+1\leq k \leq \tD-1\\
 A-\alpha \beta_{k-1}^{(n)}, \text{ if } k=\tD
 \end{cases}
\end{align*}
To ensure $\mf{v}_f^{(n)}$ is a feasible point for \eqref{eqn:P1_box}, the following must hold: $0\leq \beta_{\tD-1}^{(n)}\leq A/\alpha$ and $0 \leq \beta_1^{(n)}\leq A$. For $2\leq k \leq \tD-k_n-1$, the constraint $[\mf{v}_f^{(n)}]_k \geq 0$ implies $\beta_{k}^{(n)}\geq \alpha\beta_{k-1}^{(n)}$. Since $\beta_{1}^{(n)}\geq 0$, it follows that $\beta_k^{(n)}\geq 0$ for all $2\leq k \leq \tD-k_n-1$. For $\tD-k_n+1\leq k \leq \tD-1$, the constraint $[\mf{v}_f^{(n)}]_k \leq A$ implies $\beta_{k-1}^{(n)}\geq \beta_{k}^{(n)}/\alpha$. Since $\beta_{\tD-1}^{(n)}\geq 0$, it follows that $\beta_k^{(n)}\geq 0$ for all $\tD-k_n \leq k\leq \tD-1$.
(ii) $k_n \in \{\tD-1,\tD\}.$ {In this case, for $k=1,2,\cdots, D$, $[\mathbf{v}_f^{(n)}]_k$ satisfies}
\begin{align*}[\mf{v}_f^{(n)}]_k=\begin{cases}
[\mf{v}^{(n)}]_1+\beta_1^{(n)}, \text{ if } k=1\\
A+\beta_k^{(n)}-\alpha \beta_{k-1}^{(n)}, \text{ if } 2\leq k \leq \tD-1\\
A-\alpha\beta^{(n)}_{k-1}, \text{ if } k=\tD
\end{cases}
\end{align*}
For $2\leq k \leq \tD-1$, the box-constraint $[\mf{v}_f^{(n)}]_k \leq A$ implies $\beta_{k-1}^{(n)}\geq \beta_{k}^{(n)}/\alpha$. Since $\beta_{\tD-1}^{(n)}\geq 0$, it follows that $\beta_k^{(n)}\geq 0$ for all $1 \leq k\leq \tD-1$. Summarizing, we have established that $\beta_{i}^{(n)}\geq 0, \forall i.$\\
\textbf{Case II:} $c[n]<A$.
In this case, $\mf{v}^{(n)}$ is constructed following \eqref{eqn:construct1}, and hence $\mf{v}_f^{(n)}$ has the following structure:
\begin{align*}
    {[\mf{v}_f^{(n)}]_k=\begin{cases}
    \beta_k^{(n)},\text{if } k=1\\
    -\alpha \beta_{k-1}^{(n)}+\beta_k^{(n)},\text{ if } 2\leq k \leq \tD-1\\
    c[n]-\alpha \beta_{k-1}^{(n)}, \text{ if }k=\tD
    \end{cases}}
\end{align*}
To ensure $\mf{v}_f^{(n)}$ is a feasible point, it must hold that
$
  \beta_1^{(n)}\geq 0, \beta_k^{(n)}\geq\alpha \beta_{k-1}^{(n)}\geq 0 \text{ for } 2\leq k \leq \tD-1.$ {Hence, in both Cases I and II, we established that {$\beta_{k}^{(n)}\geq 0$}. For each case, since $\mf{v}_f^{(n)}$ is a non-negative vector $\forall n$, it can be verified that}
\begin{align*}
  \Vert \mf{v}_f\Vert_1 &=\sum_{n=0}^{M-1}\Vert \mf{v}_f^{(n)}\Vert_1={v}_f^{(0)}+\sum_{n=1}^{M-1}\sum_{k=1}^{\tD}[ \mf{v}_f^{(n)}]_k\\
  &=\underbrace{c[0]+\sum_{n=1}^{M-1}\sum_{k=1}^{\tD}[ \mf{v}^{(n)}]_k}_{\Vert \mf{v}\Vert_1}+\sum_{n=1}^{M-1}\sum_{k=1}^{\tD-1}(1-\alpha)\beta_k^{(n)}
\end{align*}
We used the fact that $\sum_{k=1}^D \sum_{t=1}^{D-1}\beta_t^{(n)} [\mathbf{w}_t]_k = \sum_{t=1}^{\tD-1}(1-\alpha)\beta_t^{(n)}$. 
If $\mf{v}_f \neq \mf{v}$, we must have  $\beta_k^{(n)} \neq 0$ for some $k$ and $n>0$. This implies that $\Vert \mf{v}_f\Vert_1>\Vert \mf{v}\Vert_1$. {It is easy to see that the support of the constructed vector is of the form \eqref{eqn:supp_l1_box}. Moreover, based on the above argument, $\mathbf{v}$ is the only vector that has the minimum $l_1$ norm among all possible feasible points of \eqref{eqn:P1_box}.}}
\end{proof}
\vspace{-0.6cm}

\section*{\small{Appendix C: Proof of Lemma $7$}}
\vspace{-0.1cm}
\begin{proof}
\bl{For any $0<\alpha \leq 0.5$, we begin by showing that for an integer $p\geq 1$ the following inequality holds:
\begin{align}
   \sum_{k=1}^{p}\alpha^{\tD-k}=\alpha^{\tD-p-1}\left(\frac{1-\alpha^p}{1/\alpha-1}\right) < \alpha^{\tD-p-1} \label{eqn:part_sum}
\end{align}
since $1/\alpha-1\geq 1$ and $1-\alpha^{p}<1$ in the regime $0<\alpha \leq 0.5$. 
Let $\mathcal{S}_1=\{0,\alpha^{\tD-1},\alpha^{\tD-2},\alpha^{\tD-1}+\alpha^{\tD-2}\}$. Notice that the elements of $\mathcal{S}_1$ are sorted in ascending order for any $\alpha$ and $\tD$. Now, we recursively define the sets $\mathcal{S}_i$ as follows:
\begin{align}
 \mathcal{S}_{i}:=\{\mathcal{S}_{i-1},\mathcal{S}_{i-1}+\alpha^{\tD-1-i}\}, \ 2\leq i \leq \tD-1 \label{eqn:S_i}
\end{align}
Our hypothesis is that for {every $2\leq i \leq \tD-1$ $\alpha \in (0,0.5]$ and $\tD$, the set $\mathcal{S}_{i}$ as defined in \eqref{eqn:S_i}, is automatically sorted in ascending order.} We prove this via induction. For $i=2$, the sets $\mathcal{S}_{1}$ and $\mathcal{S}_1+\alpha^{\tD-3}$ are individually sorted. Moreover from \eqref{eqn:part_sum}, we can show that:
$
    \max_{a\in \mathcal{S}_{1}} a =\alpha^{\tD-1} + \alpha^{\tD-2} < \alpha^{\tD-3} = \min_{b\in \mathcal{S}_1+\alpha^{\tD-3}} b.
$
This shows that $\mathcal{S}_2$ is ordered,  establishing the the base case of our induction. Now, assume $\mathcal{S}_{i}$ is ordered for some $2\leq i \leq \tD-2$. We need to show that $\mathcal{S}_{i+1}$ is also ordered. As a result of the induction hypothesis, both  $\mathcal{S}_{i}$ and $\mathcal{S}_{i}+\alpha^{\tD-2-i}$ are ordered. Using the ordering of $\mathcal{S}_{i}$, we have:
$
  \max_{a\in \mathcal{S}_{i}} a = \sum_{j=1}^{i+1} \alpha^{\tD-j},  \min_{b\in \mathcal{S}_{i}+\alpha^{\tD-2-i}} b = \alpha^{\tD-(i+1)-1}.
$
From \eqref{eqn:part_sum}, we can conclude that $\max_{a\in \mathcal{S}_{i}} a < \min_{b\in \mathcal{S}_{i}+\alpha^{\tD-2-i}} b$ and hence, $\mathcal{S}_{i+1}$ is also ordered. This completes the induction proof. Also, note that 
for $\alpha \in (0,0.5]$, we have  $\Theta^{\text{sort}}_{\alpha}=\mathcal{S}_{\tD-1}$.\\ Let $\Delta_{\min}(\mathcal{S}_i)$ be the min. distance between the elements of the set $\mathcal{S}_i$. It is easy to see that $\Delta_{\min}(\mathcal{S}_{i})=\Delta_{\min}(\mathcal{S}_{i}+\alpha^{\tD-2-i})$. {Since $\mathcal{S}_{i}$ is sorted for $\alpha \in (0,0.5]$}, $\Delta_{\min}(\mathcal{S}_{i})$ is given by: 
\begin{align}
    \Delta_{\min}(\mathcal{S}_i)&=\min(\Delta_{\min}(\mathcal{S}_{i-1}),\min_{x\in \mathcal{S}_{i-1}+\alpha^{\tD-1-i}} x-\max_{y\in \mathcal{S}_{i-1}}y)\notag\\
&=\min \{ \Delta_{\min}(\mathcal{S}_{i-1}),\alpha^{\tD-i-1}-\sum_{j=1}^{i}\alpha^{\tD-j}\}.\label{eqn:d_min_recur}
\end{align}
Now, we use induction to establish the following conjecture:
\begin{align}
\Delta_{\min}(\mathcal{S}_i)=\alpha^{\tD-1},\ 1\leq i\leq \tD-1\label{eqn:d_min_ind}
\end{align}
For the base case $i=1$,
$
  \Delta_{\min}(\mathcal{S}_1)=\min(\alpha^{\tD-1}, \alpha^{\tD-2}-\alpha^{\tD-1})=\alpha^{\tD-1},
$
where the last equality holds since $\alpha\in (0,0.5]  \Rightarrow\alpha^{\tD-1}(1/\alpha-1)\geq \alpha^{\tD-1}$. Suppose \eqref{eqn:d_min_ind} holds for some $1\leq i \leq \tD-2$. {From the definition of $\Delta_{\min}(\mathcal{S}_{i+1})$ and the induction hypothesis that $\Delta_{\min}(\mathcal{S}_{i}) = \alpha^{D-1}
$, it follows that} 
 $\Delta_{\min}(\mathcal{S}_{i+1})=\min\{\alpha^{\tD-1},\alpha^{\tD-(i+1)-1}-\sum_{j=1}^{i+1}\alpha^{\tD-j}\}
$.
{Again, from the definition of $\Delta_{\min}(\mathcal{S}_{i})$ in \eqref{eqn:d_min_recur}, and the induction hypothesis we also have $\alpha^{\tD-i-1} - \sum_{j=1}^{i} \alpha^{\tD-j} \geq \Delta_{\min}(\mathcal{S}_{i})=\alpha^{\tD-1}$.} Using this and the fact that $\alpha \leq 0.5$, we can show:
\begin{eqnarray*}
 \alpha^{\tD-i-2}&\hspace{-0.3cm}-\alpha^{\tD-i-1}-\sum_{j=1}^{i} \alpha^{\tD-j}\geq \alpha^{\tD-i-2}-2\alpha^{\tD-i-1}+\alpha^{\tD-1} \\
 &\smash{\geq\alpha^{\tD-1}+\alpha^{\tD-i-1}(1/\alpha-2)\geq \alpha^{\tD-1}}
\end{eqnarray*}
Therefore { $\Delta_{\min}(\mathcal{S}_{i+1})\!=\!\min\{\alpha^{\tD-1},\alpha^{\tD-i-2}\!-\!\sum_{j=1}^{i+1}\alpha^{\tD-j}\}=\alpha^{D-1}$}. Thus, we can conclude that  $\Delta_{\min}(\alpha,
\tD)\!=\!\Delta_{\min}(\mathcal{S}_{\tD-1})\!=\!\alpha^{\tD-1}$. 
}
\end{proof}
\vspace{-0.5cm}
\section*{\small{Appendix D: Proof of Theorem $3$}}
\vspace{-0.1cm}
\begin{proof}
The probability of incorrectly identifying $\xhi^{(n)}$ from a single measurement $c_e[n]$ is given by  
\begin{align*}
    &\qquad\qquad p_e:={\mathbb{P}(\hat{\mf{x}}{\hi}^{(n)}\neq \mathbf{x}{\hi}^{(n)})}\\
    &=\smash{\sum_{k=0}^{l_{\tD}} 
    \mathbb{P}(\hat{\mf{x}}{\hi}^{(n)}\neq \xhi^{(n)}\vert\xhi^{(n)}
    =\mathbf{\tilde{v}}_k)\mathbb{P}(\xhi^{(n)}=\mathbf{\tilde{v}}_k)} 
\end{align*}
Given a binary vector $\mathbf{z} \in \{0,1\}^{\tD}$, define the function 
$\psi(\mathbf{z}):=\sum_{k=1}^{\tD}z_k$, which denotes the count of ones in $\mathbf{z}$. Since the noisy observations are given by $c_e[n]=c[n]+e[n]$, where $e[n] = w[n]-\alpha^D w[n-1]$, it follows from  assumption (A2) that $e[n] \sim\mathcal{N}(0,\sigma_1^2)$ where $\sigma_1^2=(1+\alpha^{2\tD})\sigma^2$. 
From \eqref{eqn:err_bou}, we obtain $\mathbb{P}(\hat{\mathbf{x}}{\hi}^{(n)}\neq \xhi^{(n)}|\xhi^{(n)}=\mathbf{\tilde{v}}_0) = \mathbb{P}(e[n] \in \mathcal{E}_0)= Q(\alpha^{\tD-1}/(2\sigma_1)) $.
Similarly, 
$\mathbb{P}(\hat{\mathbf{x}}{\hi}^{(n)}\neq \xhi^{(n)}|\xhi^{(n)}=\mathbf{\tilde{v}}_{l_\tD})= \mathbb{P}(e[n]\in \mathcal{E}_{l_D})  =Q((\tilde{\theta}_{l_\tD}-\tilde{\theta}_{l_\tD-1})/(2\sigma_1))
   =Q(\alpha^{\tD-1}/(2\sigma_1)).$
The last equality follows from the fact that  $\tilde{\theta}_{l_\tD}-\tilde{\theta}_{l_\tD-1}=\alpha^{\tD-1}$. Finally, when conditioned on $\mathbf{x}{\hi}^{(n)}=\mathbf{\tilde{v}}_k$ for $0<k<l_\tD$, from \eqref{eqn:err_in}, we obtain  
$
    \mathbb{P}(\hat{\mathbf{x}}^{(n)}\neq \xhi^{(n)}|\xhi^{(n)}=\mathbf{\tilde{v}}_k) = \mathbb{P}(e[n]\in \mathcal{E}_k) =Q(\frac{\tilde{\theta}_{k}-\tilde{\theta}_{k-1}}{2\sigma_1})
   +Q(\frac{\tilde{\theta}_{k+1}-\tilde{\theta}_{k}}{2\sigma_1}).
$
Due to Assumption (\textbf{A1}) on $\xhi$, we have $
\mathbb{P}(\xhi^{(n)}=\mathbf{\tilde{v}}_k)=p^{\psi(\mathbf{\tilde{v}}_k)}(1-p)^{\tD-\psi(\mathbf{\tilde{v}}_k)}$.
Therefore, $p_e$ is given by
\begin{align}
    &p_e=Q(\alpha^{\tD-1}/(2\sigma_1))(1-p)^\tD+Q(\alpha^{\tD-1}/(2\sigma_1))p^\tD+\notag\\
    &\sum_{k=1}^{l_{\tD}-1}\left(Q(\frac{\tilde{\theta}_{k}-\tilde{\theta}_{k-1}}{2\sigma_1})+Q(\frac{\tilde{\theta}_{k+1}-\tilde{\theta}_{k}}{2\sigma_1})\right)p^{\psi(\mathbf{v}_k)}(1-p)^{\tD-\psi(\mathbf{v}_k)} \label{eqn:prob_E}
\end{align}
The spike train $\xhi$ is incorrectly decoded if at least one of the blocks are decoded incorrectly, hence, the total {probability of error is given by:
\begin{align}
    &\smash{\mathbb{P}(\bigcup_{n=0}^{M-1}\hat{\mathbf{x}}^{(n)}\neq \xhi^{(n)})\leq \sum_{n=0}^{M-1}\mathbb{P}(\hat{\mathbf{x}}^{(n)}\neq \xhi^{(n)})= M p_e} \notag\\
    &\overset{(a)}{\leq} 2M Q(\Delta \theta_{\min}(\alpha,\tD)/(2\sigma_1))\sum_{j=0}^{\tD} p^{j}(1-p)^{\tD-j}{\tD \choose j} \notag\\
    &\overset{(b)}{\leq}2M\exp(-\Delta \theta^2_{\min}(\alpha,\tD)/(4\sigma^2_1))\label{eqn:pet}
\end{align}
where the first inequality follows from union bound and second equality is a consequence of \eqref{eqn:prob_E}.} The inequality $(a)$ follows from the monotonically decreasing property of $Q(.)$ function and the sum can be re-written by grouping all terms with the same count, i.e., $\psi(\mathbf{v}_k)=j$. 
The inequality $(b)$ follows from the inequality $Q(x)\leq \exp(-x^2/2)$ for $x>0$. If the SNR condition \eqref{eqn:snr_cond} holds then from \eqref{eqn:pet} the total probability of error is bounded by $\delta$.
\end{proof}

\vspace{-0.5cm}
\section*{\small{Appendix E: Proof of Theorem $4$}}
\begin{proof}
\bl{We first begin by showing that $\alpha\in \mathcal{F}_{\tD}$ implies that \eqref{eqn:cluster} holds and hence the mapping of spikes with the same counts are clustered. {Notice that for $k=0$, $\theta_{\max}^{k}=\theta_{\min}^{k}=0$. For $k\geq 1$, it is easy to verify that}  $\theta_{\max}^{k}$ and $\theta_{\min}^{k}$ are attained by the spiking patterns $00...1111$ (with $k$ consecutive spikes at the indices $\tD-k+1$ to $\tD$)  and $111...000$ (with consecutive spikes at the indices $1$ to $k$), which allows us to simplify \eqref{eqn:cluster} as $\alpha^{\tD-1}>0$ for $k=0$ and $
 \sum_{i=1}^{k+1}\alpha^{\tD-i}>\sum_{j=0}^{k-1} \alpha^{j}, \ k=1,\cdots,\tD-1$. {The values of $\alpha$ that satisfy each of these relations can be described by the following sets:} 
\begin{equation*}
\smash{\mathcal{G}_0=\{\alpha\in(0,1)\vert \alpha^{\tD-1}>0\},\smash{\mathcal{G}_{k}=\{\alpha\in (0,1)\vert r_k(\alpha)<0\}},} 
\end{equation*}
where $r_k(\alpha)=\alpha^{\tD}-\alpha^{\tD-k-1}-\alpha^{k}+1$ for $1 \leq k\leq \tD-1 $. {It is easy to see that  $\mathcal{F}_{\tD}=\mathcal{G}_{k_0}$.} 
Observe that the relations are symmetric, i.e., $\mathcal{G}_k=\mathcal{G}_{\tD-k-1}$. Furthermore, {for $1\leq k\leq\tD/2$}, we show that $\mathcal{G}_{k} \subseteq \mathcal{G}_{k-1}$ as follows. {Trivially, $\mathcal{G}_1 \subset \mathcal{G}_0$. For $2\leq k \leq D/2$, observe that }
$
    r_k(\alpha)-r_{k-1}(\alpha)=\alpha^{\tD-k}(1-1/\alpha)-\alpha^{k}(1-1/\alpha)=(1/\alpha-1)(\alpha^{k}-\alpha^{\tD-k})\geq 0.
$
Therefore, $\alpha \in \mathcal{G}_{k} \Rightarrow \alpha \in \mathcal{G}_{k-1}$, {$k=1,2\cdots, k_0$. Moreover, since $\mathcal{G}_k=\mathcal{G}_{\tD-k-1}$, it follows that $\mathcal{F}_{\tD}=\mathcal{G}_{k_0} = \cap_{k=0}^{\tD-1} \mathcal{G}_{k}.$ Hence, 
 $\alpha \in \mathcal{F}_{\tD} \Rightarrow \alpha \in \mathcal{G}_i$ for all $0\leq i \leq \tD-1$, which implies that \eqref{eqn:cluster} holds.}
If the noise perturbation satisfies $\vert w[n]\vert < \Delta_{\min}^{\text{c}}(\alpha,\tD)/4$, it implies $\vert e[n] \vert< \Delta_{\min}^{\text{c}}(\alpha,\tD)/2$. For any block $\mf{x}\hi^{(n)}\in \mathcal{C}^{\tD}_k$, $\theta_{\min}^{k}\leq \mf{h}_{\alpha}^{\top}\mf{x}\hi^{(n)} \leq \theta^{k}_{\max}$. If $\vert e[n] \vert < \Delta_{\min}^{\text{c}}(\alpha,\tD)/2$, we have
\begin{align*}
     &\mf{h}_{\alpha}^{\top}\mf{x}\hi^{(n)}+e[n] < \theta^{k}_{\max}+\frac{\Delta^c_{\min}(\alpha,\tD)}{2}< \theta^{k}_{\max}+\frac{\theta^{k+1}_{\min}-\theta^{k}_{\max}}{2}\\
     &\mf{h}_{\alpha}^{\top}\mf{x}\hi^{(n)}+e[n] > \theta^{k}_{\min}-\frac{\Delta^c_{\min}(\alpha,\tD)}{2}> \theta^{k}_{\min}-\frac{\theta^{k}_{\min}-\theta^{k-1}_{\max}}{2}
\end{align*}
This shows that { whenever $\alpha \in \mathcal{F}_{\tD}$, the condition $\vert e[n]\vert <\Delta_{\min}^{\text{c}}(\alpha,\tD)/2$ is sufficient for \eqref{eqn:band_count} to hold  $\forall \ \gamma[n]$} and hence the spike count can be exactly recovered. 
}   
\end{proof}
\vspace{-0.5cm}
\section*{\small{Appendix F: Amplitude Estimation}}
We suggest a procedure to estimate the binary amplitude $A$, if it is unknown. We first evaluate the signal $c[n]$ from different time instants $n=1,2,\cdots,M-1$. For some $1\leq n_0 \leq M-1$, we estimate a set $\mathcal{A}=\{A_k\}$ of candidate amplitudes:
$
    \smash{A_k: = c[n_0]/\mathbf{h}_{\alpha}^T \mathbf{v}_k \text{ where } \mathbf{v}_k\in \mathcal{S}_{\text{all}}}.
$
Only a certain amplitudes can generate $c[n_0]$ from a valid binary spiking pattern $\mathbf{v}_k \in \mathcal{S}_{\text{all}}$. Our goal is to prune $\mathcal{A}$ by sequentially eliminating  certain candidate amplitudes from the set based on a consistency test across the remaining measurements $c[n]$. At the $t^{\text{th}}$ stage ($t=2,3,\cdots$), for every remaining candidate amplitude $A_k \in \mathcal{A}$, we perform the following consistency test with $c[n]$, to identify if a candidate amplitude can potentially generate the corresponding measurement $c[n]$. Suppose there exists a spiking pattern $\mathbf{v}_l \in \mathcal{S}_{\text{all}}$ such that 
\begin{equation}
    \smash{c[n] = A_k \mathbf{h}_{\alpha}^T\mathbf{v}_l} \label{eqn:amp_check1}
\end{equation}
then $A_k$ remains a valid candidate. If we cannot find a corresponding $\mathbf{v}_{l} \in \mathcal{S}_{\text{all}}$ for an amplitude $A_k$, we remove it, $\mathcal{A}= \mathcal{A}\setminus A_k$. In presence of noise, \eqref{eqn:amp_check1} can be modified to allow a tolerance $\gamma$ as we may not find an exact match. The tolerance $\gamma$ is chosen to be $0.5$ in the experiments on the GENIE dataset. This procedure prunes out possible values for the amplitude by leveraging the shared amplitude across multiple measurements $c[n]$. 
\end{appendix}
\vspace{-0.3cm}
\section{\small{Acknowledgement}}
\vspace{-0.1cm}
The authors would like to thank Prof. Nikita Sidorov, Department of Mathematics at the University of Manchester, for helpful discussions regarding 
computational challenges in finding finite $\beta$-expansion in the range $\beta\in (1,2)$. 
This work was supported by Grants ONR N00014-19-1-2256,
DE-SC0022165, NSF 2124929, and NSF CAREER ECCS 1700506.
\vspace{-0.3cm}
\bibliographystyle{IEEEtran}
\bibliography{ref}

\begin{thebibliography}{10}
\providecommand{\url}[1]{#1}
\csname url@samestyle\endcsname
\providecommand{\newblock}{\relax}
\providecommand{\bibinfo}[2]{#2}
\providecommand{\BIBentrySTDinterwordspacing}{\spaceskip=0pt\relax}
\providecommand{\BIBentryALTinterwordstretchfactor}{4}
\providecommand{\BIBentryALTinterwordspacing}{\spaceskip=\fontdimen2\font plus
\BIBentryALTinterwordstretchfactor\fontdimen3\font minus
  \fontdimen4\font\relax}
\providecommand{\BIBforeignlanguage}[2]{{%
\expandafter\ifx\csname l@#1\endcsname\relax
\typeout{** WARNING: IEEEtran.bst: No hyphenation pattern has been}%
\typeout{** loaded for the language `#1'. Using the pattern for}%
\typeout{** the default language instead.}%
\else
\language=\csname l@#1\endcsname
\fi
#2}}
\providecommand{\BIBdecl}{\relax}
\BIBdecl

\bibitem{small2014fluorophore}
A.~Small and S.~Stahlheber, ``Fluorophore localization algorithms for
  super-resolution microscopy,'' \emph{Nature methods}, vol.~11, no.~3, pp.
  267--279, 2014.

\bibitem{brette2012handbook}
R.~Brette and A.~Destexhe, \emph{Handbook of neural activity
  measurement}.\hskip 1em plus 0.5em minus 0.4em\relax Cambridge University
  Press, 2012.

\bibitem{vogelstein2009spike}
J.~T. Vogelstein, B.~O. Watson, A.~M. Packer, R.~Yuste, B.~Jedynak, and
  L.~Paninski, ``Spike inference from calcium imaging using sequential monte
  carlo methods,'' \emph{Biophysical journal}, vol.~97, no.~2, pp. 636--655,
  2009.

\bibitem{deneux2016accurate}
T.~Deneux, A.~Kaszas, G.~Szalay, G.~Katona, T.~Lakner, A.~Grinvald,
  B.~R{\'o}zsa, and I.~Vanzetta, ``Accurate spike estimation from noisy calcium
  signals for ultrafast three-dimensional imaging of large neuronal populations
  in vivo,'' \emph{Nature communications}, vol.~7, p. 12190, 2016.

\bibitem{yang2015fifty}
S.~Yang and L.~Hanzo, ``Fifty years of mimo detection: The road to large-scale
  mimos,'' \emph{IEEE Communications Surveys \& Tutorials}, vol.~17, no.~4, pp.
  1941--1988, 2015.

\bibitem{donoho1992superresolution}
D.~L. Donoho, ``Superresolution via sparsity constraints,'' \emph{SIAM journal
  on mathematical analysis}, vol.~23, no.~5, pp. 1309--1331, 1992.

\bibitem{candes2014towards}
E.~J. Cand{\`e}s and C.~Fernandez-Granda, ``Towards a mathematical theory of
  super-resolution,'' \emph{Communications on pure and applied Mathematics},
  vol.~67, no.~6, pp. 906--956, 2014.

\bibitem{li2020super}
W.~Li, W.~Liao, and A.~Fannjiang, ``Super-resolution limit of the esprit
  algorithm,'' \emph{IEEE Transactions on Information Theory}, vol.~66, no.~7,
  pp. 4593--4608, 2020.

\bibitem{batenkov2021super}
D.~Batenkov, G.~Goldman, and Y.~Yomdin, ``Super-resolution of near-colliding
  point sources,'' \emph{Information and Inference: A Journal of the IMA},
  vol.~10, no.~2, pp. 515--572, 2021.

\bibitem{schiebinger2017superresolution}
G.~Schiebinger, E.~Robeva, and B.~Recht, ``Superresolution without
  separation,'' \emph{Information and Inference: A Journal of the IMA}, vol.~7,
  no.~1, pp. 1--30, 2017.

\bibitem{bendory2017robust}
T.~Bendory, ``Robust recovery of positive stream of pulses,'' \emph{IEEE
  Transactions on Signal Processing}, vol.~65, no.~8, pp. 2114--2122, 2017.

\bibitem{liao2016music}
W.~Liao and A.~Fannjiang, ``Music for single-snapshot spectral estimation:
  Stability and super-resolution,'' \emph{Applied and Computational Harmonic
  Analysis}, vol.~40, no.~1, pp. 33--67, 2016.

\bibitem{qiao2019guaranteed}
H.~Qiao and P.~Pal, ``Guaranteed localization of more sources than sensors with
  finite snapshots in multiple measurement vector models using difference
  co-arrays,'' \emph{IEEE Transactions on Signal Processing}, vol.~67, no.~22,
  pp. 5715--5729, 2019.

\bibitem{qiao2019non}
------, ``A non-convex approach to non-negative super-resolution: Theory and
  algorithm,'' in \emph{ICASSP 2019-2019 IEEE International Conference on
  Acoustics, Speech and Signal Processing (ICASSP)}.\hskip 1em plus 0.5em minus
  0.4em\relax IEEE, 2019, pp. 4220--4224.

\bibitem{qiao2020super}
H.~Qiao, S.~Shahsavari, and P.~Pal, ``Super-resolution with noisy measurements:
  Reconciling upper and lower bounds,'' in \emph{ICASSP 2020-2020 IEEE
  International Conference on Acoustics, Speech and Signal Processing
  (ICASSP)}.\hskip 1em plus 0.5em minus 0.4em\relax IEEE, 2020, pp. 9304--9308.

\bibitem{shahsavari2021fundamental}
S.~Shahsavari, J.~Millhiser, and P.~Pal, ``Fundamental trade-offs in noisy
  super-resolution with synthetic apertures,'' in \emph{ICASSP 2021-2021 IEEE
  International Conference on Acoustics, Speech and Signal Processing
  (ICASSP)}.\hskip 1em plus 0.5em minus 0.4em\relax IEEE, 2021, pp. 4620--4624.

\bibitem{qiao2018modulus}
H.~Qiao and P.~Pal, ``On the modulus of continuity for noisy positive
  super-resolution,'' in \emph{2018 IEEE International Conference on Acoustics,
  Speech and Signal Processing (ICASSP)}.\hskip 1em plus 0.5em minus
  0.4em\relax IEEE, 2018, pp. 3454--3458.

\bibitem{chi2020harnessing}
Y.~Chi and M.~F. Da~Costa, ``Harnessing sparsity over the continuum: Atomic
  norm minimization for superresolution,'' \emph{IEEE Signal Processing
  Magazine}, vol.~37, no.~2, pp. 39--57, 2020.

\bibitem{bhaskar2013atomic}
B.~N. Bhaskar, G.~Tang, and B.~Recht, ``Atomic norm denoising with applications
  to line spectral estimation,'' \emph{IEEE Transactions on Signal Processing},
  vol.~61, no.~23, pp. 5987--5999, 2013.

\bibitem{grewe2010high}
B.~F. Grewe, D.~Langer, H.~Kasper, B.~M. Kampa, and F.~Helmchen, ``High-speed
  in vivo calcium imaging reveals neuronal network activity with
  near-millisecond precision,'' \emph{Nature methods}, vol.~7, no.~5, p. 399,
  2010.

\bibitem{pnevmatikakis2016simultaneous}
E.~A. Pnevmatikakis, D.~Soudry, Y.~Gao, T.~A. Machado, J.~Merel, D.~Pfau,
  T.~Reardon, Y.~Mu, C.~Lacefield, W.~Yang \emph{et~al.}, ``Simultaneous
  denoising, deconvolution, and demixing of calcium imaging data,''
  \emph{Neuron}, vol.~89, no.~2, pp. 285--299, 2016.

\bibitem{schmidt1986multiple}
R.~Schmidt, ``Multiple emitter location and signal parameter estimation,''
  \emph{IEEE transactions on antennas and propagation}, vol.~34, no.~3, pp.
  276--280, 1986.

\bibitem{roy1989esprit}
R.~Roy and T.~Kailath, ``Esprit-estimation of signal parameters via rotational
  invariance techniques,'' \emph{IEEE Transactions on acoustics, speech, and
  signal processing}, vol.~37, no.~7, pp. 984--995, 1989.

\bibitem{hua1990matrix}
Y.~Hua and T.~K. Sarkar, ``Matrix pencil method for estimating parameters of
  exponentially damped/undamped sinusoids in noise,'' \emph{IEEE Transactions
  on Acoustics, Speech, and Signal Processing}, vol.~38, no.~5, pp. 814--824,
  1990.

\bibitem{bernstein2019deconvolution}
B.~Bernstein and C.~Fernandez-Granda, ``Deconvolution of point sources: a
  sampling theorem and robustness guarantees,'' \emph{Communications on Pure
  and Applied Mathematics}, vol.~72, no.~6, pp. 1152--1230, 2019.

\bibitem{koulouri2020adaptive}
A.~Koulouri, P.~Heins, and M.~Burger, ``Adaptive superresolution in
  deconvolution of sparse peaks,'' \emph{IEEE Transactions on Signal
  Processing}, vol.~69, pp. 165--178, 2020.

\bibitem{morgenshtern2016super}
V.~I. Morgenshtern and E.~J. Candes, ``Super-resolution of positive sources:
  The discrete setup,'' \emph{SIAM Journal on Imaging Sciences}, vol.~9, no.~1,
  pp. 412--444, 2016.

\bibitem{batenkov2019rethinking}
D.~Batenkov, A.~Bhandari, and T.~Blu, ``Rethinking super-resolution: the
  bandwidth selection problem,'' in \emph{ICASSP 2019-2019 IEEE International
  Conference on Acoustics, Speech and Signal Processing (ICASSP)}.\hskip 1em
  plus 0.5em minus 0.4em\relax IEEE, 2019, pp. 5087--5091.

\bibitem{da2018tight}
M.~F. Da~Costa and W.~Dai, ``A tight converse to the spectral resolution limit
  via convex programming,'' in \emph{2018 IEEE International Symposium on
  Information Theory (ISIT)}.\hskip 1em plus 0.5em minus 0.4em\relax IEEE,
  2018, pp. 901--905.

\bibitem{blu2008sparse}
T.~Blu, P.-L. Dragotti, M.~Vetterli, P.~Marziliano, and L.~Coulot, ``Sparse
  sampling of signal innovations,'' \emph{IEEE Signal Processing Magazine},
  vol.~25, no.~2, pp. 31--40, 2008.

\bibitem{uriguen2013fri}
J.~A. Urig{\"u}en, T.~Blu, and P.~L. Dragotti, ``Fri sampling with arbitrary
  kernels,'' \emph{IEEE Transactions on Signal Processing}, vol.~61, no.~21,
  pp. 5310--5323, 2013.

\bibitem{onativia2013finite}
J.~Onativia, S.~R. Schultz, and P.~L. Dragotti, ``A finite rate of innovation
  algorithm for fast and accurate spike detection from two-photon calcium
  imaging,'' \emph{Journal of neural engineering}, vol.~10, no.~4, p. 046017,
  2013.

\bibitem{tur2011innovation}
R.~Tur, Y.~C. Eldar, and Z.~Friedman, ``Innovation rate sampling of pulse
  streams with application to ultrasound imaging,'' \emph{IEEE Transactions on
  Signal Processing}, vol.~59, no.~4, pp. 1827--1842, 2011.

\bibitem{rudresh2017finite}
S.~Rudresh and C.~S. Seelamantula, ``Finite-rate-of-innovation-sampling-based
  super-resolution radar imaging,'' \emph{IEEE Transactions on Signal
  Processing}, vol.~65, no.~19, pp. 5021--5033, 2017.

\bibitem{stojnic2010recovery}
M.~Stojnic, ``Recovery thresholds for $l_1$ optimization in binary compressed
  sensing,'' in \emph{2010 IEEE International Symposium on Information
  Theory}.\hskip 1em plus 0.5em minus 0.4em\relax IEEE, 2010, pp. 1593--1597.

\bibitem{keiper2017compressed}
S.~Keiper, G.~Kutyniok, D.~G. Lee, and G.~E. Pfander, ``Compressed sensing for
  finite-valued signals,'' \emph{Linear Algebra and its Applications}, vol.
  532, pp. 570--613, 2017.

\bibitem{flinth2019recovery}
A.~Flinth and S.~Keiper, ``Recovery of binary sparse signals with biased
  measurement matrices,'' \emph{IEEE Transactions on Information Theory},
  vol.~65, no.~12, pp. 8084--8094, 2019.

\bibitem{fosson2019recovery}
S.~M. Fosson and M.~Abuabiah, ``Recovery of binary sparse signals from
  compressed linear measurements via polynomial optimization,'' \emph{IEEE
  Signal Processing Letters}, vol.~26, no.~7, pp. 1070--1074, 2019.

\bibitem{tian2009detection}
Z.~Tian, G.~Leus, and V.~Lottici, ``Detection of sparse signals under
  finite-alphabet constraints,'' in \emph{2009 IEEE International Conference on
  Acoustics, Speech and Signal Processing}.\hskip 1em plus 0.5em minus
  0.4em\relax IEEE, 2009, pp. 2349--2352.

\bibitem{sarangi2021no}
P.~Sarangi and P.~Pal, ``No relaxation: Guaranteed recovery of finite-valued
  signals from undersampled measurements,'' in \emph{ICASSP 2021-2021 IEEE
  International Conference on Acoustics, Speech and Signal Processing
  (ICASSP)}.\hskip 1em plus 0.5em minus 0.4em\relax IEEE, 2021, pp. 5440--5444.

\bibitem{sarangi2022bin}
------, ``Measurement matrix design for sample-efficient binary compressed
  sensing,'' \emph{IEEE Signal Processing Letters}, 2022.

\bibitem{razavikia2019reconstruction}
S.~Razavikia, A.~Amini, and S.~Daei, ``Reconstruction of binary shapes from
  blurred images via hankel-structured low-rank matrix recovery,'' \emph{IEEE
  Transactions on Image Processing}, vol.~29, pp. 2452--2462, 2019.

\bibitem{friedrich2017fast}
J.~Friedrich, P.~Zhou, and L.~Paninski, ``Fast online deconvolution of calcium
  imaging data,'' \emph{PLoS computational biology}, vol.~13, no.~3, p.
  e1005423, 2017.

\bibitem{jewell2020fast}
S.~W. Jewell, T.~D. Hocking, P.~Fearnhead, and D.~M. Witten, ``Fast nonconvex
  deconvolution of calcium imaging data,'' \emph{Biostatistics}, vol.~21,
  no.~4, pp. 709--726, 2020.

\bibitem{sarangi2020effect}
P.~Sarangi, M.~C. H{\"u}c{\"u}meno{\u{g}}lu, and P.~Pal, ``Effect of
  undersampling on non-negative blind deconvolution with autoregressive
  filters,'' in \emph{ICASSP 2020-2020 IEEE International Conference on
  Acoustics, Speech and Signal Processing (ICASSP)}.\hskip 1em plus 0.5em minus
  0.4em\relax IEEE, 2020, pp. 5725--5729.

\bibitem{rupasinghe2020robust}
A.~Rupasinghe and B.~Babadi, ``Robust inference of neuronal correlations from
  blurred and noisy spiking observations,'' in \emph{2020 54th Annual
  Conference on Information Sciences and Systems (CISS)}.\hskip 1em plus 0.5em
  minus 0.4em\relax IEEE, 2020, pp. 1--5.

\bibitem{sidorov2003almost}
N.~Sidorov, ``Almost every number has a continuum of $\beta$-expansions,''
  \emph{The American Mathematical Monthly}, vol. 110, no.~9, pp. 838--842,
  2003.

\bibitem{glendinning2001unique}
P.~Glendinning and N.~Sidorov, ``Unique representations of real numbers in
  non-integer bases,'' \emph{Mathematical Research Letters}, vol.~8, no.~4, pp.
  535--543, 2001.

\bibitem{renyi1957representations}
A.~R{\'e}nyi, ``Representations for real numbers and their ergodic
  properties,'' \emph{Acta Mathematica Academiae Scientiarum Hungarica},
  vol.~8, no. 3-4, pp. 477--493, 1957.

\bibitem{frougny1992finite}
C.~Frougny and B.~Solomyak, ``Finite beta-expansions,'' \emph{Ergodic Theory
  Dynam. Systems}, vol.~12, no.~4, pp. 713--723, 1992.

\bibitem{komornik2011expansions}
V.~Komornik and P.~Loreti, ``Expansions in noninteger bases.'' \emph{Integers},
  vol.~11, no.~A9, p.~30, 2011.

\bibitem{hubel1959receptive}
D.~H. Hubel and T.~N. Wiesel, ``Receptive fields of single neurones in the
  cat's striate cortex,'' \emph{The Journal of physiology}, vol. 148, no.~3, p.
  574, 1959.

\bibitem{chen2013ultrasensitive}
T.-W. Chen, T.~J. Wardill, Y.~Sun, S.~R. Pulver, S.~L. Renninger, A.~Baohan,
  E.~R. Schreiter, R.~A. Kerr, M.~B. Orger, V.~Jayaraman \emph{et~al.},
  ``Ultrasensitive fluorescent proteins for imaging neuronal activity,''
  \emph{Nature}, vol. 499, no. 7458, pp. 295--300, 2013.

\bibitem{genie2015simultaneous}
H.~K. S.~c. GENIE~Project, Janelia Farm~Campus, ``Simultaneous imaging and
  loose-seal cell-attached electrical recordings from neurons expressing a
  variety of genetically encoded calcium indicators,'' \emph{CRCNS. org}, 2015.

\end{thebibliography}
\end{document}